\documentclass[14pt]{article}
\usepackage[margin=1in]{geometry}
\usepackage{amsmath,amssymb,amsthm}
\usepackage{natbib}
\usepackage{bm}
\usepackage{enumerate}
\usepackage{graphicx}
\usepackage{color}
\usepackage{mathrsfs}
\usepackage{listings}
\usepackage{multirow}
\usepackage{caption}
\usepackage{subcaption}
\usepackage{natbib}
\usepackage{bbm}
\usepackage{makecell}
\usepackage{url}
\usepackage{setspace}
\usepackage{float}
\usepackage{longtable}
\usepackage[ruled,vlined]{algorithm2e}

\newcommand{\mr}{\mathrm}

\DeclareMathOperator*{\argmin}{argmin} 
\DeclareMathOperator*{\btr}{\bigtriangleup}

\theoremstyle{plain}
\newtheorem{theorem}{Theorem}
\newtheorem{lemma}{Lemma}
\newtheorem{tech lemma}{Technical Lemma}

\newtheorem{proposition}{Proposition}
\newtheorem{definition}{Definition}
\newtheorem{condition}{Condition}

\theoremstyle{remark}

\newcommand{\blind}{1}

\begin{document}

\def\spacingset#1{\renewcommand{\baselinestretch}%
{#1}\small\normalsize} \spacingset{1}

%%%%%%%%%%%%%%%%%%%%%%%%%%%%%%%%%%%%%%%%%%%%%%%%%%%%%%%%%%%%%%%%%%%%%%%%%%%%%%

\if1\blind
{
  \title{\bf Classification Trees for Imbalanced Data: Surface-to-Volume Regularization}
  \author{Yichen Zhu\thanks{Yichen Zhu is Ph.D. student, Department of Statistical Science, Duke University, Durham, NC 27708.}\hspace{.2cm}, 
  Cheng Li\thanks{Cheng Li is Assistant Professor, Department of Statistics and Applied Probability, National University of Singapore, Singapore, 117546. The research was supported by the Singapore Ministry of Education Academic Research Funds Tier 1 Grant R-155-000-223-114.}\hspace{.2cm}
  and
    David B. Dunson\thanks{David B. Dunson is the Arts and Sciences Distinguished Professor, Departments of Statistical Science \& Mathematics, Duke University, Durham, NC 27708. The research was partially supported by grant N000141712844 of the United States Office of Naval Research.}\hspace{.2cm}}
  \maketitle
} \fi

\if0\blind
{
  \bigskip
  \bigskip
  \bigskip
  \begin{center}
    {\bf Classification Trees for Imbalanced Data: Surface-to-Volume Regularization}
\end{center}
  \medskip
} \fi

\bigskip
\begin{abstract}
Classification algorithms face difficulties when one or more classes have limited training data.  We are particularly interested in classification trees, due to their interpretability and flexibility.  
When data are limited in one or more of the classes, the estimated decision boundaries are often irregularly shaped due to the limited sample size, leading to poor generalization error. We propose a novel approach that penalizes the Surface-to-Volume Ratio (SVR) of the decision set, obtaining a new class of SVR-Tree algorithms.  We develop a simple and computationally efficient implementation while proving estimation consistency for SVR-Tree and rate of convergence for an idealized empirical risk minimizer of SVR-Tree.  SVR-Tree is compared with multiple algorithms that are designed to deal with imbalance through real data applications.
\end{abstract}

\noindent%
\vfill

%\newpage
%\spacingset{1.5} % DON'T change the spacing!

\section{Introduction}
We are interested in the common setting in which one has a set of training data $\mathscr{D}_n = \{(X_i, Y_i)\}_{i=1}^n$, with $X_i \in \Omega \subset \mathbb{R}^d$ a vector of features and $Y_i \in \{0,\ldots,K-1\}$ a class label.  The goal is to estimate a classifier $f: \Omega \to \{0,\ldots,K-1\}$, which outputs a class label given an input feature vector.  This involves splitting the feature space $\Omega$ into subsets having different class labels, as illustrated in Figure 1. Classification trees are particularly popular for their combined flexibility and interpretability \citep{hu2019optimal,lin2020generalized}.  
% Competitors like deep neural networks (DNN; \citealt{schmidhuber2015deep}) often have higher accuracy in prediction, but are essentially uninterpretable black boxes. 

%The performance of tree-based methods can be heavily dependent on not only the training data size $n$, but also the relative numbers in each class, $n_j = \sum_{k=1}^n \bm{1}_{\{Y_i = j\}}$, for $j=0,\ldots,J-1$.  

\par{}Our focus is on the case in which 
$n_j = \sum_{k=1}^n \bm{1}_{\{Y_i = j\}}$ is small, for one or more $j \in \{0,\ldots,J-1\}$.
%some of the $n_j$s are small, so that there is limited data available in one or more of the classes.  
For simplicity in exposition, we assume $J=2$, so that there are only two classes. Without loss of generality, suppose that $j=0$ is the majority class, $j=1$ is the minority class, and call set $\{x\in\Omega: f(x)=1\}$ the decision set (of the minority class). Hence, $n_1$ is relatively small compared to $n_0$ by this convention.  

\par{} To illustrate the problems that can arise from this imbalance, we consider a toy example. Let the sample space of $X$ be $\Omega = [0,1]\times[0,1]\subset \mathbb{R}^2$. Let the conditional distribution of $X$ given $Y$ be $X|Y=1 \sim U([0,0.75]\times[0.25,0.75])$, $X|Y=0 \sim U(\Omega)$. We generate a training data set with $5$ minority samples and $200$ majority samples. The estimated decision sets of unpruned CART and optimally pruned CART \citep{breiman1984classification} are shown in Figure 1 (a) and (b), respectively. Minority samples are given weight 32 while majority samples are given weight 1 to naively address the imbalance. Both decision sets are inaccurate, with unpruned CART overfitting and pruned CART also poor. 

\par{}Imbalanced data problems have drawn substantial interest; see \cite{haixiang2017learning}, \cite{he2008learning},  \cite{krawczyk2016learning} and \cite{fernandez2017insight} for reviews. Some early work relied on random under- or over-sampling, which is essentially equivalent to modifying the weights and cannot address the key problem.  \cite{chawla2002smote} proposed SMOTE, which instead creates synthetic samples. For each minority class sample, they create synthetic samples along the line segments that join each minority class sample with its k nearest neighbors in the minority class. Building on this idea, many other synthetic sampling methods have been proposed, including ADASYN \citep{he2008adasyn}, Borderline-SMOTE \citep{han2005borderline}, SPIDER \citep{stefanowski2008selective}, safe-level-SMOTE \citep{bunkhumpornpat2009safe} and WGAN-Based sampling \citep{wang2019wgan}. 
%These methods employee various procedures to create synthetic samples or modify the original samples, with the objective of producing a classifier that does not overfit the training data. 
These synthetic sampling methods have been demonstrated to be relatively effective.

\par{}However, current understanding of synthetic sampling is inadequate. \cite{chawla2002smote} motivates SMOTE as designed to ``create large and less specific decision regions'', ``rather than smaller and more specific regions''. Later papers fail to improve upon this heuristic justification. Practically, the advantage of synthetic sampling versus random over-sampling diminishes as the dimension of the feature space increases. 
%Suppose the feature space $\Omega\subset \mathbb{R}^d$ does not have any intrinsic lower-dimensional structure. Then, 
In general, for each minority sample, we require at least $d$ synthetic samples to fully describe its neighborhood. This 
is often infeasible due to the sample size of the majority class and to computational complexity.  Hence, it is typical to fix the number of synthetic samples regardless of the dimension of the feature space \citep{chawla2002smote}, which may fail to ``create large and less specific decision regions'' when the dimension is high.

Motivated by these issues, we propose to directly penalize the Surface-to-Volume Ratio (SVR) of the decision set for the construction of decision tree classifiers, based on the following considerations. First, a primary issue with imbalanced data is that the decision set without any regularization typically consists of many small neighborhoods around each minority class sample. As shown in Section 2 and 3, by penalizing the SVR of the decision set, we favor regularly shaped decision sets much less subject to such over-fitting. Second, existing classification methods are usually  justified with theoretical 
properties under strong assumptions on the local properties of the true decision sets, for example, requiring the true decision set to be convex or have sufficiently smooth boundaries. Such assumptions may be overly restrictive for many real datasets.
In contrast, the proposed SVR regularized trees are only subject to a simple global constraint on surface-to-volume ratio, which effectively controls the regularity of decision sets leading to strong theoretical guarantees (see Section 3). Third, we illustrate that SVR regularization can be efficiently implemented for tree classifiers by proposing an algorithm with similar computational complexity to standard tree algorithms such as CART.

%\par{}Motivated by these issues, we propose to directly penalize the Surface-to-Volume Ratio (SVR) of the decision set. A primary issue with imbalanced data is estimating a decision set consisting of small neighborhoods around each minority class sample.  By penalizing SVR we favor regularly shaped decision sets much less subject to such over-fitting. 
%With this motivation, we propose a new class of SVR-Tree algorithms.  

\par{} The rest of the paper is organized as follows. Section 2 describes our methodology and algorithmic implementation. Section 3 analyzes theoretical properties of SVR Trees, including consistency and rate of convergence for excess risk. Section 4 presents numerical studies for real datasets.  Section 5 contains a discussion, and proofs are included in an Appendix.

\begin{figure}[h]
     \centering
     \begin{subfigure}[b]{0.32\textwidth}
         \centering
         \includegraphics[width=\textwidth]{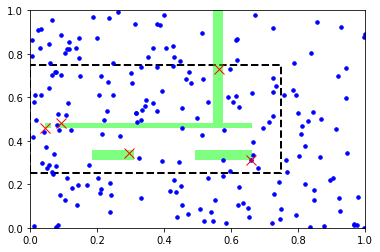}
         \caption{Unpruned CART.}
         \label{}
     \end{subfigure}
     \begin{subfigure}[b]{0.32\textwidth}
         \centering
         \includegraphics[width=\textwidth]{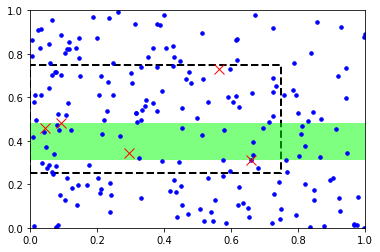}
         \caption{Optimally pruned CART.}
         \label{}
	\end{subfigure}
	\begin{subfigure}[b]{0.32\textwidth}
         \centering
         \includegraphics[width=\textwidth]{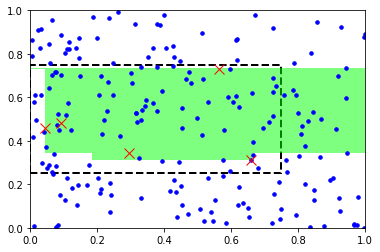}
         \caption{SVR-Tree.}
         \label{}
     \end{subfigure}
        \caption{Decision sets for different methods. Red crosses denote minority class samples, while blue points denote majority class samples. Rectangles with dashed frames denote the support of minority class samples, while rectangles filled with green color denote the minority class decision set.}
    \label{toy figure}
\end{figure}

\section{Methodology}
We first introduce the definition of surface-to-volume ratio (SVR) and tree impurity, and then define SVR-Tree as the minimizer of a weighted average of tree impurity and SVR. We then state the algorithm to estimate this tree from training data $\mathscr{D}_n = \{(X_i, Y_i)\}_{i=1}^n$. We assume readers have basic knowledge of tree-based classifiers like CART \citep{breiman1984classification} and C4.5 \citep{quinlan2014c4}. In the rest of the paper, the word ``tree'' refers specifically to classification trees that specify a class label associated with each leaf node. 

\subsection{Notation}\label{notations}
Training data are denoted as $\mathscr{D}_n = \{(X_i, Y_i)\}_{i=1}^n$, with $X_i \in \Omega \subset \mathbb{R}^d$ a vector of features and $Y_i \in \{0,1\}$ a class label. Uppercase letters $(X, Y)$ denote random variables, while lowercase $x, y$ denote specific values. Denote the $j$th feature of $X$ as $X[j]$. 
Let $\mathbb{P}^*$ denote the true distribution of i.i.d. samples $(X_i, Y_i)$ and let $\mathbb{P}$ denote the weighted distribution that up weights the minority class samples by a constant $\alpha\ge 1$. That is, for any measurable subset $A\subset \Omega$,
$$\mathbb{P}(A\times\{1\}) = \frac{\alpha\mathbb{P}^*(A\times\{1\})}{\mathbb{P}^*(\Omega\times\{0\})+\alpha\mathbb{P}^*(\Omega\times\{1\})}, \;\; \;\mathbb{P}(A\times\{0\}) = \frac{\mathbb{P}^*(A\times\{0\})}{\mathbb{P}^*(\Omega\times\{0\})+\alpha\mathbb{P}^*(\Omega\times\{1\})}.$$
Up weighting minority class samples is a conventional technique to counter imbalance and we will focus on the weighted measure $\mathbb{P}$ in the rest of the paper. Denote $\mathbb{P}_n$ as the weighted empirical distribution which assigns mass $w_0\alpha/n$ to minority class training samples and $w_0/n$ to majority class training samples. The constant $w_0$ is set to ensure the measures of all training samples add up to $1$.  \par{}
We use $f$ with various subscripts and diacritics to denote a general classifier from $\Omega$ to $\{0,1\}$, while $T$ with various subscripts and diacritics denotes a tree classifier from $\Omega$ to $\{0,1\}$. A tree classifier is formed by finite splitting steps and its decision sets are a finite union of hyperrectangles.
When discussing the probability of certain events that include $n$ random variables $\{(X_i, Y_i)\}_{i=1}^n$, we simply use $\mathbb{P}$ to represent the probability measure in the $n$-product space. For example, 
for any $a>0$, $\mathbb{P}\left(n^{-1}\sum_{i=1}^n X_i >a\right)$
denotes the probability of the event $\{\{(X_i, Y_i)\}_{i=1}^n: n^{-1}\sum_{i=1}^n X_i > a\}$. 
%\par{}
We use $\mathbb{E}$ for expectations over  $\mathbb{P}$ and $\mathbb{E}_n$ for expectations over $\mathbb{P}_n$. We use $\eta$ and $\eta_n$ to denote the conditional expectations of $Y$ given $X$ under the true and empirical measure, respectively. That is, $\eta(X) = \mathbb{E}(Y=1|X)$ and $\eta_n(X) = \mathbb{E}_n(Y=1|X)$. 

% We use $\eta$ and $\eta_j, \;1\le j\le d$ for conditional expectation of $Y$ given $X$ and $X[j]$, respectively. That is, $\eta(X) = \mathbb{P}(Y=1|X)$ and $\eta_j(X) = \mathbb{P}(Y=1|X[j])$. Their empirical counterparts are defined as $\eta_n(X)=\mathbb{P}_n(Y=1|X)$ and $\eta_{j,n}(X)=\mathbb{P}_n(Y=1|X[j])$.
%Uppercase $A$ is used to denote a leaf node of a tree, which is a hyperrectangle in $\Omega$. lowercase $p$, with various subscripts, are used to denote the conditional expectation $\mathbb{E}(Y|X)$ in various settings.

\subsection{Surface-to-Volume Ratio}
For all $d \in \mathbb{N}$, consider the space $\mathbb{R}^d$ equipped with Lebesgue measure $\mu$.
% define a $d$-dimensional measure space as $(\mathbb{R}^d, \mathscr{B}, \mu)$, where $\mathscr{B}$ is the collection of Borel sets, and $\mu$ is Lebesgue measure. 
For any Lebesgue measurable closed set $A\subset\Omega$, we define its volume as the Lebesgue measure of set $A$: $V(A) = \mu(A)$. If $A$ can be represented as a finite union of hyperrectangles, we define the  surface of $A$ as the $d-1$ dimensional Lebesgue measure of the boundary of $A$: $S(A) = \mu_{d-1}(\partial A)$; otherwise, if there exists $A_i, i\ge 1$ such that each $A_i$ is a finite union of hyperrectangles and $A_i$ converges to $A$ in Hausdorff distance, we define the surface of $A$ as $S(A) = \lim_{n\to\infty} S(A_i)$, provided the limit exists. In general, if all $A_i, i\ge 1$ are regularly shaped, it will be proven in Proposition \ref{SVR family surface} that the limit exists. Our definition of surface is different from the convention of the word ``surface''. For example, in our definition a $d$ dimensional ball with radius $1/2$ has the same surface as a $d$ dimensional cube with side length $1$, as approximating such a ball with hyperrectangles results in the same surface area as the cube. We adopt this definition of surface since we focus on tree classifiers whose decision sets are always  finite unions of hyperrectangles.
% Let $C_r$ denotes a hypercube centering at origin of $\mathbb{R}^d$ with side length $r$. The surface is defined by the following equation:
% $$S(A) = \lim_{r\to 0} \frac{V(A\backslash (\partial A \oplus C_r))}{r},$$
% where $\oplus$ denotes Minkowski addition. The above limit may not always exist. However, if $A$ is finite unions of hyperrectangles, the limit exists, is finite and equals to $d-1$ dimensional Lebesgue measure of $\partial A$.

For any set $A$ with $0<\mu(A)<\infty$, the surface-to-volume ratio (SVR) can be obtained as 
$r(A) = \frac{S(A)}{V(A)}.$
For sets with the same volume, the $d$-dimensional cube has the smallest SVR, while sets having multiple disconnected subsets and/or irregular boundaries have relatively high SVR. We visualize sets with different SVR in two-dimensional space in Figure \ref{SVR figure}, with the subfigures (a)-(c) corresponding to the decision sets in Figure \ref{toy figure}. We can see that decision sets with larger SVR tend to have more branches, disconnected components, or irregularly shaped borders than those with smaller SVR.

\begin{figure}[h]
     \centering
     \begin{subfigure}[b]{0.24\textwidth}
         \centering
         \includegraphics[width=\textwidth]{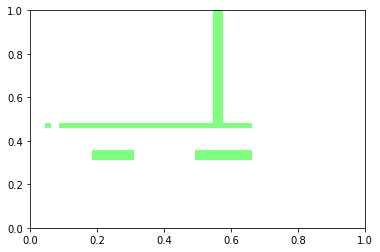}
         \caption{SVR=70.03}
         \label{}
     \end{subfigure}
     \begin{subfigure}[b]{0.24\textwidth}
         \centering
         \includegraphics[width=\textwidth]{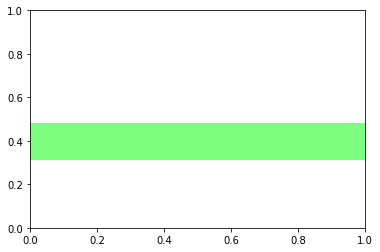}
         \caption{SVR=13.59}
         \label{}
	\end{subfigure}
	\begin{subfigure}[b]{0.24\textwidth}
         \centering
         \includegraphics[width=\textwidth]{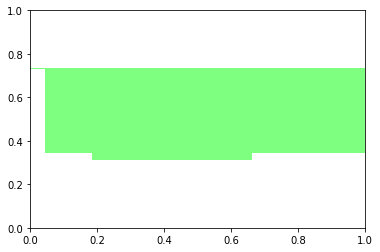}
         \caption{SVR=7.3407}
         \label{}
     \end{subfigure}
     \begin{subfigure}[b]{0.24\textwidth}
         \centering
         \includegraphics[width=\textwidth]{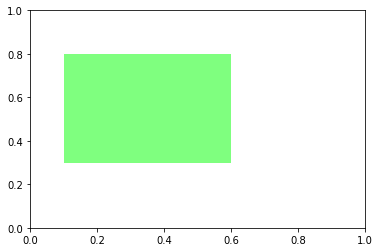}
         \caption{SVR=4}
         \label{}
     \end{subfigure}
     \begin{subfigure}[b]{0.24\textwidth}
         \centering
         \includegraphics[width=\textwidth]{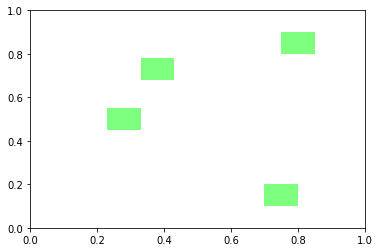}
         \caption{SVR=40}
         \label{}
     \end{subfigure}   
     \begin{subfigure}[b]{0.24\textwidth}
         \centering
         \includegraphics[width=\textwidth]{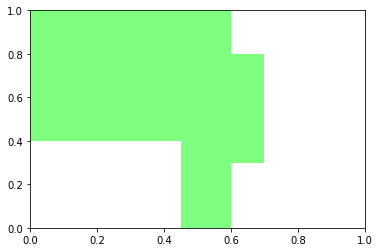}
         \caption{SVR=6.81}
         \label{}
     \end{subfigure}  
     \begin{subfigure}[b]{0.24\textwidth}
         \centering
         \includegraphics[width=\textwidth]{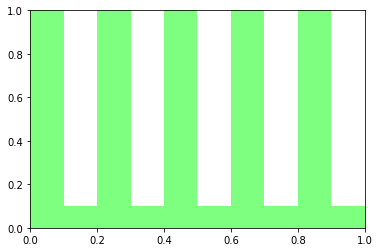}
         \caption{SVR=21.82}
         \label{}
     \end{subfigure}  
     \begin{subfigure}[b]{0.24\textwidth}
         \centering
         \includegraphics[width=\textwidth]{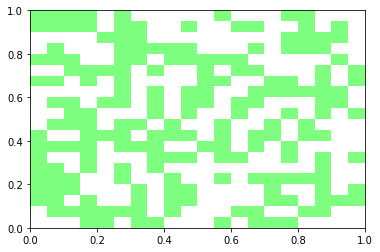}
         \caption{SVR=41.65}
         \label{}
     \end{subfigure}  
    \caption{Examples of sets with different SVR values in $[0,1]^2$. Rectangles filled with green color denote the minority class decision set.}
    \label{SVR figure}
\end{figure}

%In imbalanced data problems, naive methods that lack some form of regularizations will tend to produce an irregularly shaped decision set for the minority class, having small volume but large surface. For example, the decision set may be composed of multiple small sets around the minority class samples. By regularizing on surface-to-volume ratio, we can protect against such overfitting.

\paragraph{Surface-to-Volume Ratio of a Classification Tree}
For training data $\mathscr{D}_n$, we define a closed bounded sample space $\Omega\subset \mathbb{R}^d$ such that the support of $X$ is a subset of $\Omega$. A classification tree $T$ divides the sample space $\Omega \subset \mathbb{R}^d$ into two disjoint subsets $\Omega_1 = \{x\in\Omega: T(x)=1\}$ and $\Omega_0=\Omega \backslash \Omega_1$, where $\Omega_0 \cup \Omega_1= \Omega$. The tree $T$ predicts a new sample $X^*$ belongs to class 0 if $X^* \in \Omega_0$, and class 1 if $X^* \in \Omega_1$. The surface-to-volume ratio of a classification tree is defined as the surface-to-volume-ratio of set $\Omega_1$: 
$r(T) = r(\Omega_1).$

\subsection{Impurity Function and Tree Impurity}
A classification tree partitions the sample space into multiple leaf nodes, assigning one class label to each leaf node. The tree should be built to maximize homogeneity of the training sample class labels within nodes.  This motivates the following definition. %of impurity functions and tree impurity.
\begin{definition}[Impurity, Definition 2.5 of \citealt{breiman1984classification}]
An impurity function $\phi(\cdot, \cdot)$ is defined on the set of pairs $(p_0, p_1)$ satisfying $p_0\ge0, p_1\ge0$, $p_0+p_1=1$ with the properties (i) $\phi(\cdot, \cdot)$ achieves its maximum only at $(1/2, 1/2)$; (ii) $\phi(\cdot, \cdot)$ achieves its minimum only at $(1, 0), (0,1)$; (iii) 
$\phi(\cdot, \cdot)$ is symmetrical in $(p_0, p_1)$, i.e., $\phi(p_0, p_1) = \phi(p_1, p_0)$.
\end{definition}
Let $p_0$ and $p_1$ represent the probabilities of belonging to the majority and minority class, respectively, within some branch of the tree.  Ideally, splits of the tree are chosen so that, after splitting, $p_0$ and $p_1$ move closer to 0 or 1 and further from 1/2.  Many different tree building algorithms use impurity to measure the quality of a split; for example, CART uses Gini impurity and C4.5 uses entropy, which is effectively a type of impurity measure.  
% When data are imbalanced, it is important to modify the definition of impurity to account for the fact that $p_1$ is much smaller than $p_0$.  With this in mind, we propose the following weighted impurity function.
% \begin{definition}[Weighted Impurity]\label{weighted impurity}
% Letting $\phi(\cdot, \cdot)$ be an impurity function, a weighted impurity function with weight $\alpha$ for the minority class is defined as
% $$\varphi_\alpha(p_0, p_1) = \phi\Big(\frac{p_0}{p_0+\alpha p_1}, \frac{\alpha p_1}{p_0 + \alpha p_1}\Big).$$
% \end{definition}
In the remainder of the paper, the term `impurity function' refers to the Gini impurity. That is, $\phi(p_0,p_1) = 1 - p_0^2 - p_1^2$. 
% $$\varphi_\alpha(p_0, p_1) = 1 - \Big(\frac{p_0}{p_0+\alpha p_1}\Big)^2 - \Big(\frac{\alpha p_1}{p_0 + \alpha p_1}\Big)^2.$$

Let $A_1, A_2, \ldots, A_m$ be the leaf nodes of a classification tree $T$ and let $z_j\in\{0,1\}$ be the predictive class label for node $A_j$, for $j = 1,\ldots,m$. Let $\mathbb{P}$ be the weighted probability measure defined in section \ref{notations}. Then the impurity of leaf node $A_j$ is $I(A_j, \mathbb{P}) = \phi (\mathbb{P}(Y=0|X\in A_j), \mathbb{P}(Y=1|X\in A_j))$. The impurity of node $A_j$ measures the class homogeneity in $A_j$, but does not depend on the predictive class label $z_j$.  Let $\tilde{z}_j =\mathbbm{1}_{\{\mathbb{P}(Y=1|X\in A_j)\ge 1/2\}}$ denote the dominant class label in $A_j$ under weighted measure $\mathbb{P}$.  We define a signed impurity taking into account $z_j$ as
$$\tilde{I}(A_j, \mathbb{P}) = \mathbbm{1}_{\{z_j=\tilde{z}_j\}} I(A_j, \mathbb{P}) + \mathbbm{1}_{\{z_j\ne\tilde{z}_j\}} (1-I(A_j, \mathbb{P})).$$
Signed impurity serves as a proxy to classification accuracy. If the predictive class label $z_j$ matches the dominant class label $\tilde{z}_j$ in node $A_j$, the signed impurity of node $A_j$ is equal to the impurity of node $A_j$ and is no greater than $1/2$. Otherwise, an extra penalty of $1-2I(A_j,\mathbb{P})$ is applied. It is easy to see $1-2I(A_j,\mathbb{P}) = [1-2\mathbb{P}(Y=0|X\in A_j)]^2$. In other words, the signed impurity uses quadratic loss  instead of the absolute value loss function of classification accuracy.
Taking a weighted average of the signed impurities across the leaf nodes, one obtains the tree impurity and signed tree impurity.
\begin{definition}[Tree Impurity]\label{tree impurity}
Let $T$ be a tree and $A_1, A_2, \ldots A_m$ be all the leaf nodes of this tree. Denoting the weighted probability measure as $\mathbb{P}$, the tree impurity of $T$ is $I(T, \mathbb{P}) = \sum_{j=1}^m \mathbb{P}(X\in A_j) I(A_j, \mathbb{P}).$
\end{definition}

\begin{definition}[Signed Tree Impurity]\label{signed tree impurity}
Under the notation of Definition \ref{tree impurity}, the signed tree impurity is 
$\tilde{I}(T, \mathbb{P}) = \sum_{j=1}^m \mathbb{P}(X\in A_j) \tilde{I}(A_j, \mathbb{P}).$
%where the $\mathbb{P}_\alpha(X\in A_j)$ is the still weighted version of $\mathbb{P}$ defined in section 2.1.
\end{definition}
%The tree impurity is independent of the predictive class labels on leaf nodes, and is mainly used in feature selection steps in section 2.5. The signed tree impurity takes into account both the training error and class homogeneity in all nodes of the tree, and serves as the loss function in the next subsection. 

\subsection{SVR-Tree Classifiers}
The SVR-Tree classifier is the minimizer of the weighted average of signed tree impurity and surface-to-volume ratio. Letting  
 $\mathscr{T}$ be the collection of possible trees, the SVR-Tree classifier is defined as
\begin{equation}\label{Tree Impu}
\hat{T} = \argmin_{T\in\mathscr{T}}[\tilde{I}(T,\mathbb{P}_n) + \lambda_n r(T)],
\end{equation}
where $\lambda_n$ is a penalty. The unknown probability measure $\mathbb{P}$ is replaced with the empirical measure $\mathbb{P}_n$ that assigns mass $1/n$ to each training sample $(X_i, Y_i), 1\le i\le n$. Unfortunately, without restricting the space of trees $\mathscr{T}$,  optimization problem (\ref{Tree Impu}) is intractable.  In the following subsection, we introduce an iterative greedy search algorithm that limits the size of $\mathscr{T}$ in each step to efficiently obtain a tree having provably good performance.

\subsection{The SVR-Tree Algorithm}
The SVR-Tree Algorithm is designed to find a nearly optimal SVR-Tree classifier. SVR-Tree proceeds in a greedy manner. We begin with the root node. At each step, we  operate on one leaf node of the current tree, splitting it into two new leaf nodes by finding the solution of (\ref{Tree Impu}).  The node to split at each step is uniquely specified by a breadth-first searching order. After splitting, the tree is updated and the node to split in the next step will be specified. The process stops 
when further splitting of leaf nodes either does not improve the loss or a prespecified maximum number of leaf nodes is achieved.
%leads to too few training samples within a leaf.

\par{} We first describe how to split a current leaf node to improve the solution to 
(\ref{Tree Impu}). Suppose the current tree is $T$ and the node to split is $A$, 
with $n'$ training samples. For each feature $j$, sort all samples in $A$ by increasing order of the $j$th feature as $X[j]_{j_1}, X[j]_{j_2}, \ldots, X[j]_{j_{n'}}$. We only allow splits to occur at $(X[j]_{j_i}+ X[j]_{j_{i+1}})/2, \;1\le i\le n'-1,\; 1\le j\le d$, the midpoint of two adjacent values of each feature. The total number of different splits is no more than $(n'-1)d$. After each such split of $A$, we keep all other leaf nodes unchanged while allowing all $4$ different class label assignments at the two new daughter nodes of $A$.  The current set of trees $\mathscr{T}$ to choose from in optimizing (\ref{Tree Impu}) includes the initial $T$ and all the split trees described above.  The cardinality of $\mathscr{T}$ is no more than $1 + 4(n'-1)d$, a linear order of $n'$.  We compute the risk for all $T\in\mathscr{T}$ to find the minimizer.  If the initial $T$ is the risk minimizer, we do not make any split in this step and mark the node $A$ as ``complete''. Any node marked as complete will no longer be split in subsequent steps.

\par{} It remains to specify the `breadth-first' searching order determining which node to split in each step.  Let depth of a node be the number of edges to the root node, which has depth 0. Breadth-first algorithms explore all nodes at the current depth before moving on  (\citealt{cormen2009introduction}, chapter 22.2).  To keep track of changes in the tree, we use a queue\footnote{A queue is a dynamic set in which the elements are kept in order and the principal operations on the set are the insertion of elements to the \textbf{tail}, known as enqueue, and deletion of elements from the \textbf{head}, known as dequeue. See chapter 10.1 of \cite{cormen2009introduction} for details.}. We begin our algorithm with a queue where the only entity is the root node. At each step, we remove the node at the front terminal of the queue, and split at this node as described in the previous paragraph. If a split tree is accepted as the risk minimizer over the current set $\mathscr{F}$, we enqueue two new leaf nodes; otherwise, the unsplit tree is the risk minimizer over the current set $\mathscr{F}$, so we don't enqueue any new node. The nodes in the front of the queue have the lowest depth. Therefore, our algorithm naturally follows a breadth-first searching order. We preset the maximal number of leaf nodes as $\bar{a}_n$. The process is stopped when either the queue is empty, in which case all the leaf nodes are marked as complete, or the number of leaf nodes is $\bar{a}_n$. On many datasets, Algorithm 1 will stop before reaching the maximal number of leaf nodes $\bar{a}_n$.  The upper bound is necessary in certain circumstances, since our algorithm, like the majority of practical tree algorithms, is greedy and not guaranteed to find the global minimum. Provided we can solve the global minimum, the upperbound $\bar{a}_n$ is no longer required, as shown in Theorem 2 and discussed thereafter.

% \begin{figure}[h]
%      \centering
%      \includegraphics[width=10cm]{breadth-first-diagram.png}
%      \caption{Diagram for breadth-first tree building procedure.}
%      \label{}
% \end{figure}

\par{} Our SVR-Tree algorithm has a coarse to fine tree building style, tending to first split the sample space into larger pieces belonging to two different classes, followed by modifications to the surface of the decision set to decrease tree impurity and SVR. The steps are sketched in Algorithm 1, where feature selection steps are marked as optional. A more detailed and rigorous version is in Section 4.1 of the supplementary material.

The average time complexity of Algorithm 1 is $O(dn\log n)$, which is the same as many tree algorithms like CART and C4.5. A detailed analysis of computational complexity is available at Section 4.2 of the supplementary material.
% A diagram illustrating the breadth first searching order is shown in Figure 2.

\paragraph{Optional Step for Feature Selection}
It is likely that some features will not help to predict class labels. These features should be excluded from our estimated tree. Under some mild conditions, a split on a redundant feature has minimal impact on the tree impurity compared to a split in a non-redundant feature. Thus feature selection can be achieved by thresholding. Suppose we are splitting node $A$ into two new leaf nodes $A_1, A_2$. Then the (unsigned) \textbf{tree impurity decrease} after this split is defined as:
$$
\Delta I(T, \mathbb{P}_n) = \mathbb{P}_{n}(A) [I(A, \mathbb{P}_n) - I(A_1,\mathbb{P})\mathbb{P}_{n}(A_1|A) -  I(A_2,\mathbb{P})\mathbb{P}_{n}(A_2|A)]. $$
Let $J_0\subset \{1, 2, \ldots d\}$ be the indices of features that have been split in previous tree building steps. Given that we are splitting on node $A$, let $\Delta I_0$ be the maximal tree impurity decrease over all splits in feature $X[j], j\in J_0$. Then a split in a new feature $X[j'], j'\not\in J_0$, with tree impurity decrease $\Delta I(T, \mathbb{P}_n)$, is accepted if
\begin{equation}\label{impu decr}
\Delta I(T, \mathbb{P}_n) \ge \Delta I_0 + c_0\lambda_n,
\end{equation}
where $c_0$ is a constant independent of the training data. By equation (\ref{impu decr}), a split on a new feature is accepted if its tree impurity decrease is greater than $\Delta I_0$, the maximal tree impurity decrease over all previously split features, plus a threshold term $c_0\lambda_n$. 
%The threshold level $c_0\lambda_n$ is chosen to properly control the noise in the training data. 
%This choice of threshold limits the risk of false splitting of redundant features in two ways: 1. If there are enough samples in node $A$, with high probability, a split on a redundant feature will not decrease the impurity of this node too much and will not decrease tree impurity over the threshold; 2. If there are very few samples in node $A$, the probability $\mathbb{P}_n(A)$ will be small enough to get the split rejected. In addition,  
%As long as $\lambda_n \to 0$,  thresholding will not prevent splitting of informative features, because such splits can often generate a considerable impurity decrease, at least at the first few steps. As long as a split on a feature is already accepted, further splits on this feature do not need to satisfy (\ref{impu decr}), so smaller changes in the decision boundary can continue to be accepted. Rigorous theoretical analysis will be provided in section 5.
% DD - the above is babbling
Theoretical support for this thresholding approach is provided in the supplementary material.

\begin{algorithm}[H]
\SetAlgoLined
\KwResult{Output the fitted tree}
Input training data $\{(X_i, Y_i)\}_{i=1}^n$, impurity function $f(\cdot)$,  weight for minority class $\alpha$, SVR penalty parameter $\lambda_n$, and maximal number of leaf nodes $\bar{a}_n\in\mathbb{N}$. Let $\mr{node\_queue}$ be a queue of only root node

\For{$1\le j \le d$}{Sort all the samples in values of $j$th feature;}

\While{$\mr{node\_queue}$ is not empty and number of leaf nodes $\le \bar{a}_n$}{		
Dequeue the first entity in $\mr{node\_queue}$, denoting it as $\mr{node}$;

\For{all possible splits in $\mr{node}$}{
(optional) Check if the current split satisfies feature selection conditions; 
if not satisfied, reject the current split and continue;
Compute the signed tree impurity of the current split;
}

Find the split that minimizes the weighted average of signed impurity and SVR as in equation \eqref{Tree Impu};

\eIf{the weighted average of signed impurity and SVR is decreased}{
Accept the current split. Enqueue two child nodes of $\mr{node}$ into $\mr{node\_queue}$;}{
Reject the current split;
}
}
\caption{Outline of steps of SVR-Tree}
\end{algorithm}

% This shows the efficiency of Algorithm 1.
% \textcolor{blue}{(Cheng: Please add some comments on the comparison of this complexity with the complexity of other tree algorithms such as CART, C4.5., and SMOTE, i.e. those you are going to compare in the simulation studies. You can assume they are compared under the same parameters for the depth, leaf nodes, etc.)}

\section{Theoretical Properties}
In this section, we will discuss the theoretical properties of SVR-Tree from two aspects. First, we show the classifier obtained from Algorithm 1 is consistent in the sense that a generalized distance between SVR-Tree and the oracle classifier goes to zero as sample size goes to infinity. Second, we show for an idealized empirical risk minimizer over a family of SVR constrained classifiers, the excess risk converges to zero at the rate $O((\log n/n)^{\frac{1}{d+1-1/\kappa}})$ as sample size goes to infinity, where $\kappa \ge 1$ is defined in Condition \ref{margin cond}. The consistency result provides some basic theoretical guarantees for SVR-Tree while the rate of convergence result provides some further insights on how the surface-to-volume ratio works in nonparametric classification. We also developed results regarding feature selection consistency of Algorithm 1 with optional steps enabled. Due to limited space, these results are provided in the  supplementary material.
% In this section, we will discuss the consistency of our classification tree obtained from Algorithm 1 in two aspects: the first is estimation consistency -- as sample size goes to infinity, a generalized distance between SVR-Tree and the oracle classifier converges to zero; the other is feature selection consistency -- for redundant features that are conditionally independent of $Y$ given the other features, the probability of SVR-Tree excluding these features converges to one. 

\subsection{Estimation Consistency}
% We define a generalized metric on the space of all nonrandom classifiers, introduce classifier risk, and define our notions of 
% classifier consistency.
% \begin{definition}
% For any weight $\alpha>0$, any two classifiers $f_1, f_2: \Omega \to \{0,1\}$, the generalized metric between $f_1, f_2$ is defined as
% $$d(f_1, f_2) = \int |f_1 - f_2| d\mathbb{P}_X,$$
% \end{definition}
% We then define the risk of classifiers, the oracle classifiers and consistency based on this generalized metric.

% \begin{definition}
% For any weight $\alpha$ and classifier $f:\Omega \to \{0,1\}$, we define the risk as
% $$R(f) = \mathbb{E} \left[\frac{\alpha}{1+\alpha} \max(Y-f(X),0) + \frac{1}{1+\alpha}\max(f(X)-Y,0)\right],$$
% where the expectation is taken over probability measure $\mathbb{P}$ for 
% random variables $(X,Y)$.
% \end{definition}
Let the classification risk under measure $\mathbb{P}$ be $R(f) = \mathbb{E}|Y-f(X)|$ and its empirical version under measure $\mathbb{P}_n$ be $R_n(f) = \mathbb{E}_n|Y-f(X)|$.
Let $\mathscr{F}$ be the collection of all measurable functions from $\Omega$ to $\{0,1\}$. Then the oracle classifier and oracle risk can be defined as
$f^* = \arg\min_{f\in\mathscr{F}} R(f)$ and $R^* = R(f^*)$, respectively.
% Without loss of generality, we assume . 
% Then the oracle classifier is unique almost surely. ,\;a.s.
We define the consistency of a sequence of classifiers based on the excess risk comparing to oracle risk.
\begin{definition}
A sequence of classifiers $f_n$ is consistent if
$ \lim_{n\to\infty}  \mathbb{E} R(f_n) = R^*.$
\end{definition}
In the case when $\mathbb{P}(\{X\in\Omega:\eta(X)=1/2\})=0$, consistency of $f_n$ is equivalent to requiring $\lim_{n\to\infty} \mathbb{E}|f_n-f^*|=0$ almost surely. For any hyperrectangle $A\subset \Omega$, denote the component-wise conditional expectation of $A$ as $\eta_{A,j}(X[j]) = \mathbb{E}(Y|X[j],X\in A)$. To ensure the consistency of Algorithm 1, we need the following identifiability condition. 
\begin{condition}[Identifiability]\label{identifiable}
For any hyperrectangle $A\subset \Omega$, if the component-wise conditional expectation $\eta_{A,j}(X)$ is a constant on $A$, $\forall 1\le j\le d$, then the conditional expectation $\eta(X)$ is a constant on $A$. 
\end{condition}

Condition \ref{identifiable} is indispensable for proving the consistency of our SVR regularized tree classifiers without requiring the diameter of each leaf node to go to zero. This is because CART-type algorithms can only make splits based on the component-wise conditional expectation. If all the component-wise conditional expectations $\eta_{A,j}(X[j])$ are constants but the overall conditional expectation $\eta(x)$ varies in $x$, then the problem is not identifiable to CART-type split rules. Condition \ref{identifiable} can also be viewed as a relaxed version of Condition (H1) in \cite{scornet2015consistency} who studied the consistency of random forest regression.
Their Condition (H1) requires that the mean of $Y$ given $X$ follows an additive model in the regression tree setup. An analogous condition to their (H1) for our classification tree setup would require $\eta(x)$ to be an additive function, which would imply our Condition \ref{identifiable}. 

%The identifiable condition is a generalization of Condition (H1) of \cite{scornet2015consistency}. 

%expect specially designed $\eta(X)$ with certain symmetric or periodic properties, . \textcolor{blue}{(Cheng: What does the last sentence mean?)}

Denote the tree obtained from Algorithm 1 as $\hat{T}_n$. Theorem \ref{consistency} shows $\hat{T}_n$ is consistent under mild conditions.
\begin{theorem}[Estimation consistency]\label{consistency}
Let $\Omega=[0,1]^d$ and the weighted measure $\mathbb{P}$ have continuous density $\rho(x)$ with lower bound $\rho(x)\ge\rho_{\min} >0$ for all $x\in \Omega$. Assume that Condition \ref{identifiable} is satisfied, $\eta(x)$ is continuous everywhere in $\Omega$, and $\bar a_n$, $\lambda_n$ satisfy
$\lim_{n\to\infty}\bar{a}_n = \infty$, $\lim_{n\to\infty}\lambda_n = 0$, 
$\lim_{n\to\infty}\frac{\bar{a}_n d\log n}{n} = 0.$
Then the classifier $\hat{f}_n$ obtained from Algorithm 1 is consistent.
\end{theorem}

\subsection{Rate of Convergence for Empirical Risk Minimizer}

We now consider the empirical risk minimizer over a class of trees under the SVR constraint and derive the rate of convergence for this empirical risk minimizer. This empirical risk minimizer can be considered as an idealized classifier.  
%In general for tree algorithms, efficient algorithms to find the global risk minimizer among all classification trees are available, and greedy search is typically used. 
%Hence, the empirical risk minimizer provides a useful alternative in consider properties of greedy search algorithms.

We first introduce the family of SVR constrained classifiers. For ease of notation, for any classifier $f$, we denote its decision set by $\Omega_1(f) = \{x\in\Omega: f(x)=1\}$. 
Define $\mathscr{F}_0$ as the family of all possible classifiers generated from trees with finite depth:
$\mathscr{F}_0 = \left\{f:\Omega\to [0,1]: \Omega_1(f) \;\;\text{is a finite union of hyperrectangles}\right\}.$
For $\gamma>0$, the family of SVR constrained classifiers $\mathscr{F}_{\gamma}$ is defined as
\begin{equation}\label{SVR family}
\mathscr{F}_\gamma = \left\{f\in\mathscr{F}_0: \forall\; f'\in\mathscr{F}_0, \;
S(f') - S(f) \ge -\gamma V(\Omega_1(f)\btr\Omega_1(f'))
\right\},
\end{equation}
where $A\btr B = (A\cap B^c) \cup (A^c \cap B)$ denotes the set difference between two generic sets $A$ and $B$.
The elements of $\mathscr{F}_\gamma$ are constrained by SVR in the sense that the surface area of $\Omega_1(f)$ can decrease at most proportionally to the change in its volume. The parameter $\gamma$ is set to be no smaller than $2d$ to prevent $\mathscr{F}_\gamma$ from being empty. The definition of $\mathscr{F}_\gamma$ matches the intuition in Algorithm 1 where SVR in every leaf node is checked. Moreover, the following proposition reveals that the regularized risk minimizer of all trees over the true measure, provided it exists, always lies in the family $\mathscr{F}_\gamma$. Therefore, in the subsequent analysis, we will focus on the effect of $\gamma$ on the convergence of excess risk since it plays an equivalent role to $\lambda$ in the minimization problem \eqref{Tree Impu}.

% \begin{equation}\label{SVR family}
% \mathscr{F}_\gamma = \left\{f\in\mathscr{F}_0: \forall\; \text{hyperrectangle} \;A\subset \Omega, \;\frac{S(\Omega_1(f)\cup A)-S(f)}{V(A\backslash \Omega_1(f))}\ge-\gamma, \; \frac{S(\Omega_1(f)\backslash A)-S(f)}{V(A\cap \Omega_1(f))}\ge-\gamma \right\}.
% \end{equation}
% The family $\mathscr{F}_\gamma$ is defined such that including (or excluding) any hyperrectangle to (or from) the decision sets can not result a decrease in surface that is greater than $\gamma$ times the volume change. The parameter $\gamma$ is set to be no smaller than $2d$ to avoid $\mathscr{F}_\gamma$ being empty. The definition of $\mathscr{F}_\gamma$ matches the intuition in Algorithm 1 where SVR in every leaf node is checked. Moreover, the following proposition shows the regularized risk minimizer of all trees over the true measure, provided exists, always lies in $\mathscr{F}_\gamma$.

\begin{proposition}\label{SVR family well}
Suppose that $\Omega=[0,1]^d$, the density $p$ of the weighted marginal measure $\mathbb{P}$ exists and satisfies $\rho(x)\le \rho_{\max}<\infty$ for all $x\in \Omega$. If $\gamma \ge 2d\lambda + 1 + \rho_{\max}/\lambda$, then
$$\inf_{f\in\mathscr{F}_\gamma}[\tilde{I}(f,\mathbb{P}) + \lambda r(f)] = \inf_{f\in\mathscr{F}_0}[\tilde{I}(f,\mathbb{P}) + \lambda r(f)],$$
$$\inf_{f\in\mathscr{F}_\gamma}[R(f) + \lambda r(f)] = \inf_{f\in\mathscr{F}_0}[R(f) + \lambda r(f)].$$
\end{proposition}

% For empirical measure induced by dataset $\mathscr{D}_n = \{(X_i,Y_i)\}_{i=i}^n$, for all coordinate $j$, $1\le j \le d$, sort $X_i, 1\le i\le n$ by their $j$th feature as $X[j]_{j1}\le X[j]_{j2} \le \cdots \le X[j]_{jn}$. Then the family $\mathscr{T}_{\gamma,n}$ is defined as
% \begin{align*}
% \mathscr{F}_{\gamma,n} = \big\{ & T\in\mathscr{F}_r: T \;\text{is a tree formed by finite splits}; \\
% & \text{splits at} \; j\mr{th} \; \text{feature can only occur at} \; (X[j]_{jl_1}+X[j]_{jl_2})/2, l_1\ne\l_2\big\}.
% \end{align*}
% Family $\mathscr{F}_{\gamma,n}$ is essentially the intersection between $\mathscr{F}_{\gamma}$ and the family of possible trees determined by the dataset. Any possible CART tree satisfying SVR constraints of $\mathscr{F}_\gamma$ is included in $\mathscr{F}_{\gamma,n}$. 

% and designate that family as the family where the true classifier comes from
%For the optimal classifier that minimizes the classification risk, it

The family $\mathscr{F}_\gamma$ only contains tree classifiers and is not closed. It is desirable to consider the class of all classifiers that can be well approximated by tree classifiers. Let $d_H(A, B)$ be the Hausdorff distance between two generic sets $A$ and $B$. We define the closure of $\mathscr{F}_\gamma$ with respect to Hausdorff distance of decision sets as
$$\overline{\mathscr{F}_\gamma} \triangleq \left\{ f:\Omega\to[0,1]: \exists \;\text{a sequence}\;f_i\in\mathscr{F}_\gamma, \;\text{such that}\; \lim_{i\to\infty}d_H(\Omega_1(f_i),\Omega_1(f))=0\right\}.$$
The elements of $\overline{\mathscr{F}_\gamma}$ have a well-defined surface area as shown in Proposition \ref{SVR family surface}.

\begin{proposition}\label{SVR family surface}
$S(f)$ is well defined for all $f\in\overline{\mathscr{F}_\gamma}$ in the sense that: (i) For all sequences $f_i\in\mathscr{F}_\gamma$ whose decision sets $\Omega_1(f_i)$ converge in Hausdorff distance, $S(f_i)$ is a convergent sequence in $\mathbb{R}$; (ii) for any two sequences $f_{i1}\in\mathscr{F}_\gamma$ and $f_{i2}\in\mathscr{F}_\gamma$ satisfying $\lim_{i\to\infty}d_H(\Omega_1(f_{i1}),\Omega_1(f))=0$ and  $\lim_{i\to\infty}d_H(\Omega_1(f_{i2}),\Omega_1(f))=0$, we have $\lim_{i\to\infty}S(f_{i1}) = \lim_{i\to\infty}S(f_{i2})$.
\end{proposition}

The proof of Proposition \ref{SVR family well} and \ref{SVR family surface} are available in the supplementary material. To better control the rate of convergence of excess risk, we also need the following margin condition.

\begin{condition}\label{margin cond}
There exist constants $\epsilon_0\in (0,1)$, $c_1>0$, $\kappa\ge 1$, such that for all classifiers $f$ satisfying $\mathbb{P}(\Omega_1(f)\btr \Omega_1(f^*))<\epsilon_0$, we have 
$$R(f)-R^* \ge c_1 \left[\mathbb{P}(\Omega_1(f)\btr \Omega_1(f^*)) \right]^{\kappa}.$$
\end{condition}
Condition \ref{margin cond} is essentially Assumption (A1) in \cite{tsybakov2004optimal}, which is known as the Tsybakov's margin condition
and widely used in classification problems for obtaining fast convergence rate of excess risk \citep{tsybakov2004optimal, audibert2007fast}. 
Condition \ref{margin cond} holds as long as $\mathbb{P}(|\eta(X)-1/2|\leq t) \leq C_{\eta} t^{1/(\kappa-1)}$ for all $t\in (0,t^*)$ and some constants $C_{\eta}>0$ and $t^*\in (0,1/2)$; see Proposition 1 in \cite{tsybakov2004optimal}. Under Condition \ref{margin cond}, our main theorem states the rate of convergence for the empirical risk minimizer over $\mathscr{F}_{\gamma}$.

\begin{theorem}[Rate of Convergence] \label{rate}
Suppose that $\Omega=[0,1]^d$, $d\ge 2$ and Condition \ref{margin cond} holds. 
%Moreover, the density $\rho_\alpha$ of the weighted marginal measure $\mathbb{P}_\alpha$ exists and satisfies $\rho_\alpha(X)\in[\rho_{\min},\rho_{\max}], \forall X\in \Omega$. 
Let the empirical risk minimizer be
$\hat{T}_{n,\gamma} = \argmin_{T\in\mathscr{F}_{\gamma}} R_n(T)$ and let $\tilde{c} = \min \big(2\log 2, \frac{21c_1^{1/\kappa}}{16} \big)$, then there exist constants $N_0$ as in equation (\ref{N0}) and $c'$ as in Lemma \ref{border}, such that for all $n\ge N_0$, 
\begin{equation}\label{rate bound}
\sup_{f^*\in\overline{\mathscr{F}_{\gamma}}}\mathbb{E} \left[R(\hat{T}_{n,\gamma}) - R(f^*)\right] \le 16 (2/\tilde c)^{\frac{1}{d+1-1/\kappa}} (c'\gamma)^{\frac{d}{d+1-1/\kappa}} \left(\frac{\log n}{n}\right)^{\frac{1}{d+1-1/\kappa}}.
\end{equation}
\end{theorem}

Theorem \ref{rate} provides an upper bound in \eqref{rate bound} for the rate of convergence of the excess risk from the empirical risk minimizer $\hat{T}_{n,\gamma}$ among all SVR constrained tree classifiers in $\mathscr{F}_\gamma$ to the theoretically best possible classifier $f^*$ in $\overline{\mathscr{F}_{\gamma}}$ with a well-defined surface area. Because we have limited restrictions on the split locations of $\hat{T}_{n,\gamma}$ apart from belonging to $\mathscr{F}_\gamma$, it is possible for a split to land directly on an observed sample under the empirical measure. In such situations, we allow the tree of $\hat{T}_{n,\gamma}$ to assign the sample to either leaf node. Hence, the leaf nodes of $\hat{T}_{n,\gamma}$ can be open, closed, or half open half closed. On the other hand, the theoretically best classifier $f^*$ can lie in the much larger family of $\overline{\mathscr{F}_\gamma}$. Although the classifiers from $\mathscr{F}_{\gamma}$ are tree classifiers, there is no limitation on how many leaf nodes they can have. Therefore, they can approximate decision sets with smooth boundaries, for example a $d$-dimensional ball. In fact, by adjusting $\gamma$, $\overline{\mathscr{F}_\gamma}$ will contain most classifiers with regularly shaped decision sets.
% \textcolor{blue}{(Cheng: I wonder if we can claim this more rigorously.)}.
%Whatever minimizing the loss is be adopted.

The rate of convergence for the excess risk in the upper bound \eqref{rate bound} is $O((\log n / n)^{\frac{1}{d+1-1/\kappa}})$, which is the same as the minimax rate of set estimation on the family of sets with first order differentiable boundaries as in \cite{mammen1995asymptotical}. Although \cite{mammen1995asymptotical} do not use the margin condition and hence the term $1/\kappa$ does not appear explicitly in their rates, their results will yield the same rate as ours up to logarithm factors if the margin condition is applied; see for example, Theorem 4 of \cite{scott2006dyadic}. Furthermore, in the upper bound \eqref{rate bound} in Theorem \ref{rate}, we have explicitly characterized the effect of the SVR regularization $\gamma$ from the set $\mathscr{F}_{\gamma}$. By definition, the SVR constraint in $\mathscr{F}_{\gamma}$ becomes weaker as $\gamma$ increases. As expected, the upper bound in \eqref{rate bound} increases with $\gamma$ given the increasing complexity of the set $\mathscr{F}_{\gamma}$. A direct implication of this upper bound is that if the true tree classifier belongs to a family $\overline{\mathscr{F}_\gamma}$ with a smaller $\gamma$, i.e. a smaller SVR ratio, then the excess risk of the empirical risk minimizer $\hat{T}_{n,\gamma}$ also tends to be smaller.

% Some previous works on classification have discussed the ``fast rate" (faster than $n^{-1/2}$) under the margin condition, such as \cite{tsybakov2004optimal} and \cite{audibert2007fast}. The fast rates in these works are possible only if the exponent of logarithm of $\epsilon$ entropy is smaller than $1$, which does not hold in our problem unless $d=1$. \textcolor{blue}{(Cheng: Can you explain more clearly on this?)} 

Theorem \ref{rate} is related to the literature on classification risk of decision trees, but with important differences.
In general, the rate of convergence for the excess risk depends on two assumptions: the margin condition as in our Condition \ref{margin cond} and the complexity, or the $\epsilon$-entropy, of the family of decision sets \citep{scott2006dyadic}. Suppose the $\epsilon$-entropy is $(1/\epsilon)^\rho$ for some $\rho>0$. If $\rho<1$, then together with a suitable margin condition like Condition \ref{margin cond}, the rate of convergence can be faster than $O(n^{-1/2})$ or even faster than $O(n^{-1})$ \citep{tsybakov2004optimal, audibert2007fast}. If $\rho\ge 1$, then the rate of convergence is typically of order $n^{-\frac{1}{\rho+c}}$ with $c$ being a constant independent of $\rho$ \citep{mammen1995asymptotical}. Although it is possible to slightly improve the constant $c$ by refining the margin condition, the overall rate of convergence is still governed by the $\epsilon$-entropy if $\rho$ is very large. 
%
% In our problems, by Lemma \ref{net size} in the appendix, the epsilon entropy is always  greater than $1$ when $d\ge 2$, resulting the rate for excess error in Theorem \ref{rate}. 
Therefore the major objective of our research is to find a suitable family of sets with controllable $\epsilon$-entropy as well as a practically feasible algorithm. This is not easy, because there are limited families of sets having known 
 $\epsilon$-entropy, including convex sets and sets with smooth boundaries.
Both smoothness and convexity impose strong assumptions on the local properties of the boundary $\partial \Omega_1(f)$. In contrast, we propose to consider the family of SVR constrained sets $\mathscr{F}_{\gamma}$ as a new family of decision sets, which directly imposes a simple and intuitive assumption on the shape of decision sets and additionally has a tractable $\epsilon$-entropy.

 \cite{scott2006dyadic} also noticed that the family of sets with smooth boundaries is not very realistic for practical problems. They proposed a family of decision sets named the ``box-counting class'' which satisfies the conclusion of our Lemma 5. However, they did not further characterize the behavior of sets in their class, and hence they must compute the tree for all smallest dyadic splits and then merge all these trees bottom-up, which can be regarded as an exhaustive search in the lowest split level. This leads to a computational complexity $O(dnL^d\log(nL^d))$, with $L$ being the number of dyadic splits. The complexity will become intractable when $d$ is large, with the authors noting that $d\ge 10$ will pose a problem. \cite{scott2006dyadic} also proposed an alternative to the Tsybakov's margin condition (Condition \ref{margin cond}) and obtained a faster convergence rate in some specific situations, such as when $\Omega$ lies in a low dimensional manifold. However, their alternative condition directly bounds the excess error between a possible dyadic tree and the oracle classifier, which is much stronger than the standard margin condition we have used.

The proof of Theorem \ref{rate} is in the Appendix. The key idea is to construct an $\epsilon$-net on $\mathscr{F}_{\gamma}$ under a symmetric set difference from the true measure and to show that this $\epsilon$-net is also a $2\epsilon$-net under a symmetric set difference from the empirical measure with high probability. This nontrivial property is the reason why we can derive the rate of convergence for the empirical risk minimizer over the whole family of $\mathscr{F}_{\gamma}$ rather than just for the minimizer over an $\epsilon$-net of this family, as adopted by some previous classification literature \citep{mammen1995asymptotical,tsybakov2004optimal}. 

\section{Numerical Studies}
We compare SVR trees with popular imbalanced classification methods on real datasets, adding redundant features to these datasets in evaluating feature selection. A confusion matrix (Table \ref{confusion matrix}) is often used to assess classification performance. A common criteria is accuracy  $\frac{\text{TP} + \text{TN}}{\text{TP} + \text{TN} + \text{FP} + \text{FN}}.$ When minority samples are relatively rare, it is often important to give higher priority to true positives, which is accomplished using the true positive rate (TPR, also known as recall)
$\frac{\text{TP}}{\text{TP}+\text{FN}}$ and precision $\frac{\text{TP}}{\text{TP}+\text{FP}}.$ To combine these, 
the F-measure is often used: $\frac{2\cdot \text{TPR}\cdot \text{Precision}}{ \text{TPR} + \text{Precision}}$. The Geometric mean of true positive rate and true negative rate ($G_{\text{mean}}=\frac{\text{TP}\cdot \text{TN}}{(\text{TP}+\text{FN})\cdot(\text{TN}+\text{FP})}$) is also used to evaluate the overall performance on both classes.

\begin{table}[h]
\centering
\caption{Confusion matrix for two class classification.}\label{confusion matrix}
\begin{tabular}{ccc|c|}
		&  & \multicolumn{2}{c}{True Label} \\ \cline{3-4}
		& \multicolumn{1}{c|}{ } & 1 & 0 \\ \cline{2-4}
\multirow{2}{*}{\makecell{Predicted \\Label}}
			& \multicolumn{1}{|c|}{1} & True Positive (TP) & False Positive (FP) \\ \cline{2-4}
			& \multicolumn{1}{|c|}{0} & False Negative (FN) & True Negative (TN) \\ \cline{2-4}
\end{tabular}
\end{table}

\subsection{Datasets}
We test our method on 12 datasets from the UCI Machine Learning Repository \citep{Dua:2019} and KEEL-dataset repository \citep{alcala2011keel}, varying in size, number of features and level of imbalance. For datasets with three or more classes, classes are combined to form binary class datasets. Table \ref{datasets} gives an overview of the datasets; detailed descriptions and preprocessing of these datasets are available in Section 5.1 of the supplementary material.
\begin{table}[h]
\centering
\caption{Overview of data sets}\label{datasets}
\begin{tabular}{c|cccc}
\hline
Data set& number of & number of & proportions of & number of \\
name & total samples & minority samples & minority samples & features \\ \hline
% Glass & 213 & 29 & 12.0\% & 9 \\
Pima & 768 & 268 & 34.9\% & 8 \\
Titanic & 2201 & 711 & 32.3\% & 3 \\
Phoneme & 5404 & 1586 & 29.3\% & 5 \\
Vehicle & 846 & 199 & 19.0\% & 18 \\
Ecoli & 336 & 52 & 15.5\% & 6 \\
Segment & 2308 & 329 & 14.3\% & 18 \\
Wine & 1599 & 217 & 11.9\% & 11 \\
Page & 5472 & 559 & 10.2\% & 10 \\
Satimage & 6435 & 626 & 8.9\% & 36 \\
Glass & 214 & 17 & 8.0\% & 9 \\
Abalone & 731 & 42 & 5.4\% & 7 \\
Yeast & 1484 & 51 & 3.4\% & 8 \\
 \hline
\end{tabular}
\end{table}

\subsection{Experimental Setup}\label{setup}
We test the performance of SVR-Tree, SVR-Tree with feature selection, CART \citep{breiman1984classification} with duplicated oversampling, CART with SMOTE \citep{chawla2002smote}, CART with Borderline-SMOTE \citep{han2005borderline}, CART with ADASYN \citep{he2008adasyn} and Hellinger Distance Decision Tree \citep{cieslak2012hellinger} on all 12 datasets. Features are linearly transformed so that samples lie in $[0,1]^d$. \par{}

We use $3\times 5$ nested stratified cross validation to evaluate the out-of-sample performance of all methods. The nested cross validation has two layers: the inner layer is used to choose tuning parameters and the outer layer is used to evaluate the out-of-sample performance.
For each method and each dataset, we run the nested cross validation procedure 20 times. In each run, we randomly partition the whole dataset into three stratified folds\footnote{the proportions of samples of each class are roughly the same in all folds.} and run three times. Each time one of the three folds is selected as the testing set and the other two folds are training sets. On the training sets, we use 5-fold stratified cross validation to choose the tuning parameter. That is, we further divide the training set into five folds and run for five times. Each time we use four folds as (inner) training sets and the other fold for validation. We define the cross-validation TPR, TNR in the same manner as the cross-validation accuracy and use the cross validation TPR, TNR to compute the cross-validation F-measure. The tuning parameter with the highest cross validation F-measure is selected. We then train the model with the selected tuning parameter on the whole training set and evaluate its performance on the test set. The cross validation accuracy, precision, TPR, F-measure and G-mean are recorded. The mean and standard error of these statistics from 20 nested cross validation runs are reported in Table \ref{results}.

% The average values of TP, TN, FP, FN on testing sets are used to compute the accuracy, TPR, precision and F-measure. The specific settings for each method are discussed below.\par{}

SVR-Tree with and without feature selection are described in Algorithm 1. The weight for the minority class, $\alpha$, is set to be the largest integer that makes the total weight of the minority class no greater than the total weight of the majority class. The maximal number of leaf nodes $\bar{a}_n$ is $2\sqrt{n}$. The penalty parameter $\lambda_n$ for SVR is chosen from a geometric sequence in $[2^0, 2^{10}]\times10^{-3}\times n^{-1/3}$ and the value with the highest F-measure across 20 runs is selected. For SVR-Tree with feature selection, the constant $c_0$ in equation (\ref{impu decr}) is fixed to $c_0=4$. In practice, the results are insensitive to $c_0$.\par{}

For the other methods except Hellinger distance decision tree, we first over sample the minority class samples, such that the number of minority samples are multiplied by the same weight $\alpha$ used for SVR tree. We than build the CART tree on the over sampled dataset and prune it. The pruning parameter of CART is selected to maximize the F-measure. By the algorithm proposed by \cite{breiman1984classification}, the range from which to choose the pruning parameter will be available after we build the tree and does not need to be specified ahead of time. \par{}

%\textcolor{blue}{(Cheng: Is this the same $\alpha$ as the one used for SVR methods? If so, please clarify. You can add in the previous paragraph that the same weight $\alpha$ for the minority class is used for the other methods. Answer: Yes)}

For duplicated oversampling, we sample each minority class sample $\alpha-1$ times where $\alpha$ is the weight used for SVR tree. For SMOTE, the number of nearest samples is set as $k=5$. For Borderline-SMOTE, we adopt the Borderline-SMOTE1 of \cite{han2005borderline}, with the number of nearest samples $k=5$. For both SMOTE and Borderline-SMOTE, if $\alpha-1\ge k$, some nearest neighbors may be used multiple times to generate synthetic samples, such that the total weight of minority class samples of these oversampling methods are the same as that of SVR tree. For ADASYN, if we denote the number of majority class samples as $n_0$ and the number of minority class samples as $n_1$, then the parameter $\beta$ is set to be $\beta = \alpha n_1 / n_0$.

%\textcolor{blue}{(Cheng: Again is this the same $\alpha$ as before?)}

%of Hellinger distance decision tree
For Hellinger distance decision tree, we directly call its java implementation available at Weka software (\url{https://www.cs.waikato.ac.nz/ml/weka/}). 
Although it is claimed that the Hellinger distance tree is insensitive to imbalance, their current implementation occasionally classifies all samples into a single class.
We set the default leaf prediction strategy as ``Naive Bayes'' since it works well on most datasets. However, if the algorithm classifies all samples into a single class, we will set the leaf prediction strategy as ``Majority Voting''. The third leaf prediction strategy ``Naive Bayes Adapative'' always produces the same results as one of the other two strategies in our experiments. 
% The Hellinger distance decision tree is very insensitive to all other turning parameters provided in the java implementation, so the inner 5-fold cross validation is skipped. 
%\textcolor{blue}{(I suggest that we remove the last sentence. Usually this will invite reviewers to ask how "insensitive" it is to the tuning parameters. Answer: in all my experiments other tuning parameter never changes the result; meaning the TP, TN, FP, FN always remain the same regardless how I change other tuning parameters. I even guess there are bugs in their codes.)}

\begin{table}
\centering
\scriptsize
\renewcommand{\arraystretch}{1.5}
% \begin{table}
\begin{tabular}[h]{c|c|ccccc}
\hline
Data set & Method & Accuracy & Precision & TPR & F-measure & G-mean  \\ \hline
\multirow{7}{*}{Phoneme} 
				& SVR & 0.8350(0.0060) & 0.6749(0.0113) & \textbf{0.8458}(0.0154) & 0.7506(0.0076) & 0.8380(0.0062) \\*
				& SVR-Select & 0.8377(0.0040) & 0.6807(0.0064) & 0.8420(0.0098) & 0.7528(0.0060) & \textbf{0.8389}(0.0049)  \\*
				& Duplicate & 0.8554(0.0038) & 0.7567(0.0110) & 0.7482(0.0154) & 0.7522(0.0068) & 0.8205(0.0065) \\* 
				& SMOTE & 0.8598(0.0054) & \textbf{0.7637}(0.0119) & 0.7567(0.0121) & 0.7601(0.0089) & 0.8264(0.0068) \\*
				& BSMOTE & 0.8568(0.0060) & 0.7580(0.0143) & 0.7527(0.0127) & 0.7552(0.0093) & 0.8230(0.0070) \\*
				& ADASYN & \textbf{0.8607}(0.0050) & 0.7661(0.0116) & 0.7565(0.0115) & \textbf{0.7612}(0.0082) & 0.8269(0.0063)  \\* 
                & HDDT & 0.7065(0.0000) & 0.0000(0.0000) & 0.000(0.0000) & 0.0000(0.0000) & 0.0000(0.0000) \\ \hline 
\multirow{7}{*}{Segment} & SVR & 0.9922(0.0017) & 0.9696(0.0089) & 0.9763(0.0061) & 0.9729(0.0060) & 0.9855(0.0033) \\*
				& SVR-Select & \textbf{0.9932}(0.0013) & \textbf{0.9753}(0.0080) & 0.9769(0.0053) & \textbf{0.9761}(0.0047) & \textbf{0.9863}(0.0027)  \\*
				& Duplicate & 0.9913(0.0012) & 0.9707(0.0069) & 0.9685(0.0073) & 0.9696(0.0044) & 0.9817(0.0036) \\* 
				& SMOTE & 0.9916(0.0010) & 0.9722(0.0061) & 0.9688(0.0081) & 0.9705(0.0035) & 0.9820(0.0038) \\*
				& BSMOTE & 0.9915(0.0011) & 0.9726(0.0082) & 0.9676(0.0059) & 0.9701(0.0037) & 0.9814(0.0027)  \\*
				& ADASYN & 0.9919(0.0011) & 0.9713(0.0063) & 0.9716(0.0064) & 0.9714(0.0040) & 0.9833(0.0032) \\*
				& HDDT & 0.8330(0.0021) & 0.4599(0.0031) & \textbf{0.9825}(0.0023) & 0.6266(0.0031) & 0.8910(0.0017) \\ \hline
\multirow{7}{*}{Page} 
				& SVR & 0.9656(0.0015) & 0.8248(0.0098) & 0.8429(0.0117) & 0.8337(0.0076) & 0.9087(0.0062) \\*
				& SVR-Select & 0.9647(0.0016) & 0.8252(0.0124) & 0.8311(0.0140) & 0.8280(0.0077) & 0.9024(0.0073)  \\*
				& Duplicate & \textbf{0.9686}(0.0013) & 0.8465(0.0081) & \textbf{0.8463}(0.0135) & \textbf{0.8463}(0.0072) & \textbf{0.9119}(0.0071) \\* 
				& SMOTE & 0.9683(0.0021) & \textbf{0.8468}(0.0133) & 0.8426(0.0144) & 0.8446(0.0101) & 0.9099(0.0078) \\*
				& BSMOTE & 0.9682(0.0015) & 0.8444(0.0103) & 0.8448(0.0127) & 0.8445(0.0077) & 0.9109 (0.0067) \\*
				& ADASYN & 0.9677(0.0021) & 0.8436(0.0131) & 0.8403(0.0148) & 0.8418(0.0103) & 0.9084(0.0080)  \\* 
                & HDDT & 0.9024(0.0030) & 0.5479(0.0338) & 0.2720(0.0260) & 0.3622(0.0215) & 0.5141(0.0233)  \\ \hline
\multirow{7}{*}{Yeast} 
				& SVR & 0.9368(0.0100) & 0.2641(0.0410) & 0.4510(0.0888) & \textbf{0.3290}(0.0416) & 0.6524(0.0619)  \\*
				& SVR-Select & 0.9287(0.0122) & 0.2360(0.0397) & 0.4559(0.0715) & 0.3077(0.0391) & 0.6544(0.0496) \\*
				& Duplicate & \textbf{0.9641}(0.0037) & \textbf{0.3907}(0.2084) & 0.1196(0.0784) & 0.1727(0.0981) & 0.3077(0.1544) \\* 
				& SMOTE & 0.9594(0.0045) & 0.3620(0.0859) & 0.2098(0.0830) & 0.2554(0.0752) & 0.4458(0.0887) \\*
				& BSMOTE & 0.9602(0.0043) & 0.3664(0.1012) & 0.2020(0.0884) & 0.2499(0.0937) & 0.4333(0.1065)  \\*
				& ADASYN & 0.9605(0.0045) & 0.3709(0.1246) & 0.2078(0.1099) & 0.2491(0.0986) & 0.4288(0.1436) \\* 	
                & HDDT & 0.7604(0.0200) & 0.1049(0.0101) & \textbf{0.7853}(0.0287) & 0.1849(0.0162) & \textbf{0.7722}(0.0211) \\ \hline
\multirow{7}{*}{\makecell{Average Ranking}}
				& SVR & 4.92 & 5.25 & \textbf{2.25} & \textbf{2.08} & \textbf{2.17} \\*
				& SVR-Select & 5.58 & 5.00 & \textbf{2.25} & 2.75 & 2.50 \\*
				& Duplicate & 2.79 & 3.29 & 5.13 & 4.88 & 4.96  \\* 
				& SMOTE & \textbf{2.54} & \textbf{2.29} & 5.04 & 4.21 & 4.71 \\*
				& BSMOTE & 3.33 & 2.92 & 5.58 & 5.00 & 5.42 \\*
				& ADASYN & 2.75 & 2.92 & 4.58 & 4.00 & 4.42 \\* 
				& HDDT & 6.08 & 6.33 & 3.17 & 5.08 & 3.83 \\ \hline	
\end{tabular}
\caption{Results of numerical studies, where the ranking is averaged over all 12 datasets.}\label{results}
\end{table}

\subsection{Results}
We compute the mean and standard error of accuracy, precision, TPR, F-measure and G-mean across the 20 nested cross validations. Due to limited space, results of 4 representative datasets are shown in Table \ref{results} with the remaining results in Section 5.2 of the supplementary material.
If a method classifies all samples into the majority class, precision is not well-defined and we simply set it to be zero. In the column ``Method'', SVR = SVR-Tree, SVR-Select = SVR-Tree with feature selection, Duplicate = CART with duplicated oversampling, SMOTE = CART with SMOTE, BSMOTE = CART with Borderline-SMOTE, ADASYN = CART with ADASYN and HDDT = Hellinger Distance Decision Tree. For each dataset and evaluation measure, the method with the highest mean value ranks first and is highlighted in bold. 
As suggested by \cite{brazdil2000comparison}, the average rankings of all 12 datasets are displayed in Table \ref{results} to evaluate the overall performance of each method across all datasets. If two methods are tied, the ranking is set to the average. 
% For example, both Duplicate and SMOTE have ranking 1.5 on accuracy of the Titanic dataset.

%parsimonious
The results in Table \ref{results} show that oversampling based methods generally perform well in accuracy and precision, with SMOTE the best of these methods. However, SVR and SVR-select outperform the oversampling methods in the rankings of TPR, F-measure and G-mean on many datasets and on average. 
%The oversampling methods may be too conservative in classifying samples into the minority class even though they assign equal weights (or equivalent amounts of oversampling) on the minority class as the SVR methods. 
%In comparison, SVR and SVR-select are able to achieve a more balanced performance, reflected by the higher values of F-measure (a trade-off between precision and TPR) and G-mean (a trade-off between TPR and TNR). 
The behavior of HDDT is unstable across datasets, being competitive in TPR and G-mean for some datasets but sometimes assigning all samples to the majority class.
Since all methods (except duplicated sampling) are built on ideal theoretical assumptions, their practical performance is typically data-dependent. For example, SMOTE is outperformed by SVR-select in all measures for the Segment dataset, while the opposite happens for the Page dataset.

\section{Discussion}
%\paragraph{Regularization for Classification Tree} 

A major challenge in analyzing imbalanced data is small sample size in the minority class leading to overfitting. It is natural to consider using  regularization to address this problem.  Regularization of classification trees is an old idea; for example, \cite{breiman1984classification} proposed to penalize the number of leaf notes in the tree. 
%by solving an optimization problem to prune the tree.
%$$\hat{T}_n = \arg\min_{T\subset\mathbb{T}} \sum_{i=1}^n|Y_i-T(X_i)| + \lambda a_n(T),$$
%where $a_n(T)$ is the number of leaf nodes of tree $T$, $\lambda$ is a penalty parameter on the number of leaf nodes, and $\sum_{i=1}^n|Y_i-T(X_i)|$ is the error on training data. 
Other classification trees like C4.5 \citep{quinlan2014c4} and Ripper \citep{cohen1995fast} also prune overgrown trees. However, the number of leaf nodes may not be a good measure of complexity of a classification tree. Recently, \cite{hahn2020bayesian} add a Bayesian prior to an ensemble of trees, which functions as indirect regularization.
Following the idea of directly regularizing the shape of the decision set and complexity of the decision boundary, we instead penalize the surface-to-volume ratio. To our knowledge, this is a new idea in the statistical literature. 

%\paragraph{Generalization of SVR-Tree}
SVR-Tree can be trivially generalized to the multi-class case and balanced datasets. For multiple classes, we can apply SVR to one or more minority classes and take the sum of these ratios as regularization. For balanced datasets, we can either compute the SVR ratio of all classes, or we can simply regularize the surface of the decision boundary. The principle behind these generalizations is to regularize the complexity of the decision boundaries and shapes of the decision sets.

%\par{}
%The application of surface-to-volume ratio to a general classifier is much more challenging. Unlike classification trees whose decision sets have a regular shape, a general decision set may make the computation of surface and volume intractable.  

\bigskip
\newpage
\begin{center}
{\large\bf SUPPLEMENTARY MATERIAL}
\end{center}
{\bf Proofs, a Detailed Algorithm and Additional Results:} Supplementary Material for ``Classification Trees for Imbalanced Data: Surface-to-Volume Regularization''. \\
{\bf Codes:} \url{https://github.com/YichenZhuDuke/Classification-Tree-with-Surface-to-} \\
\url{Volume-ratio-Regularization.git}.

\bibliographystyle{chicago}
\bibliography{IBbib}

\appendix
\section{Proofs}\label{proofs}
We prove Theorem 1 and 2 here. Proofs for Corollary 1, $c_1$ bounds in examples 1-2 and lemmas introduced in the Appendix are in the supplementary material. Without loss of generality, we assume $\Omega = [0,1]^d$ in this section.

\subsection{Proof of Theorem \ref{consistency}}
The proof builds on \cite{nobel1996histogram}, \cite{gyorfi2006distribution}, \cite{tsybakov2004optimal} and \cite{scornet2015consistency}. We first establish a sufficient condition for consistency, showing a classification tree whose signed impurity converges to an oracle bound is consistent. We then break the difference between signed impurity of $\hat T_n$ and the oracle bound into two parts. The first is estimation error, which goes to zero if the number of leaves $a_n$ increases  slower than $n$; the second is approximation error, which goes to zero if $\mathbb{E}(Y|X\in A)$ goes to a constant within each leaf node and penalty  $\lambda_n$ goes to zero. \par{}

% Denote $p(x) = \mathbb{E}(Y=1|X=x)$, and its weighted version $p_\alpha(x) = [\alpha p(x)][1-p(x)+\alpha p(x)]$. 
Recall the conditional expectation is denoted as $\mathbb{E}(Y|X)=\eta(X)$, define the oracle lower bound for signed impurity as
$I^* = \int_\Omega 2\eta(x)(1-\eta(x)) d\mathbb{P}(x).$
The following lemma shows if the signed impurity of a classification tree converges to $T^*$ as $n\to\infty$, the classifier associated with the tree is consistent.

\begin{lemma}\label{lem: signed impu consistency}
Let $T_n$ be a sequence of classification trees, let $f_n: \Omega\to \{0,1\}$ be the classifier associated with $T_n$. $T_n$ is consistent if
%\begin{equation}\label{signed impu consistency}
$\tilde{I}(T_n, \mathbb{P}) \to I^*$
%\end{equation}
in probability.
\end{lemma}

We then decompose the difference between signed impurity of $\hat{T}_n$ and the oracle bound into estimation error and approximation error.

\begin{lemma}\label{lem:decompose}
Let $\hat{T}_n$ be a classification tree trained from  data $\mathscr{D}_n$, $A_1, A_2, \ldots A_m$ be all the leaf nodes of $\hat{T}_n$. Define the set of classifiers $\mathscr{T}_n$ as:
$$\mathscr{T}_n = \{T: T\text{'s associated classifier }f:\Omega\to\{0,1\}\text{ is constant on all }A_j, \;\; 1\le j \le m\}.$$
We have
\begin{equation}\label{decompose}
\tilde{I}(\hat{T}_n,\mathbb{P}) - I^* \le 2 \sup_{T\in\mathscr{T}_n}|\tilde{I}(T, \mathbb{P}) - \tilde{I}(T,\mathbb{P}_n)| + \inf_{T\in\mathscr{T}_n}|\tilde{I}(T, \mathbb{P})+\lambda_n r(T) - I^*|.
\end{equation}
\end{lemma}

The first term on the right hand side of equation (\ref{decompose}) is the  ``estimation error'', which measures the difference between functions evaluated under the empirical and true distributions. The second term is ``approximation error'', which measures the ability of $\mathscr{T}_n$ to approximate the oracle prediction function. The next two lemmas show both terms go to zero in probability.

\begin{lemma}\label{estimation error}
If $\frac{\bar{a}_n d\log n}{n} = o(1),$
we have
$\sup_{T\in\mathscr{T}_n}|\tilde{I}(T, \mathbb{P}) - \tilde{I}(T,\mathbb{P}_n)| \to 0$
in probability.
\end{lemma}

\begin{lemma}\label{approximation error}
As $n\to\infty$,  if $\lambda_n\to 0$ and $\bar{a}_n\to\infty$, 
$\inf_{T\in\mathscr{T}_n} |\tilde{I}(T,\mathbb{P}) + \lambda_n r(T) - I^*| \to 0$
in probability.
\end{lemma}

%We are now prepared to prove Theorem \ref{consistency}.
\begin{proof}[Proof of Theorem \ref{consistency}]
Combining Lemma \ref{lem: signed impu consistency}, \ref{lem:decompose}, \ref{estimation error}, \ref{approximation error}, we finish the proof.\par{}
%The proof is still valid when the optional steps in Algorithm 1 are enabled. For all fixed $k\in \mathbb{N}$, as $n$ goes to infinity, the threshold level $c_0\lambda_n$ goes to zero; Therefore, the probability of rejecting in the first $k$ splits by the optional steps goes to zero as $n$ goes to infinity. Hence, the bound for approximation error, i.e. lemma \ref{approximation error}, still holds. The proofs of the lemmas \ref{lem: signed impu consistency}, \ref{lem:decompose}, \ref{estimation error} do not change. Estimation consistency still holds when we perform feature selection.
\end{proof}

\subsection{Proof of Theorem \ref{rate}}
The proof consists of two parts. In the first part, we construct a subset on $\mathscr{F}_\gamma$ and prove it is both an $\epsilon$-net under symmetric set difference of true measure and with high probability an $2\epsilon$-net under symmetric set differences of empirical measure. In the second part, we prove that with high probability, any element in the above constructed $\epsilon$-net that is far away from the true classifier will incur a large excess error. Combining these two parts and selecting a proper value for $\epsilon$ will conclude the proof of Theorem \ref{rate}.

We first construct the $\epsilon$-net. We divide the space $\Omega=[0,1]^d$ into $M^d$ hypercubes, each with volume $1/M^d$ as below:
$$\mathscr{H} = \left\{ \Big[\frac{i_1}{M}, \frac{i_1+1}{M}\Big]\times \Big[\frac{i_2}{M}, \frac{i_2+1}{M}\Big]\times \cdots \Big[\frac{i_d}{M}, \frac{i_d+1}{M}\Big]: i_1, i_2, \cdots i_d \in \{0, 1, \cdots M-1\}\right\}.$$
The SVR constraint of family $\mathscr{F}_\gamma$ directly implies that the surface area of the decision set is no greater than $\gamma$. If the shape of the decision set is indeed regular, one may imagine that the border of the decision set will intersect only with a finite number of hypercubes in $\mathscr{H}$. The following Lemma formalizes this intuition.
\begin{lemma}\label{border}
There exists a constant $c'\ge 1$ only dependent on $d$, such that for $1/M\ge2/(3\gamma)$, for any $f \in\mathscr{F}_\gamma$, 
$$\big| \{A\in\mathscr{H}: A\cap \partial \Omega_1(f)\ne \emptyset\}\big| \le c' \gamma M^{d-1},$$
where $|\cdot |$ denotes the cardinality.
\end{lemma}
The general proof idea is to examine the local region of each hypercube in $\mathscr{H}$. If $\partial \Omega_1(f)$ intersects with a hypercube, then the SVR constraint and a $d-1$ dimensional isoperimetric inequality will imply that the surface area passing through the local region of that hypercube is lower bounded. Since the total surface area of $\Omega_1(f)$ is upper bounded given the maximum volume is 1 and the SVR constraint, the number of cubes intersecting with $\partial \Omega_1(f)$ is also upper bounded. 

Let $f$ be a classifier such that $\Omega_1(f)$ is a finite union of hypercubes of $\mathscr{H}$. Then for $A\in\mathscr{H}$ and $A\subset \Omega_1(f)$, we say $A$ is a \textit{border hypercube} of $\Omega_1(f)$ if $\partial A \cap \partial \Omega_1(f) \ne \emptyset$. Let $\epsilon = c'\gamma/M$ and define the set $\mathscr{N}_\epsilon$ as
\begin{align*}
\mathscr{N}_\epsilon = \big\{ f: &\; \Omega_1(f) \text{ is a finite union of elements in }\mathscr{H};\\
& \text{ and the number of border hypercubes of } \Omega_1(f) \text{ is no greater than } M^d\epsilon \big\}.
\end{align*}

We have the following lemma regarding the $\epsilon$-entropy of $\mathscr{N}_\epsilon$.
\begin{lemma}\label{net size}
For $\epsilon < 1/4$,  we have
$$\log |\mathscr{N}_\epsilon| \le \sum_{k=1}^{c'\gamma M^{d-1}}{M^d\choose k} \le c'\gamma \left(\frac{c'\gamma}{\epsilon}\right)^{d-1} \left[\log\left(\frac{c'\gamma}{\epsilon}\right) + 1 - \log(c'\gamma)\right].$$
\end{lemma}

The following lemma shows $\mathscr{N}_\epsilon$ is not only an $\epsilon$-net on $\mathscr{F}_\gamma$ under the true measure $\mathbb{P}$, but also simultaneously an $2\epsilon$-net under empirical measure $\mathbb{P}_n$ with high probability.

\begin{lemma}\label{epsilon net}
We have the following two facts regarding $\mathscr{N}_\epsilon$:
\begin{align*}
&~~\sup_{f\in \mathscr{F}_\gamma} \inf_{f_0\in\mathscr{N}_\epsilon} \mathbb{P}(\Omega_1(f)\btr \Omega_1(f_0)) \le \epsilon,\\
&~~\mathbb{P}\left( \sup_{f\in \mathscr{F}_\gamma} \inf_{f_0\in\mathscr{N}_\epsilon} \max\big\{\mathbb{P}_{n}(\Omega_1(f)\btr \Omega_1(f_0)), 2\mathbb{P}(\Omega_1(f)\btr \Omega_1(f_0)) \big\} > 2\epsilon \right) \\
&\le \sum_{k=1}^{c'\gamma M^{d-1}}{M^d\choose k} \exp\{-2(\log 2) n\epsilon\}.
\end{align*}
\end{lemma}

We now move to the second part. By Lemma \ref{epsilon net}, there exists $f_1\in\mathscr{N}_\epsilon$, such that $R(f_1) - R^* \le \mathbb{P}(\Omega_1(f_1) \btr \Omega_1(f^*)) \le \epsilon$. The next lemma shows with high probability, $f_1$ has a smaller empirical risk than all elements in $\mathscr{N}_\epsilon$ that are sufficiently far away from $f^*$.

\begin{lemma}\label{compare lem}
Suppose Condition \ref{margin cond} holds. Then we have
$$\mathbb{P}\left( \sup_{T\in \mathscr{N}_\epsilon, R(T)-R(f^*)\ge 7\epsilon} R_n(T) - R_n(f_1) > 4\epsilon\right) \ge 1 - |\mathscr{N}_\epsilon|\exp\left(-\frac{21}{16} nc_1^{1/\kappa} \epsilon^{2-1/\kappa}\right).$$
\end{lemma}

The key proof idea is to write the empirical loss $R_n(\cdot)$ as a summation of i.i.d. random variables and bound the deviation of the summation to its expectation by Bernstein inequality. 

% We then state two auxiliary lemmas. 
% \begin{lemma}\label{I2R}
% The empirical risk minimizer $\hat{T}_{n,\gamma}$ satisfies $\tilde{I}(\hat{T}_{n,\gamma},\mathbb{P}_n) = R_n(\hat{T}_{n,\gamma})$.
% \end{lemma}

% \begin{lemma}\label{exist}
% For all $f^*\in \mathscr{F}_{\gamma}$, with high probability\footnote{specific probability to be determined}, there exists $T'_1\in \mathscr{F}_{\gamma,n}$ satisfying $\mathbb{P}_{n}(\Omega_1(T'_1) \Delta \Omega_1(f_1))\le 2\epsilon$.
% \end{lemma}

We are now ready to prove Theorem \ref{rate}.
\begin{proof}[Proof of Theorem \ref{rate}]
Define the events $E_1, E_2$ as:
% $$E_1 = \left\{\exists T'_1\in \mathscr{F}_{\gamma, n}, \text{ such that } \mathbb{P}_{n}(\Omega_1(T'_1) \Delta \Omega_1(f_1))\le 2\epsilon \right\},$$
$$E_1 = \left\{ \sup_{f\in \mathscr{F}_\gamma} \inf_{f_0\in\mathscr{N}_\epsilon} \max\big\{\mathbb{P}_{n}(\Omega_1(f)\btr \Omega_1(f_0)), 2\mathbb{P}(\Omega_1(f)\btr \Omega_1(f_0)) \big\} \le 2\epsilon \right\},$$
$$E_2 = \left\{  \sup_{T\in |\mathscr{N}_\epsilon|, R(T)-R(f^*)\ge 7\epsilon} R_n(T) - R_n(f_1) > 4\epsilon\right\}.$$
We claim on event $E_1\cap E_2$, $R(\hat{T}_{n,\gamma}) - R(f^*) \le 8\epsilon$. We prove this by contradiction. Suppose that $R(\hat{T}_{n,\gamma})-R(f^*) > 8\epsilon$. Then by definition of event $E_1$, there exists a classifier $f_2\in\mathscr{N}_\epsilon$, such that
\begin{equation}\label{Rn 1}
|R(\hat{T}_{n,\gamma}) - R(f_2)| \le \mathbb{P}(\Omega_1(\hat{T}_{n,\gamma})\btr \Omega_1(f_2)) \le \epsilon,
\end{equation}
\begin{equation}\label{Rn 2}
|R_n(\hat{T}_{n,\gamma}) - R_n(f_2)| \le \mathbb{P}_{n}(\Omega_1(\hat{T}_{n,\gamma}) \btr \Omega_1(f_2)) \le 2\epsilon.    
\end{equation}
Combining equation (\ref{Rn 1}) and the assumption that $R(\hat{T}_{n,\gamma})-R(f^*) > 8\epsilon$, we have $R(f_2) - R(f^*) > 7\epsilon$. It then follows from the definition of event $E_2$ that
\begin{equation}\label{Rn 3}
R_n(f_2) - R_n(f_1) > 4\epsilon.
\end{equation}
By definition of event $E_1$, $\exists T'_1\in\mathscr{F}_\gamma$, such that 
\begin{equation}\label{Rn 4}
|R_n(T'_1) - R_n(f_1)| \le \mathbb{P}_{n}(\Omega_1(T'_1) \btr \Omega_1(f_1)) \le 2\epsilon.
\end{equation}
% By definition of event $E_2$, we have
% $$|R_n(\hat{T}_{n,\gamma}) - R_n(f_2)| \le \mathbb{P}_{n,}(\Omega_1(\hat{T}_{n,\gamma}) \Delta \Omega_1(f_2)) \le 2\epsilon.$$
Combining the equations (\ref{Rn 2}), (\ref{Rn 3}) and (\ref{Rn 4}), we have
$$R_n(\hat{T}_{n,\gamma}) - R_n(T_1')>0,$$
which contradicts the definition that $\hat{T}_{n,\gamma}$ is the empirical risk minimizer. Therefore we have $R(\hat{T}_{n,\gamma}) - R(f^*) \le 8\epsilon$ on the event $E_1\cap E_2$.
For all $f^*\in\mathscr{F}_\gamma$, the expected loss can be bounded by
\begin{align*}
& ~\quad \mathbb{E} \left[ R(\hat{T}_{n,\gamma})-R(f^*) \right] \\
& \le 
	\mathbb{P}(E_1^c \cup E_2^c) + 8\epsilon \\
 & 	\le  \exp\left[ c'\gamma \left( \frac{c'\gamma}{\epsilon} \right)^{d-1} \log\left(\frac{c'\gamma}{\epsilon}\right) \right]
	\left[\exp(-2(\log 2)n\epsilon) + \exp\left(-\frac{21}{16} nc_1^{1/\kappa}  \epsilon^{2-1/\kappa}\right) \right] + 8\epsilon \\
 & 	\le 2\exp\left[ c'\gamma \left( \frac{c'\gamma}{\epsilon} \right)^{d-1} \log\left(\frac{c'\gamma}{\epsilon}\right) \right] \exp\left(-\tilde{c} n\epsilon^{2-1/\kappa}\right) + 8\epsilon,
\end{align*}
where the second inequality uses Lemmas \ref{border}, \ref{net size}, \ref{epsilon net} and $1-\log(c'\gamma)\le 0$ for $c'\ge 1,\; \gamma\ge 2d\ge 4$; the last inequality uses the definition that $\tilde{c} = \min\left\{2\log 2, \frac{21c_1^{1/\kappa}}{16}\right\}$.
We now defines a constant $N_0$ as
\begin{equation}\label{N0}
\begin{aligned}
N_0 = \inf\Bigg\{n\in\mathbb{N}: ~ & \frac{(d-1/\kappa)\log n + \log\log n}{d+1-1/\kappa} \ge \log(c'\gamma) + \log c_4, \\
    & \frac{\tilde c}{2} n^{\frac{d-1}{d+1-1/\kappa}} (\log n)^{\frac{2-1/\kappa}{d+1-1/\kappa}} c_5^{2-1/\kappa} \ge -\log(4c_4) + \frac{\log n - \log\log n}{d+1-1/\kappa} \Bigg\}.
\end{aligned}
\end{equation}
In the two inequalities of (\ref{N0}), the left-hand side obviously dominates the right-hand side, so $N_0$ is well defined. Let $\epsilon = c_4 (\log n/n )^{\frac{1}{d+1-1/\kappa}}$, where $c_4 = [2(c'\gamma)^d/\tilde c]^{\frac{1}{d+1-1/\kappa}}$. Then for $n\ge N_0$,
\begin{align*}
\frac{c'\gamma \big(\frac{c'\gamma}{\epsilon}\big)^{d-1}\log\big(\frac{c'\gamma}{\epsilon}\big)}{n\tilde{c} \epsilon^{2-1/\kappa} }
    & = \frac{(c'\gamma)^d [\log(c'\gamma) + \log(1/\epsilon)]}{\tilde c n \epsilon^{d+1-1/\kappa}} \\
    & = \frac{(c'\gamma)^d \big[\log(c'\gamma) + \log(1/c_4) + \frac{1}{d+1-1/\kappa}\log(\frac{n}{\log n})\big]}{\tilde c c_5^{d+1-1/\kappa} \log n} \\
    & \le \frac{\log(c'\gamma) + \log(1/c_4) + \frac{1}{d+1-1/\kappa}\log(\frac{n}{\log n})}{2\log n}  \le \frac{1}{2} ,
\end{align*}
where the third inequality utilizes the definition of $c_4$ and the last inequality utilizes the definition of $N_0$.
Therefore we have
$$\mathbb{E}\left[R(\hat{T}_{n,\gamma})-R(f^*)\right] 
	\le  2\exp\left(-\frac{\tilde{c}}{2} n\epsilon^{2-1/\kappa} \right)
 + 8 c_4 \left(\frac{\log n}{n}\right)^{\frac{1}{d+1-1/\kappa}}.$$
By definition of $\epsilon$, $n\epsilon^{2-1/\kappa}$ increases with $n$ in a polynomial order, so the second term in the display above dominates. Specifically, by definition of $N_0$, when $n\ge N_0$, we have
$2\exp\left(-\frac{\tilde{c}}{2} n\epsilon^{2-1/\kappa} \right) \le  8c_4 \left(\frac{\log n}{n}\right)^{\frac{1}{d+1-1/\kappa}}$, and hence
$$\mathbb{E}\left[R(\hat{T}_{n,\gamma})-R(f^*)\right] 
	\le  16 (2/\tilde c)^{\frac{1}{d+1-1/\kappa}} (c'\gamma)^{\frac{d}{d+1-1/\kappa}} \left(\frac{\log n}{n}\right)^{\frac{1}{d+1-1/\kappa}}.$$

% let $N_0$ be:
% \begin{equation}\label{N0}
% N_0 = \inf \left\{n\in\mathbb{N}:  \frac{\tilde c}{2} (c_4/16)^{2-1/\kappa} n^{\frac{d-1}{1+d-1/\kappa}} \ge (\log n - \log\log n)^{\frac{d-1}{1+d-1/\kappa}} - \log (c_4/4) (\log n)^{-\frac{d-1}{1+d-1/\kappa}}  \right\}.
% \end{equation}

\end{proof}

\end{document}

\if1\blind
{
  \title{\bf Supplementary Material for ``Classification Trees for Imbalanced Data: Surface-to-Volume Regularization''}
  \author{Yichen Zhu\thanks{Yichen Zhu is Ph.D. student, Department of Statistical Science, Duke University, Durham, NC 27708.}\hspace{.2cm}, 
  Cheng Li\thanks{Cheng Li is Assistant Professor, Department of Statistics and Applied Probability, National University of Singapore, Singapore, 117546. The research was supported by the Singapore Ministry of Education Academic Research Funds Tier 1 Grant R-155-000-223-114.}\hspace{.2cm}   and
    David B. Dunson\thanks{David B. Dunson is the Art and Science Professor, Department of Statistical Science, Duke University, Durham, NC 27708. The research was partially supported by grant N000141712844 of the United States Office of Naval Research.}\hspace{.2cm}}
  
  %\author{Yichen Zhu\\
    % Department of Statistical Science, Duke University\\
    % and \\
    % David B. Dunson \\
    % Department of Statistical Science, Duke University}
  \maketitle
} \fi

\if0\blind
{
  \bigskip
  \bigskip
  \bigskip
  \begin{center}
    {\LARGE\bf Supplementary Material for ``Classification Trees for Imbalanced Data: Surface-to-Volume Regularization''}
\end{center}
  \medskip
} \fi

The supplementary material contains proofs, algorithms, and results of numerical experiments that are not included in the main paper. Section \ref{sec:lemma1-4} includes the proof of Lemma 1-4 in Appendix A.1 of the main text. Section \ref{sec:lemma5-8} includes the proof of Proposition 1 and 2 in Section 3.2 and Lemma 5-8 in Appendix A.2 of the main text. Section \ref{sec:selection} discusses the feature selection consistency in SVR Trees. Section \ref{sec:algorithm1} provides the details of Algorithm 1 in Section 2.5 of the main text and the detailed analysis of computational complexity. Section \ref{sec:datasets} includes the detailed description and experimental results of 12 datasets in Section 4 of the main text.

% Without loss of generality, we assume $\Omega = [0,1]^d$ in all proofs. We use $\mathbb{E}$, $\mathbb{E}_{n,\alpha}$ to denote the expectation over probability measures $\mathbb{P}$, $\mathbb{P}_{n,\alpha}$, respectively. For all $A\subset \Omega$, we let $p(A) = \mathbb{E}(Y|X\in A)$ and $p(A) = \mathbb{E}_{n,\alpha}(Y|X\in A)$.

\section{Proof of Lemma 1-4}
\label{sec:lemma1-4}
\subsection{Proof of Lemma 1}
\begin{proof}
Decompose the difference between $\tilde{I}(T_n, \mathbb{P})$ and $I^*$ as:
$$\tilde{I}(T_n, \mathbb{P}) - I^* = \tilde{I}(T_n, \mathbb{P})-I(T_n, \mathbb{P}) + I(T_n, \mathbb{P}) - I^*.$$
The first term is nonnegative by definition. Now we show the second term is also nonnegative. Let $A_1, A_2, \ldots A_m$ be all the leaf nodes of tree $T_n$. Then the impurity can be computed as
\begin{align*}
I(T_n, \mathbb{P}) = & \sum_{l=1}^m 2\frac{\int_{A_l} \eta(x)d\mathbb{P}(x)\cdot \int_{A_l}(1-\eta(x))d\mathbb{P}(x) }{\mathbb{P}(A_l)}.
\end{align*}
Thus we have
\begin{align*}
I(T_n, \mathbb{P}) - I^* = & \sum_{l=1}^m \left[ 2\frac{\int_{A_l}\eta(x)d\mathbb{P}(x)\cdot \int_{A_l}(1-\eta(x))d\mathbb{P}(x) }{\mathbb{P}(A_l)} - \int_{A_l} 2\eta(x)(1-\eta(x))d\mathbb{P}(x)  \right] \\
        = & \sum_{l=1}^m \frac{2}{\mathbb{P}(A_l)}\left[ \int_{A_l}\eta(x)d\mathbb{P}(x)\cdot \int_{A_l}(1-\eta(x))d\mathbb{P}(x) - \int_{A_l} \eta(x)(1-\eta(x))d\mathbb{P}(x) \cdot \mathbb{P}(A_l)\right] \\
        = & \sum_{l=1}^m \frac{2}{\mathbb{P}(A_l)}\left[ -\Big(\int_{A_l}\eta(x)d\mathbb{P}(x)\Big)^2 + \int_{A_l} \eta^2(x)d\mathbb{P}(x) \cdot \mathbb{P}(A_l)\right] 
        \ge 0,
\end{align*}
where the last inequality is obtained by Jensen's inequality. Therefore, $\tilde{I}(T_n, \mathbb{P}) \to I^*$ in probability implies
\begin{equation}\label{sign impu 2 impu}
\tilde{I}(T_n, \mathbb{P}) - I(T_n, \mathbb{P}) \to 0 \text{ in probability and} \end{equation}
\begin{equation}\label{impu 2 oracle}
I(T_n, \mathbb{P}) - I^* \to 0 \text{ in probability}.
\end{equation}
We first consider (\ref{impu 2 oracle}). Denote $\psi_n(x)$ as the leaf node of $T_n$ that contains $x$. Define $\bar \eta_{n}(x)$ as
$$\bar \eta_n(x) = \frac{\int_{\psi_n(x)}p(t) d\mathbb{P}(t)}{\mathbb{P}(\psi_n(x))},$$
the average value of $\eta(x)$ at the leaf node of $T_n$ that contains $x$. Then we can rewrite $I(T_n, \mathbb{P}) - I^*$ as
$$I(T_n, \mathbb{P}) - I^* = \sum_{l=1}^m 2 \int_{A_l} (\eta^2(x) - \bar \eta_n^2(x)) \mathbb{P}(x) = 2 \int_{\Omega} (\eta^2(x) - \bar \eta_n^2(x)) \mathbb{P}(x).$$
Therefore (\ref{impu 2 oracle}) implies $\forall \epsilon>0$, 
$$
\lim_{n\to\infty}\mathbb{P}\left(X\in\Omega:|\eta(X) - \bar{\eta}_{n}(X)|>\epsilon\right) = 0.$$
For any $\epsilon>0$, denote $A_\epsilon\subset \Omega$ as: $A_\epsilon = \{X\in \Omega: |\eta(X) - 1/2| >\epsilon\}$. Then we have
\begin{equation}\label{p barp 2}
\lim_{n\to\infty}\mathbb{P}\left(\{|\eta(X) - \bar{\eta}_{n}(X)|\le\epsilon\}\cap \ A_{2\epsilon}\right) = \mathbb{P}(A_{2\epsilon}).
\end{equation}
Define a classifier: $\tilde{f}_n: \Omega\to\{0,1\}$ such that:
$$\tilde{f}_n(x) = \Bigg\{ \begin{aligned} 1, & \quad \mathbb{P}(Y=1|X\in\psi_n(x))\ge 1/2, \\
0, & \quad \mr{otherwise}. \end{aligned}$$
That is, $\tilde{f}_n$ achieves the minimal error $R(\tilde{f}_n)$ among all the classifiers that are piecewise constant on all leaf nodes of $T_n$. Because $|\eta(X)-1/2|>2\epsilon$, $|\eta(X) - \bar{\eta}_{n}(X)|\le \epsilon$ implies: 1. $\eta(X)-1/2$ has the same sign as $\bar{\eta}_{n}(X)-1/2$; 2. $|\bar{\eta}_{n}(X)-1/2|\ge\epsilon$. So by (\ref{p barp 2}), we have for any oracle classifier $f^*\in F^*$,
\begin{equation}\label{p barp 3}
\lim_{n\to\infty}\mathbb{P}\left(\{\tilde{f}_n(X) = f^*(X)\} \cap\{|\bar{\eta}_{n}(X)-1/2|\ge\epsilon\} \cap A_{2\epsilon}\right) = \mathbb{P}(A_{2\epsilon}).
\end{equation}
Denote the $B_\epsilon = \{\tilde{f}(X) = f^*(X)\} \cap\{|\bar{\eta}_{n}(X)-1/2|\ge\epsilon\} \cap A_{2\epsilon}$.
We then consider (\ref{sign impu 2 impu}). We have
\begin{align}\label{sign impu 2 impu 1}
\tilde{I}(T_n, \mathbb{P}) - I(T_n, \mathbb{P}) = & \int_\Omega |1-4\bar{\eta}_{n}(x)(1-\bar \eta_{n}(x))| \mathbbm{1}_{\{f_n(x) \ne \tilde{f}_n(x)\}} d\mathbb{P}(x) \nonumber\\
    \ge & \int_{B_\epsilon} |1-4\bar{\eta}_{n}(x)(1-\bar \eta_{n}(x))| \mathbbm{1}_{\{f_n(x) \ne \tilde{f}_n(x)\}} d\mathbb{P}(x) \nonumber\\
    \ge & \int_{B_\epsilon} 4\epsilon^2 \mathbbm{1}_{\{f_n(x) \ne f^*(x)\}} d\mathbb{P}(x) 
\end{align}
Combining (\ref{sign impu 2 impu}) and (\ref{sign impu 2 impu 1}), we have 
\begin{equation}\label{sign impu 2 impu 2}
\lim_{n\to\infty} \mathbb{P}(B_\epsilon\cap \{f_n(X) \ne f^*(X)\}) = 0.
\end{equation}
Combining equation (\ref{p barp 3}) and (\ref{sign impu 2 impu 2}), we have
$$\lim_{n\to\infty}\mathbb{P}\left( \{f_n(X) \ne f^*(X)\} \cap A_{2\epsilon} \right) = 0.$$
Since $\epsilon$ is arbitrary, and $\lim_{\epsilon\to 0}\mathbb{P}(A_{2\epsilon}) = \mathbb{P}(\{\mathbb{E}(Y|X)\ne 1/2\}) = 1$, we have
$$\lim_{n\to\infty}\mathbb{P}\left( \{f_n(X) \ne f^*(X)\} \right) =0.$$
Therefore 
$$\lim_{n\to\infty}\mathbb{E}|f_n-f^*| = 0.$$
\end{proof}

\subsection{Proof of Lemma 2}
\begin{proof}
Since $\mathscr{T}_n$ is a finite set, we can find $\tilde{T}_n\in\mathscr{F}_n$ satisfying 
$$|\tilde{I}(\tilde{T}_n, \mathbb{P})+\lambda_n r(\tilde{T}_n) - I^*| = \inf_{T\in\mathscr{T}_n} |\tilde{I}(T, \mathbb{P})+\lambda_n r(T) - I^*|.$$
Therefore, we have
\begin{align*}
\tilde{I}(\hat{T}_n,\mathbb{P})+\lambda_n r(\hat{T}_n) - I^*  
		& = [\tilde{I}(\hat{T}_n,\mathbb{P})+\lambda_n r(\hat{T}_n)] - 
		    [\tilde{I}(\tilde{T}_n,\mathbb{P})+\lambda_n r(\tilde{T}_n)] + 
		    [\tilde{I}(\tilde{T}_n,\mathbb{P})+\lambda_n r(\tilde{T}_n)] - I^*.
\end{align*}
For the first two terms, 
\begin{align*}
[\tilde{I}(\hat{T}_n,\mathbb{P})+\lambda_n r(\hat{T}_n)] - 
		    [\tilde{I}(\tilde{T}_n,\mathbb{P})+\lambda_n r(\tilde{T}_n)]
		= &  [\tilde{I}(\hat{T}_n,\mathbb{P})+\lambda_n r(\hat{T}_n)] - 
		    [\tilde{I}(\hat{T}_n,\mathbb{P}_n)+\lambda_n r(\hat{T}_n)]  \\
		& + [\tilde{I}(\hat{T}_n,\mathbb{P}_n)+\lambda_n r(\hat{T}_n)] - 
		    [\tilde{I}(\tilde{T}_n,\mathbb{P}_n)+\lambda_n r(\tilde{T}_n)]  \\
		& + [\tilde{I}(\tilde{T}_n,\mathbb{P}_n)+\lambda_n r(\tilde{T}_n)] - 
		    [\tilde{I}(\tilde{T}_n,\mathbb{P})+\lambda_n r(\tilde{T}_n)] \\
		\le & 2 \sup_{T\in\mathscr{T}_n}|\tilde{I}(T, \mathbb{P}) - \tilde{I}(T,\mathbb{P}_n)|,
\end{align*}
where we use the fact that $\hat{T}_n$ is the minimizer of $\tilde{I}(T,\mathbb{P}_n)+\lambda_n r(T)$ for all $T\in\mathscr{T}_n$. Recalling the definition of $\tilde{T}_n$, by the triangle inequality, we have
$$\tilde{I}(\hat{T}_n,\mathbb{P})+\lambda_n r(\hat{T}_n) - I^* \le 2 \sup_{T\in\mathscr{T}_n}|\tilde{I}(T, \mathbb{P}) - \tilde{I}(T,\mathbb{P}_n)| + \inf_{T\in\mathscr{T}_n}|\tilde{I}(T, \mathbb{P})+\lambda_n r(T) - I^*|.$$
Since $\lambda_n r(\hat{T}_n)\ge 0$,  we have 
$$\tilde{I}(\hat{T}_n,\mathbb{P}) - I^* \le 2 \sup_{T\in\mathscr{T}_n}|\tilde{I}(T, \mathbb{P}) - \tilde{I}(T,\mathbb{P}_n)| + \inf_{T\in\mathscr{T}_n}|\tilde{I}(T, \mathbb{P})+\lambda_n r(T) - I^*|.$$
\end{proof}

\subsection{Proof of Lemma 3}

Let $A_1, A_2, \ldots A_m (m\le \bar{a}_n)$ be all the leaf nodes of $T$. Define $\eta(A_j)=\mathbb{E}(Y|X\in A_j)$, $\eta_{n}(A_j) = \mathbb{E}(Y|X\in A_j)$. We first introduce a technical lemma  which says the maximal difference of $|\eta(A_j)-\eta_{n}(A_j)|$ over all leaf nodes goes to zero in probability.
\begin{tech lemma}\label{sup leaf difference}
If $\frac{\bar{a}_n d\log n}{n} = o(1)$,
then $\forall \epsilon>0$, regardless whether optional steps of Algorithm 1 are enabled, we have
$$\lim_{n\to\infty}\mathbb{P}\left(\sup_{1\le j\le m}|\mathbb{P}(A_j) - \mathbb{P}_{n}(A_j)|>\epsilon\right) = 0,$$
$$\lim_{n\to\infty}\mathbb{P}\left(\sup_{1\le j\le m}|\eta(A_j)-\eta_{n}(A_j)|>\epsilon\right) = 0.$$
\end{tech lemma}
\begin{proof}[Proof of Technical Lemma \ref{sup leaf difference}]
This technical lemma is a special case of Lemma 3 of \cite{nobel1996histogram}. We let the feature space be $\Omega_0 = \Omega\times\{0,1\}\subset \mathbb{R}^{d+1}$. Let $\Pi_n$ be the collection of trees that agree with  tree $T_n$ in the first $d$ dimensions and do not split on the last dimension. Define three real valued functions $g_0, g_1, g_2: \Omega\to\mathbb{R}$ such that $g_0(x) = 0, g_1(x) = 1, g_2(x) =x_{d+1}, \forall x\in\Omega_0\subset\mathbb{R}^{d+1}$. Let $\mathscr{G}= \{g_0, g_1, g_2\}$. It suffices to verify two conditions in Lemma 3 of \cite{nobel1996histogram}.

Since $\Pi_n$ is a subset of the collection of all binary trees with at most $\bar{a}_n$ leaf nodes, the verification process is the same as the analysis in example 3 of \cite{nobel1996histogram}. Noting each element of $\Pi_n$ has at most $\bar{a}_n$ leaf nodes, the restricted cell count is bounded by $m(\Pi_n:V )  \le \bar{a}_n = o(n)$. Each leaf node of a tree in $\Pi_n$ is determined by at most $\bar a_n$ splits and hence at most the intersection of $\bar a_n$ half spaces. Since a single half space can partition $n$ points in at most $n^d$ different ways, the elements of of $\Pi_n$ can induce at most $n^{d\bar{a}_n}$ different splits on $n$ points. Therefore
$\log \Delta^*_n(\Pi_n) = \log(n^{(m-1)(d+1)}) \le \log(n^{\bar{a}_n (d+1)}) = o(n)$.
\end{proof}

\begin{proof}[Proof of Lemma 3]
For all $\epsilon\in (0,1)$, by Technical Lemma \ref{sup leaf difference}, the event 
$$H_0 = \left\{\sup_{1\le j\le m}|\mathbb{P}(A_j) - \mathbb{P}_{n}(A_j)|\le\epsilon, \sup_{1\le j\le m}|\eta(A_j)-\eta_{n}(A_j)|\le\epsilon\right\}$$
holds with probability tending to one. Therefore it suffices to prove $\sup_{T\in\mathscr{T}_n}|\tilde{I}(T, \mathbb{P}) - \tilde{I}(T,\mathbb{P}_n)| \to 0$ under event $H_0$. For all $T\in\mathscr{T}_n$, define two collection of leaf nodes $\mathscr{A}$, $\mathscr{B}$ as:
$$\mathscr{A} = \{A_j, 1\le j\le m: \tilde{I}(A_j,\mathbb{P}) - \tilde{I}(A_j,\mathbb{P}_n) = I(A_j,\mathbb{P}) - I(A_j,\mathbb{P}_n)\},$$
$$\mathscr{B} = \{A_j, 1\le j\le m: |\tilde{I}(A_j,\mathbb{P}) - \tilde{I}(A_j,\mathbb{P}_n)| = |1-I(A_j,\mathbb{P}) - I(A_j,\mathbb{P}_n)|\}.$$
That is, $\mathscr{A}$ contains all leaf nodes where both $\eta(A_j)$ and $\eta_{n}(A_j)$ are no greater than $1/2$ (or both $\eta(A_j)$ and $\eta_{n}(A_j)$ are no less than $1/2$) while $\mathscr{B}$ contains all leaf nodes where one and only one of $\eta(A_j)$ and $\eta_{n}(A_j)$ is less than $1/2$. For all $A_j\in\mathscr{A}$,  we have
\begin{align*}
|\tilde{I}(A_j,\mathbb{P}) - \tilde{I}(A_j,\mathbb{P}_n)| 
        & = |I(A_j,\mathbb{P}) - I(A_j,\mathbb{P}_n)\}|    \\
        & = |2\eta(A_j)(1-\eta(A_j)) - 2\eta_{n}(A_j)(1-\eta_{n}(A_j))| \\
        & = 2|(\eta(A_j)-\eta_{n}(A_j))(1-\eta(A_j)-\eta_{n}(A_j))| \le 2 (\epsilon + \epsilon^2)
\end{align*}
For all $A_j\in\mathscr{B}$, since $\eta_{n}(A_j) < 1/2 < \eta(A_j)$ or $\eta(A_j) < 1/2 < \eta_{n}(A_j)$, recalling $|\eta_{n}(A_j) - \eta(A_j)|<\epsilon$, we have
$|\eta_{n}(A_j) - 1/2|<\epsilon$ and $|\eta(A_j) - 1/2|<\epsilon$. Therefore
\begin{align*}
|\tilde{I}(A_j,\mathbb{P}) - \tilde{I}(A_j,\mathbb{P}_n)| 
    & = |1-I(A_j,\mathbb{P}) - I(A_j,\mathbb{P}_n)| \\
    & \le |1/2-I(A_j,\mathbb{P})| + |1/2 - I(A_j,\mathbb{P}_n)| \\
    & \le |1/2-2\eta(A_j)(1-\eta(A_j)| + |1/2 - 2\eta_{n}(A_j)(1-\eta_{n}(A_j))| \\
    & = \Big|1/2-2[1/2-(1/2-\eta(A_j))][1/2+(1/2-\eta(A_j)]\Big| \\
		& + \Big|1/2-2[1/2-(1/2-\eta_{n}(A_j))][1/2+(1/2-\eta_{n}(A_j)]\Big| \\
    & \le 4\epsilon^2.
\end{align*}
We can finally compute the difference of two signed tree impurities as
\begin{align*}
|\tilde{I}(T,\mathbb{P}) - \tilde{I}(T,\mathbb{P}_n)| 
	\le & \Big|\sum_j \mathbb{P}(A_j)\tilde{I}(A_j,\mathbb{P}) - \sum_j \mathbb{P}_{n}(A_j)\tilde{I}(A_j,\mathbb{P}_n) \Big| \\
	\le & \sum_j [\mathbb{P}(A_j)|\tilde{I}(A_j,\mathbb{P}) - \tilde{I}(A_j,\mathbb{P}_n)| + |\mathbb{P}(A_j) - \mathbb{P}_{n}(A_j)|] \\
	\le & m [\max(2\epsilon+2\epsilon^2, 4\epsilon^2) + \epsilon] \\
	\le & (3+2\epsilon)m\epsilon.
\end{align*}
Since the above equation holds for all $T\in\mathscr{T}_n$ and $\epsilon\in (0,1)$, $\sup_{T\in\mathscr{T}_n}|\tilde{I}(T, \mathbb{P}) - \tilde{I}(T,\mathbb{P}_n)|$ goes to zero in probability.
\end{proof}

\subsection{Proof of Lemma 4}

Since it is not easy to directly compute the approximation error, we introduce a ``theoretical tree'' as a bridge connecting our estimator $\hat{T}_n$ and the oracle lower bound $T^*$. Algorithm \ref{theoretical tree algorithm} describes how we construct theoretical trees. 

\setcounter{algocf}{1}
\begin{algorithm}[H]
\SetAlgoLined
\KwResult{Output the fitted tree}
Input distribution $\mathbb{P}$, impurity function $f(\cdot)$,  weight for minority class $\alpha$, and maximal number of leaf nodes $k\in\mathbb{N}$. Set SVR penalty parameter $\lambda = 0$. Let the root node be $\Omega$, and $\mr{node.X} = \Omega$\;
\While{$\mr{node\_queue}$ is not empty and number of leaf nodes $\le k$}{	
		Dequeue the first entity in $\mr{node\_queue}$, denoting it as $\mr{node}$\; 
		\For{$j$ in $1:d$}{
			Find the best split $\hat{x}_j$ and its corresponding class label assignments inside the current node, such that if we divide $\mr{node}$ into two nodes $X[j] \le \hat{x}$ and $X[j]>\hat{x}$ and assign class labels to two left, right child node as $\mr{lab}^l_j$, $\mr{lab}^r_j$, tree impurity is minimized\;
        }
		Find the best $\hat{x}_j$ among $\hat{x}_1, \hat{x}_2, \cdots \hat{x}_d$ that minimizes tree impurity. Denote it as $\hat{x}$, its class label assignments for left and right child node as $\mr{lab}^l$, $\mr{lab}^r$, respectively\;	
		\eIf{tree impurity is decreased}{
		Let $\mr{node.left.X} =\{X\in \mr{node}: X[j]\le \hat{x}\}$, and $\mr{node.right.X} = \{X\in \mr{node}: X[j] > \hat{x}\}$. Assign $\mr{node.Y}$ to $\mr{node.left}$,  and $\mr{node.right}$ according to the assignment of $\mr{node.X}$. Let the class label of $\mr{node.right}$ and $\mr{node.left}$ be $\mr{lab}^l$, $\mr{lab}^r$, respectively. Enqueue $\mr{node.left}$, $\mr{node.right}$ to the end of $\mr{node\_queue}$\;}{
		Reject the split\;}
}
\caption{Steps for building the theoretical tree}\label{theoretical tree algorithm}
\end{algorithm}

Theoretical trees are computed with SVR penalty parameter $\lambda=0$, because no regularization is needed if we know the true distribution. Consequently, the theoretical tree always ``correctly'' assigns class labels, resulting in $I(T^*_k, \mathbb{P}) = \tilde{I}(T^*_k, \mathbb{P})$. For simplicity, in all the following proofs, we assume for each $k\in\mathbb{N}$, the theoretical tree $T^*_k$ is unique. Our proofs can be easily generalized to the case that there exist multiple theoretical trees, where the distance between SVR-Tree and theoretical tree $T^*_k$ is replaced with the infimum between SVR-Tree and all theoretical trees with $k$ leaf nodes.\par{}
The proof of Lemma 4 is mainly built on three technical lemmas. The first one shows the variation of $\eta(X)$ in each leaf nodes goes to zero and the sequence of theoretical trees is consistent, the second one shows two trees with similar structures are also close in impurity, and the last one shows as $n$ goes to infinity, theoretical trees and our estimated tree are close in their splits structures. We begin with the consistency of theoretical trees. $\forall x \in \Omega$, let $\psi_k(x)$ be the leaf node of $T^*_k$ that contains $x$, then we have $\forall k_1<k_2$, $\psi_{k_2}(x)\subset\psi_{k_1}(x)$. Define $\Delta\eta_k(x)$ as the maximal variation of $\mathbb{E}(Y|X)$ in $\psi_k(x)$: 
$$\Delta\eta_k(x) = \sup_{X\in \psi_k(x)}\mathbb{E}(Y|X) - \inf_{X\in\psi_k(x)}\mathbb{E}(Y|X).$$

\begin{tech lemma}\label{theoretical consistent}
The sequence of theoretical trees $T^*_1, T^*_2, \ldots$ satisfies
\begin{equation}\label{eq variation leaf}
    \lim_{k\to\infty}\Delta\eta_k(x) = 0, \;\; \forall x\in \Omega,
\end{equation}
\begin{equation}\label{eq theoretical consistent}
    \lim_{k\to\infty} \tilde{I}(T_k^*, \mathbb{P}) = I^*.
\end{equation}
\end{tech lemma}

\begin{proof}[Proof of Technical Lemma \ref{theoretical consistent}]
\hfill
\paragraph{Step 1}
We first prove equation (\ref{eq variation leaf}). 
% The proof is decomposed into 4 steps.
% \paragraph{Step 1} 
For a leaf node $A$, we denote the impurity decrease of node $A$ after a split at $X[j]=z$ as $\Delta I(A, (z,j),\mathbb{P})$. That is, 
$$\Delta I(A, (z,j),\mathbb{P}) = I(A,\mathbb{P}) - \mathbb{P}(X[j]<z|X\in A)I(A\cap \{X[j]<z\},\mathbb{P}) -  \mathbb{P}(X[j]\ge z|X\in A)I(A\cap \{X[j]\ge z\},\mathbb{P}).$$
We claim for all set $A$ of the shape $A=[a_1,b_1]\times[a_2,b_2]\times\cdots\times[a_k,b_k]$ with $a_j\le b_j, \forall 1\le j \le d$, if $\Delta I(A,(z,j),\mathbb{P})=0, \;\forall 1\le j\le d\;\text{and}\;a_j\le z\le b_j$, then $\eta(X)$ is constant on $A$. 

By definition of impurity, if $\Delta I(A,(z,j),\mathbb{P})=0$, then $I(A,\mathbb{P}) = I(A\cap \{X[j]<z\},\mathbb{P}) = \mathbb{P}(X[j]\ge z|X\in A)I(A\cap \{X[j]\ge z\},\mathbb{P})$, thus $\mathbb{E}(Y|X) = \mathbb{E}(Y|A\cap \{X[j]< z\}) = \mathbb{E}(Y|A\cap \{X[j]\ge z\})$. Since this holds for all $(z,j)$, recall $\eta_{A,j}=\mathbb{E}(Y|X[j],X\in A)$, we have $\eta_{A,j}$ is constant for all $1\le j\le d$. By Condition 1, $\eta(X)$ is constant on $A$.

For all $x\in\Omega$, denote $\psi_k(x)$ as $[a_{k,1}(x), b_{k,1}(x)]\times[a_{k,2}(x), b_{k,2}(x)]\times\cdots\times[a_{k,d}(x), b_{k,d}(x)]$. Since $\psi_k(x)$ is a sequence of decreasing sets, denote its limit as $\psi_\infty(x) = [a_{\infty,1}(x), b_{\infty,1}(x)]\times[a_{\infty,2}(x), b_{\infty,2}(x)]\times\cdots\times[a_{\infty,d}(x), b_{\infty,d}(x)]$.
If $\max_{1\le j\le d}|b_{\infty,j}(x)-a_{\infty,j}(x)|=0$, then by continuity of $\eta(X)$, $\lim_{k\to\infty}\Delta\eta_k(x)=0$. So it suffices to consider the case where $\max_{1\le j\le d}|b_{\infty,j}(x)-a_{\infty,j}(x)|>0$. If $\Delta I(\psi_\infty(x),(z,j),\mathbb{P})=0$ for all $1\le j\le d$ and $a_{\infty,j}\le z \le b_{\infty,j}$, by the claim above $\eta(X)$ is constant in $\psi_\infty(x)$ and hence $\lim_{k\to\infty}\Delta\eta_k(x)=0$. Otherwise, there exists $\delta_0>0$, $1\le j_0\le d$ and $a_{\infty,j_0}<z_0<b_{\infty,j_0}$, such that $I(\psi_\infty(x),(z_0,j_0),\mathbb{P})\ge \delta_0$. 
Since $\eta(x), \rho(x)$ are continuous on a compact set $\Omega=[0,1]^d$ and thus uniformly continuous, for all $\epsilon>0$, $\exists \epsilon'>0$, such that
$\forall x_1, x_2\in\Omega$ with $|x_1-x_2|<\epsilon'$, 
\begin{equation} \label{eq continuity}
|\eta(x_2)-\eta(x_1)|<\epsilon \text{ and } |\rho(x_2)-\rho(x_1)|<\epsilon\rho_{\min}.
\end{equation}
Because $\psi_\infty(x)$ is the limit set of $\psi_k(x)$, $\exists K\in\mathbb{N}$, such that $\forall k \ge K$, $\max_{1\le j\le d}|a_{k,j}-a_{\infty,j}|<\epsilon'$ and $\max_{1\le j\le d}|b_{k,j}-b_{\infty,j}|<\epsilon'$. We now show for $k\ge K$, $\Delta I(\psi_k(x),(z_0,j_0),\mathbb{P})$ is sufficiently close to $\Delta I(\psi_\infty(x),(z_0,j_0),\mathbb{P})$.
Define a sequence of intermediate sets between $\psi_k(x)$ and $\psi_\infty(x)$ as 
$$A'_j=[a_{k,1}(x), b_{k,1}(x)]\times[a_{k,2}(x), b_{k,2}(x)]\times\cdots\times[a_{k,j-1}(x),b_{k,j-1}(x)]\times[a_{\infty,j}(x),b_{\infty,j}(x)]\times\cdots\times[a_{\infty,d}(x), b_{\infty,d}(x)].$$
That is, $A'_j$ is a hyperrectangle shaped set that is the same as $\psi_k(x)$ at component $1, 2, \ldots, j-1$ and the same as $\psi_\infty(x)$ at components $j, j+1, \ldots, d$. $A'_d$ is just $\psi_k(x)$ and we define $A'_0=\psi_\infty(x)$. We further denote the subset of $A'_j$ to the left and right side of $X[j_0]=z_0$ as $A'_{j,l} \triangleq \{X\in A'_j: X[j_0]<z_0\}$ and $A'_{j,r} \triangleq \{X\in A'_j: X[j_0]\ge z_0\}$, respectively.
We have the following decomposition
$$|\Delta I(\psi_k(x),(z_0,j_0),\mathbb{P})-\Delta I(\psi_\infty(x),(z_0,j_0),\mathbb{P})|=\sum_{j=0}^{d-1} \left|\Delta I(A'_j,(z_0,j_0),\mathbb{P})-\Delta I(A'_{j+1},(z_0,j_0),\mathbb{P}) \right|.$$
We first bound $|\Delta I(A'_{j_0-1},(z_0,j_0),\mathbb{P})-\Delta I(A'_{j_0},(z_0,j_0),\mathbb{P})|$.
Noting $\rho(x)$ is uniformly continuous on compact set $\Omega$, there exists $\rho_{\max}>0$ such that $\max_{x\in\Omega}\rho(x)\le \rho_{\max}$. Therefore we have 
\begin{equation}\label{eq psi A'}
\frac{\mathbb{P}(A'_{j_0}\backslash A'_{j_0-1})}{\mathbb{P}(A'_{j_0-1})} \triangleq \frac{\mathbb{P}(A'_{j_0}\backslash A'_{j_0-1}|A'_{j_0})}{\mathbb{P}(A'_{j_0-1}|A'_{j_0})} \le \frac{2\rho_{\max}\epsilon'}{\rho_{\min}(b_{\infty,j_0}-a_{\infty,j_0})}.
\end{equation}
Here the formula $\frac{\mathbb{P}(A'_{j_0}\backslash A'_{j_0-1})}{\mathbb{P}(A'_{j_0-1})}$ may not be well-defined since both denominator and numerator can be zero, but it is well defined by the ratio of two conditional probabilities as shown above. In the following derivations, we also assume that the ratio of probabilities of two sets is defined by the ratio of conditional probabilities if it is not well-defined itself. The difference between the conditional probability of left leaf node can be bounded as
\begin{align*}
    & \left|\mathbb{P}(A'_{j_0-1,l}|A'_{j_0-1})-\mathbb{P}(A'_{j_0,l}|A'_{j_0}) \right| \\
= & \left|\frac{\mathbb{P}(A'_{j_0-1,l})\mathbb{P}(A'_{j_0}) - \mathbb{P}(A'_{j_0,l})\mathbb{P}(A'_{j_0-1})}{\mathbb{P}(A'_{j_0-1})\mathbb{P}(A'_{j_0})}\right| \\
= & \left|\frac{\mathbb{P}(A'_{j_0-1,l})\mathbb{P}(A'_{j_0}\backslash A'_{j_0-1}) - \mathbb{P}(A'_{j_0,l}\backslash A'_{j_0-1,l})\mathbb{P}(A'_{j_0-1})}{\mathbb{P}(A'_{j_0-1})\mathbb{P}(A'_{j_0})}\right| \\
\le & \frac{\mathbb{P}(A'_{j_0}\backslash A'_{j_0-1})\mathbb{P}(A'_{j_0-1})}{\mathbb{P}(A'_{j_0})\mathbb{P}(A'_{j_0-1})}  \\
\le & \frac{\mathbb{P}(A'_{j_0}\backslash A'_{j_0-1})}{\mathbb{P}(A'_{j_0-1})}
\end{align*}
Noting $Y\in\{0,1\}$, we have
\begin{align*}
 \left|\mathbb{E}(Y|A'_{j_0-1,l}) - \mathbb{E}(Y|A'_{j_0,l})\right|
\le & \frac{\mathbb{P}(A'_{j_0,l}\backslash A'_{j_0-1,l})}{\mathbb{P}(A'_{j_0-1,l})}
\le \frac{\mathbb{P}(A'_{j_0}\backslash A'_{j_0-1})}{\mathbb{P}(A'_{j_0-1,l})}.
\end{align*}
Thus 
\begin{align*}
    & \left|I(A'_{j_0-1,l},\mathbb{P}) - I(A'_{j_0,l},\mathbb{P})\right| 
\le 2\left|\mathbb{E}(Y|A'_{j_0-1,l}) - \mathbb{E}(Y|A'_{j_0,l})\right|
\le 2\frac{\mathbb{P}(A'_{j_0}\backslash A'_{j_0-1})}{\mathbb{P}(A'_{j_0-1,l})}.
\end{align*}
Therefore we have
\begin{align}\label{eq PI left}
    & \left|\mathbb{P}(A'_{j_0-1,l}|A'_{j_0-1})I(A'_{j_0-1,l},\mathbb{P}) - \mathbb{P}(A'_{j_0,l}|A'_{j_0})I(A'_{j_0,l},\mathbb{P})\right| \nonumber \\
={} & \left|\mathbb{P}(A'_{j_0-1,l}|A'_{j_0-1})[I(A'_{j_0-1,l},\mathbb{P})-I(A'_{j_0,l},\mathbb{P})] +  
    [\mathbb{P}(A'_{j_0-1,l}|A'_{j_0-1})-\mathbb{P}(A'_{j_0,l}|A'_{j_0})]I(A'_{j_0,l},\mathbb{P})\right| \nonumber\\
\le{} & \mathbb{P}(A'_{j_0-1,l}|A'_{j_0-1}) \frac{2\mathbb{P}(A'_{j_0}\backslash A'_{j_0-1})}{\mathbb{P}(A'_{j_0-1,l})}
    + 2\frac{\mathbb{P}(A'_{j_0}\backslash A'_{j_0-1})}{\mathbb{P}(A'_{j_0-1})} \nonumber\\
\le{} & 3\frac{\mathbb{P}(A'_{j_0}\backslash A'_{j_0-1})}{\mathbb{P}(A'_{j_0-1})}.
\end{align}
Similarly, we can prove
\begin{equation}\label{eq PI right}
    \left|\mathbb{P}(A'_{j_0-1,r}|A'_{j_0-1})I(A'_{j_0-1,r},\mathbb{P}) - \mathbb{P}(A'_{j_0,r}|A'_{j_0})I(A'_{j_0,r},\mathbb{P})\right| \le 3\frac{\mathbb{P}(A'_{j_0}\backslash A'_{j_0-1})}{\mathbb{P}(A'_{j_0-1})}.
\end{equation}
Since 
$$\left|\mathbb{E}(Y|A'_{j_0}) - \mathbb{E}(Y|A'_{j_0-1})\right| \le \frac{\mathbb{P}(A'_{j_0}\backslash A'_{j_0-1})}{\mathbb{P}(A'_{j_0-1})},$$
we have 
\begin{equation}\label{eq I original}
\left|I(A'_{j_0},\mathbb{P}) - I(A'_{j_0-1},\mathbb{P})\right| \le 2\left|\mathbb{E}(Y|A'_{j_0}) - \mathbb{E}(Y|A'_{j_0-1})\right| \le 2\frac{\mathbb{P}(A'_{j_0}\backslash A'_{j_0-1})}{\mathbb{P}(A'_{j_0-1})}.    
\end{equation}
Combining equations (\ref{eq PI left}), (\ref{eq PI right}), (\ref{eq I original}), (\ref{eq psi A'}) and definitions of impurity decrease, we have
\begin{equation}\label{eq Delta I j0}
\left|\Delta I(A'_{j_0-1},(z_0,j_0),\mathbb{P}) - \Delta I(A'_{j_0},(z_0,j_0),\mathbb{P})\right| \le 8\frac{\mathbb{P}(A'_{j_0}\backslash A'_{j_0-1})}{\mathbb{P}(A'_{j_0-1})} \le \frac{16\rho_{\max}\epsilon'}{\rho_{\min}(b_{\infty,j_0}-a_{\infty,j_0})}.
\end{equation}

We then bound $|\Delta I(A'_{j},(z_0,j_0),\mathbb{P})-\Delta I(A'_{j+1},(z_0,j_0),\mathbb{P})|$ for arbitrary $0\le j\le d-1$ with $j\ne j_0-1$. Since $j\ne j_0-1$, either $\{X[j_0]:X\in A'_j\}=\{X[j_0]:X\in A'_{j+1}\}=[a_{\infty,j_0}(x),b_{\infty,j_0}(x)]$ or $\{X[j_0]:X\in A'_j\}=\{X[j_0]:X\in A'_{j+1}\}=[a_{k,j_0}(x),b_{k,j_0}(x)]$. In both cases, the domain of the $j_0$ component of $A'_j$ and $A'_{j+1}$ are the same interval and we denote it as $[\tilde{a}_j,\tilde{b}_j]$.
% and we denote this interval as $B_j$ with length (one-dimensional Lebesgue measure) $l_j$. 

Denote $X_{else}$ as the collection of features excluding $X[j+1]$ and $X[j_0]$. Noting $\{X[j+1]:X\in A'_j\}=[a_{\infty,j+1},b_{\infty,j+1}]$, $\{X[j+1]:X\in A'_{j+1}\}=[a_{k,j+1},b_{k,j+1}]$ and $|a_{\infty,j+1}-a_{k,j+1}|<\epsilon'$, $|b_{\infty,j+1}-b_{k,j+1}|<\epsilon'$. Recall equation (\ref{eq continuity}), if
$a_{\infty,j+1} = b_{\infty,j+1}$, then $\forall X_{else}, X[j_0]$,
$$\left|\mathbb{E}_{X[j+1]}(Y|X_{else},X[j_0],X\in A'_j)-\mathbb{E}_{X[j+1]}(Y|X_{else},X[j_0],X\in A'_{j+1})\right|\le \epsilon;$$
if $a_{\infty,j+1} < b_{\infty,j+1}$, then $\forall X_{else}, X[j_0]$,
$$\left|\mathbb{E}_{X[j+1]}(Y|X_{else},X[j_0],X\in A'_j)-\mathbb{E}_{X[j+1]}(Y|X_{else},X[j_0],X\in A'_{j+1})\right|\le \frac{2\rho_{\max}\epsilon'}{\rho_{\min}(b_{\infty,j+1}-a_{\infty,j+1})}.$$
Here and in the following equations, $\mathbb{E}_{X[j+1]}$ and $\mathbb{E}_{X_{else}}$ denote expectations with respect to $X[j+1]$ and $X_{else}$, respectively. 
Considering the component-wise conditional expectation on the $j_0$th feature at set $A'_j$ and $A'_{j+1}$, by tower rule of expectations, we have
\begin{align*}
    & \left|\mathbb{E}(Y|X[j_0],X\in A'_j)-\mathbb{E}(Y|X[j_0],X\in A'_{j+1})\right| \\
\le & \mathbb{E}_{X_{else}} \left|\mathbb{E}_{X[j+1]}(Y|X_{else},X[j_0],X\in A'_j)-\mathbb{E}_{X[j+1]}(Y|X_{else},X[j_0],X\in A'_{j+1})\right| \\
\le & \mathbb{E}_{X_{else}}\left[\epsilon\mathbbm{1}_{\{a_{\infty,j+1}=b_{\infty,j+1}\}}+ \frac{2\rho_{\max}\epsilon'\mathbbm{1}_{\{a_{\infty,j+1}<b_{\infty,j+1}\}}}{\rho_{\min}(b_{\infty,j+1}-a_{\infty,j+1})}\right] \\
\le & \left[\epsilon\mathbbm{1}_{\{a_{\infty,j+1}=b_{\infty,j+1}\}}+ \frac{2\rho_{\max}\epsilon'\mathbbm{1}_{\{a_{\infty,j+1}<b_{\infty,j+1}\}}}{\rho_{\min}(b_{\infty,j+1}-a_{\infty,j+1})}\right]
\end{align*}
By exactly the same argument, we have
$$\left|\mathbb{E}(\rho(X)|X[j_0],X\in A'_j)-\mathbb{E}(\rho(X)|X[j_0],X\in A'_{j+1})\right|\le \left[\epsilon\rho_{\min}\mathbbm{1}_{\{a_{\infty,j+1}=b_{\infty,j+1}\}}+ \frac{2\rho_{\max}\epsilon'\mathbbm{1}_{\{a_{\infty,j+1}<b_{\infty,j+1}\}}}{\rho_{\min}(b_{\infty,j+1}-a_{\infty,j+1})}\right].$$
Therefore
$$\left|I(A'_{j,l},\mathbb{P}) - I(A'_{j+1,l})\right| \le 2|\mathbb{E}(Y|A'_{j,l})-\mathbb{E}(Y|A'_{j,l})| \le 2 \left[\epsilon\mathbbm{1}_{\{a_{\infty,j+1}=b_{\infty,j+1}\}}+ \frac{2\rho_{\max}\epsilon'\mathbbm{1}_{\{a_{\infty,j+1}<b_{\infty,j+1}\}}}{\rho_{\min}(b_{\infty,j+1}-a_{\infty,j+1})}\right],$$
\begin{align*}
    & \left|\mathbb{P}(A'_{j,l}|A'_j)-\mathbb{P}(A'_{j+1,l}|A'_{j+1})\right| \\
= & \left|\frac{\int_{\tilde{a}_j}^{z_0} \mathbb{E}(\rho(X)|X[j_0]=z,X\in A'_j)dz }{\int_{\tilde{a}_j}^{\tilde b_j} \mathbb{E}(\rho(X)|X[j_0]=z,X\in A'_j)dz} - \frac{\int_{\tilde{a}_j}^{z_0} \mathbb{E}(\rho(X)|X[j_0]=z,X\in A'_{j+1})dz }{\int_{\tilde{a}_j}^{\tilde b_j} \mathbb{E}(\rho(X)|X[j_0]=z,X\in A'_{j+1})dz}  \right| \\
= & \Bigg|\frac{1}{\int_{\tilde{a}_j}^{\tilde b_j} \mathbb{E}(\rho(X)|X[j_0]=z,X\in A'_j)dz \int_{\tilde{a}_j}^{\tilde b_j} \mathbb{E}(\rho(X)|X[j_0]=z,X\in A'_{j+1})dz} \\
 & \Bigg[\int_{\tilde{a}_j}^{z_0} \mathbb{E}(\rho(X)|X[j_0]=z,X\in A'_j)dz\Big\{\int_{\tilde{a}_j}^{\tilde{b}_j} \mathbb{E}(\rho(X)|X[j_0]=z,X\in A'_{j+1})dz-\int_{\tilde{a}_j}^{\tilde{b}_j} \mathbb{E}(\rho(X)|X[j_0]=z,X\in A'_j)dz\Big\} \\
 & - \Big\{\int_{\tilde{a}_j}^{z_0} \mathbb{E}(\rho(X)|X[j_0]=z,X\in A'_{j+1})dz-\int_{\tilde{a}_j}^{z_0} \mathbb{E}(\rho(X)|X[j_0]=z,X\in A'_j)dz\Big\} \int_{\tilde{a}_j}^{\tilde{b}_j} \mathbb{E}(\rho(X)|X[j_0]=z,X\in A'_j)dz 
\Bigg]
\Bigg| \\
\le & \frac{\epsilon\rho_{\min}\left[(\tilde b_j-\tilde a_j)\int_{\tilde{a}_j}^{z_0} \mathbb{E}(\rho(X)|X[j_0]=z,X\in A'_j)dz + (z_0-\tilde a_j)\int_{\tilde{a}_j}^{\tilde{b}_j} \mathbb{E}(\rho(X)|X[j_0]=z,X\in A'_j)dz\right]   }{\int_{\tilde{a}_j}^{\tilde b_j} \mathbb{E}(\rho(X)|X[j_0]=z,X\in A'_j)dz \int_{\tilde{a}_j}^{\tilde b_j} \mathbb{E}(\rho(X)|X[j_0]=z,X\in A'_{j+1})dz}  \\
\le & 2\epsilon.
\end{align*}
Therefore we have
\begin{align}\label{eq PI left 2}
    & |\mathbb{P}(A'_{j,l}|A'_j)I(A'_{j,l},\mathbb{P}) - \mathbb{P}(A'_{j+1,l}|A'_j)I(A'_{j+1,l},\mathbb{P})| \nonumber\\
= & |\mathbb{P}(A'_{j,l}|A'_j)[I(A'_{j,l},\mathbb{P})-I(A'_{j+1,l},\mathbb{P})] +[\mathbb{P}(A'_{j,l}|A'_j)-\mathbb{P}(A'_{j+1,l}|A'_j)]I(A'_{j+1,l},\mathbb{P})| \nonumber\\
\le & 2\left[\epsilon\mathbbm{1}_{\{a_{\infty,j+1}=b_{\infty,j+1}\}}+ \frac{2\rho_{\max}\epsilon'\mathbbm{1}_{\{a_{\infty,j+1}<b_{\infty,j+1}\}}}{\rho_{\min}(b_{\infty,j+1}-a_{\infty,j+1})}\right] + 2\epsilon.
\end{align}
By similar derivations, we have 
\begin{equation}\label{eq PI right 2}
|\mathbb{P}(A'_{j,r}|A'_j)I(A'_{j,r},\mathbb{P}) - \mathbb{P}(A'_{j+1,r}|A'_j)I(A'_{j+1,r},\mathbb{P})| \le 2\left[\epsilon\mathbbm{1}_{\{a_{\infty,j+1}=b_{\infty,j+1}\}}+ \frac{2\rho_{\max}\epsilon'\mathbbm{1}_{\{a_{\infty,j+1}<b_{\infty,j+1}\}}}{\rho_{\min}(b_{\infty,j+1}-a_{\infty,j+1})}\right] + 2\epsilon,
\end{equation}
\begin{equation}\label{eq I original 2}
|I(A'_{j},\mathbb{P}) - I(A'_{j+1})| \le 2\left[\epsilon\mathbbm{1}_{\{a_{\infty,j+1}=b_{\infty,j+1}\}}+ \frac{2\rho_{\max}\epsilon'\mathbbm{1}_{\{a_{\infty,j+1}<b_{\infty,j+1}\}}}{\rho_{\min}(b_{\infty,j+1}-a_{\infty,j+1})}\right].
\end{equation}
Combining equations (\ref{eq PI left 2}), (\ref{eq PI right 2}), (\ref{eq I original 2}) and the definition of $\Delta I(A'_j,(z_0,j_0),\mathbb{P})$, $\Delta I(A'_{j+1},(z_0,j_0),\mathbb{P})$, we have
\begin{equation}\label{eq Delta I else}
|\Delta I(A'_j,(z_0,j_0),\mathbb{P})-\Delta I(A'_{j+1},(z_0,j_0),\mathbb{P})|\le 6\left[\epsilon\mathbbm{1}_{\{a_{\infty,j+1}=b_{\infty,j+1}\}}+ \frac{2\rho_{\max}\epsilon'\mathbbm{1}_{\{a_{\infty,j+1}<b_{\infty,j+1}\}}}{\rho_{\min}(b_{\infty,j+1}-a_{\infty,j+1})}\right] + 4\epsilon.
\end{equation}
Since equation (\ref{eq Delta I else}) holds for all $j\ne j_0-1$, combining these equations with equation (\ref{eq Delta I j0}), we have for constant $C_1, C_2$ independent of $\epsilon, \epsilon'$, $|\Delta I(\psi_k(x),(z_0,j_0),\mathbb{P})-\Delta I(\psi_\infty(x),(z_0,j_0),\mathbb{P})| \le C_1 \epsilon + C_2 \epsilon'$, implying
\begin{equation}\label{eq Delta final}
\Delta I(\psi_k(x),(z_0,j_0),\mathbb{P}) \ge \delta_0 - C_1\epsilon - C_2 \epsilon'.
\end{equation}
That being said, after the first $K$th splits, one can always find a split at node $\psi_k(x), k\ge K$ with impurity decrease close to a constant, which seems impossible. We will now derive such a contradiction.

Suppose the split with highest impurity decrease on $\psi_k(x)$ is achieved at $(j_1, z_1)$ and splits $\psi_k(x)$ into $\tilde{A}_l$ and $\tilde{A}_r$. Without loss of generality, we assume $x\in \tilde{A}_l$. 
Since $\max_j\|a_{\infty,j}-a_{k,j}\|<\epsilon'$ and $\max_j\|b_{\infty,j}-b_{k,j}\|<\epsilon'$, we have $z_1-a_{k,j_1}<\epsilon'$ or $b_{k,j_1}-z_1<\epsilon'$. 
Without loss of generality, we assume $x\in\mathbb{P}(\tilde{A}_l)$. By similar proof as that for equation (\ref{eq Delta I else}), with $A'_{j+1}$ corresponding to $\psi_k(x)$, $A'_j$ corresponding to $\tilde A_l$ and $(z_0,j_0)$ being some split on the border of node $\psi_k(x)$ with $j_0\ne j_1$ (The equation (\ref{eq Delta I j0}) holds for all $(j_0, z_0)$ as long as $j_0\ne j+1$), we obtain
\begin{equation}\label{eq Delta small}
|\Delta I(\psi_k(x),(z_1,j_1),\mathbb{P}) - \Delta I(\tilde A_l,(z_1,j_1),\mathbb{P})|<
\left[\epsilon\mathbbm{1}_{\{a_{\infty,j_1}=b_{\infty,j_1}\}}+ \frac{2\rho_{\max}\epsilon'\mathbbm{1}_{\{a_{\infty,j_1}<b_{\infty,j_1}\}}}{\rho_{\min}(b_{\infty,j_1}-a_{\infty,j_1})}\right].
\end{equation}

Combining equations (\ref{eq Delta final}) and (\ref{eq Delta small}) and noting $\epsilon, \epsilon'$ can be arbitrarily small while $\delta_0$ is a fixed constant, we achieve a contradiction. Therefore $\Delta(\psi_\infty(x), (z,j),\mathbb{P})$ must be zero for all possible splits $(z,j)$. Hence by the claim proven at the beginning of this proof, we have $\eta(x)$ being constant on $\psi_\infty(x)$. By the uniform continuity of $\eta(x)$ and arbitrarity of $x$, we have proved equation (\ref{eq variation leaf}).

\paragraph{Step 2}
We now prove equation (\ref{eq theoretical consistent}).
By equation (\ref{eq variation leaf}), for all $\epsilon\in(0,1)$, $\exists K\in\mathbb{N}$, $\forall k > K$, $\Delta\eta_k(x) < \epsilon.$
Let $\eta(\psi_k(x)) = \mathbb{E}(Y|X\in\psi(x))$. By the definition of the theoretical tree, we have
\begin{align*}
|I(T_k^*,\mathbb{P}) - I^*| 
    = & |2\eta(\psi_k(x))(1-\eta(\psi_k(x))) -2\eta(x)(1-\eta(x))| d\mathbb{P}(x) \\
  \le & 2|[\eta(\psi_k(x))-\eta(x)][1-\eta(\psi_k(x))-\eta(x)]| d\mathbb{P}(x) \\
    \le & 2\epsilon
\end{align*}
where the second inequality utilizes $|\eta(\psi_k(x)) - \eta(x)| \le \Delta\eta_k(x) \le \epsilon$. 
% \begin{align*}
% |I(T_k^*,\mathbb{P}) - I^*| \le & 2\int_{X_0}|[\eta(\psi_k(x))-\eta(x)][1-\eta(\psi_k(x))-\eta(x)]| d\mathbb{P}(x) +\epsilon
% 	\le 3\epsilon.
% \end{align*}
Because $\tilde{I}(T_k^*,\mathbb{P}) = I(T_k^*,\mathbb{P})$ for theoretical trees, we have $|\tilde{I}(T_k^*,\mathbb{P}) - I^*|<2\epsilon $ for all $k > K$. Since $\epsilon$ is arbitrary, we finish the proof.
\end{proof}

We now consider the relation between tree structures and tree impurity.
\begin{tech lemma}\label{nodes 2 impu}
Let $T, T'$ be two trees both having $k$ leaf nodes and $\mathbb{P}$ be a probability measure. Denote all the leaf nodes of $T$ as $A_1, A_2, \ldots A_k$, all the leaf nodes of $T'$ as $A_1', A_2', \ldots A_k'$. Then if 
$$
\sup_{1\le j\le k} \mathbb{P}(A_j \btr A_j')\le \epsilon  $$
we have
$|I(T, \mathbb{P}) - I(T', \mathbb{P})|\le 5 k\epsilon.$
\end{tech lemma}
\begin{proof}[Proof of Technical Lemma \ref{nodes 2 impu}]
We first consider nodes $A_1$ and $A_1'$. 
\begin{align*}
|I(A_1, \mathbb{P}) - I(A_1\cup A_1', \mathbb{P})| 
    = & |2p(A_1)(1-p(A_1)) - 2p(A_1\cup A_1')(1-p(A_1\cup A'))| \\
    = & 2|(p(A_1)-p(A_1\cup A_1'))(1 - p(A_1) - p(A_1\cup A_1'))| \\
    \le & 2|p(A_1)-p(A_1\cup A_1')| \\
    \le & 2\frac{\mathbb{P}(A_1 \btr (A_1\cup A_1'))}{\mathbb{P}(A_1)}
    \le \frac{2\epsilon}{\mathbb{P}(A_1)}.
\end{align*}
Similarly,  we have $|I(A_1', \mathbb{P}) - I(A_1\cup A_1', \mathbb{P})| \le 2\epsilon/\mathbb{P}(A_1')$. Therefore
\begin{align*}
|I(A_1, \mathbb{P})\mathbb{P}(A_1) - I(A_1', \mathbb{P})\mathbb{P}(A_1')| 
        = & \Big|[I(A_1, \mathbb{P}) - I(A_1\cup A_1', \mathbb{P})]\mathbb{P}(A_1) - \\
            & [I(A_1', \mathbb{P}) - I(A_1\cup A_1', \mathbb{P})]\mathbb{P}(A_1') + [\mathbb{P}(A_1) - \mathbb{P}(A_1')]I(A_1\cup A_1', \mathbb{P})\Big| \\
        \le & 2\epsilon + 2\epsilon + \epsilon 
        \le 5\epsilon.
\end{align*}
Therefore, the difference in tree impurity can be computed as
\begin{align*}
|I(T,\mathbb{P}) - I(T',\mathbb{P})|
    = & |\sum_{l=1}^k [I(A_l, \mathbb{P})\mathbb{P}(A_l) - I(A_l', \mathbb{P})\mathbb{P}(A_l')]| \\
    \le & \sum_{l=1}^k |[I(A_l, \mathbb{P})\mathbb{P}(A_l) - I(A_l', \mathbb{P})\mathbb{P}(A_l')]| 
    \le 5k \epsilon
\end{align*}
\end{proof}

We then define a generalized distance metric\footnote{In fact, $D(\cdot, \cdot)$ is a metric on space $\{s(T):T\in\mathscr{F}_k\}$. As a result, $D(s(\cdot), s(\cdot))$ becomes a generalized metric on $\mathscr{F}_k$ where all trees share the same split structures are treated equally in terms of metric.} for tree splits and discuss its properties. For any classification tree $T$ built by Algorithms 1 or 2, let $s(T) = ((x_1, j_1), (x_2, j_2), \ldots, (x_k, j_k))$ denote all the splits of $T$. That is, $T$ is obtained by first spliting $\Omega$ at $X[j_1]=x_1$, then spliting the left child of the root at $X[j_2]=x_2$, the right child of the root at $X[j_3]=x_3$ and so on. If no split is accepted at the $l$th step, we let $x_{l}=0, j_{l}=0$. For two possible splits $(x_l, j_l)$ and $(x_l', j_l')$ at the same step $l$, we define their distance as
$$D((x_l, j_l), (x_l', j_l')) = \max(|x_l-x_l'|, |j_l-j_l'|).$$
Recall that we let the sample space $\Omega = [0,1]^d$, so $D((x_l, j_l), (x_l', j_l'))\ge 1$ if and only if $j_l\ne j_l'$, which means either they split at different features, or one of the $j_l, j_l'$ indicates no split is accepted at that step. If two trees $T, T'$ are both obtained with $k-1, (k\in\mathbb{N})$ steps of splits, we define the distance of $s(T)=((x_1, j_1), (x_2, j_2), \ldots, (x_k, j_k)), s(T)=((x_1', j_1'), (x_2', j_2'), \ldots, (x_k', j_k'))$ as
\begin{equation}\label{splits distance}
D(s(T_1), s(T_2)) = \max_{1\le i\le k-1} D((x_i, j_i), (x_i', j_i')).
\end{equation}

Let $\hat{T}_{n,k}$ be the tree obtained by only the first $k-1$ splits of SVR-Tree $\hat{T}_n$. The following technical lemma shows the distance between $\hat{T}_{n,k}$ and the theoretical tree $T^*_k$ goes to zero as $n$ goes to infinity. 

\begin{tech lemma}\label{hat Tn 2 theoretical}
For all $k\in\mathbb{N}$ and $\delta>0$, 
$$\lim_{n\to\infty}\mathbb{P}\left(D(s(\hat{T}_{n,k}), s(T^*_k))>\delta\right) = 0$$ 
\end{tech lemma}

\begin{proof}[Proof of Technical Lemma \ref{hat Tn 2 theoretical}]
We first prove this technical lemma with optional steps in Algorithm 1 disabled. The proof can be broken into 3 steps. In the first step, we establish the uniform continuity of $I(T,\mathbb{P})$ and $I(T,\mathbb{P}_n)$ with respect to the generalized distance metric $D(s(\cdot), s(\cdot))$; In the second step, we show $I(T,\mathbb{P}_n)$ and $I(T,\mathbb{P})$ are very close as $n$ goes to infinity; The final step utilizes the optimality of $\hat{T}_{n,k}$ in each split step. In the end of the proof, we show the case with optional steps enabled are essentially the same.\par{}

\paragraph{Step 1}
In this step,  we will show for all $\epsilon>0$, $\sigma>0$, $k\in\mathbb{N}$, there exists $\delta\in (0,1)$, such that for any two trees $T, T'$ obtained with $k-1$ splits, if $D(s(T), s(T')) < \sigma$, then 
$$|I(T, \mathbb{P}) - I(T', \mathbb{P})|<\epsilon$$ 
and 
$$\mathbb{P}\left(\sup_{\substack{T,T'\in\mathscr{F}_k \\ D(s(T),s(T'))<\delta}}  
|I(T, \mathbb{P}_n)-I(T',\mathbb{P}_n)|<\epsilon\right) \ge 1-\sigma.$$
\par{}

We assume both $T, T'$ have $k$ leaf nodes, i.e., the split of each step is accepted. If one of $T, T'$ does not have $k$ leaf nodes, then some splits are not accepted. Because $D(s(T), s(T')) < \delta<1$, a split of $T$ is accepted if and only if the corresponding split of $T'$ is accepted. So $T, T'$ still have the same number of leaf nodes, but the number of leaf nodes is smaller than $k$, which still follows the same proof. \par{}

We first consider $|I(T, \mathbb{P}) - I(T', \mathbb{P})|$. Because the marginal probability measure of $\mathbb{P}$ on $\Omega$ is absolutely continuous with respect to the Lebesgue measure $\mu$, for all $\epsilon_0>0$, there exists $\delta_0\in (0,1)$, for any Borel set $A \subset \Omega$ with $\mu(A)<\delta_0$, we have $\mathbb{P}(A)<\epsilon_0$. Let $\delta = \delta_0/d$. Because $D(s(T), s(T')) < \delta$, denote all the leaf nodes of $T$ as $A_1, A_2, \ldots A_k$, there exists an order of all the leaf nodes of $T$: $A_1', A_2', \ldots, A_k'$, such that $\forall 1\le l\le k$, $\mathbb{\mu}(A_l \btr A_l') < d\delta = \delta_0$. Therefore 
$\mathbb{P}(A_l \btr A_l') < \epsilon_0, \;\;\forall 1\le l \le k.$
By Technical Lemma \ref{nodes 2 impu}, we have 
$$|I(T, \mathbb{P}) - I(T',\mathbb{P})|< 5 k \epsilon_0.$$
Since $\epsilon_0$ is arbitrary, let $\epsilon_0 = \epsilon / (5 k)$, we have $|I(T, \mathbb{P}) - I(T',\mathbb{P})|< \epsilon$.

We then consider $|I(T, \mathbb{P}_n) - I(T', \mathbb{P}_n)|$, where we still  show the condition of Technical Lemma \ref{nodes 2 impu} is satisfied, but with probability greater than $1-\sigma$. By absolute continuity of the marginal of $\mathbb{P}$, for all $\epsilon_0>0$, there exists $\delta_0$, for any Borel set $A$ with Lebesgue measure $\mu(A)<\delta_0$, such that we have $\mathbb{P}(A)<\epsilon_0$. Let $m = \lceil \frac{1}{\delta_0}\rceil$. Then we define a collection of subsets of $\Omega$ as
$$\mathscr{A} = \left\{[0,1]\times[0,1]\times \ldots \times\underset{j}{[l/m, (l+1)/m]}\times\ldots \times [0,1] \bigg|j, l\in\mathbb{N}, 1\le j\le d , 1\le l \le m-1\right\}.$$

That is, each $A\in\mathscr{A}$ is a hyperrectangle with $j$th dimension equaling to interval $[[l/m, (l+1)/m]$ and the other $d-1$ dimensions equaling to $[0,1]$. The set $\mathscr{A}$ consists of all such hyperrectangles where $j$ varies from $1$ to $d$ and $l$ varies from $0$ to $m-1$. Since $m$ is usually very large, $\mathscr{A}$ can be viewed as the collection of all $d-1$-dimensional regularly spaced ``slabs'' with thickness $1/m$. These ``slabs'' will help us control the difference of two corresponding leaf nodes of $T$ and $T'$.

Because there are finitely many elements in $\mathscr{A}$, and for each $A\in\mathscr{A}$, $\mathbb{P}_n(A)$ converges to $\mathbb{P}(A)$ almost surely, we have for the prementioned $\epsilon_0$ and all $\sigma>0$, there exists $N_1\in\mathbb{N}$, such that $\forall n> N_1$, 
\begin{equation}\label{Pn P slabs}
\mathbb{P}\left( \sup_{A\in\mathscr{A}} |\mathbb{P}_{n}(A) - \mathbb{P}(A)|\ge \epsilon_0\right) < \sigma,    
\end{equation}
where the outer probability is taken over the distribution of $\mathscr{D}_n$. Combining equation (\ref{Pn P slabs}) and the definition of $\mathscr{A}$, we have 
$$\mathbb{P}(\sup_{A\in\mathscr{A}} \mathbb{P}_n(A)\ge 2\epsilon_0) < \sigma.$$
Define an event $\mathcal{E} = \{\sup_{A\in\mathscr{A}} \mathbb{P}_n(A) < 2\epsilon_0\}$. Then $\mathbb{P}(\mathcal{E})\ge 1-\sigma$. We then consider a larger collection of ``slabs''. Denote the set $\mathscr{B}$ as:
$$\mathscr{B} = \left\{[0,1]\times[0,1]\times \ldots \times\underset{j}{[b_1, b_2]}\times\ldots \times [0,1] \bigg|j\in\mathbb{N}, 1\le j\le d , b_2-b_1\le 1/m\right\}.$$
The collection $\mathscr{B}$ consists of all the ``slabs'' whose thickness is no more than $1/m$. By definition, we have $\forall B\in\mathscr{B}$, $\exists A_1, A_2\in\mathscr{A}$, such that $B\subset A_1 \cup A_2$.
Let $\delta = 1/m \le \delta_0$. For all $T, T'$ obtained by $k-1$ splits and satisfying $D(s(T), s(T')) <\delta$, denoting all leaf nodes of $T$ as $A_1, A_2, \ldots, A_k$, there exists an order of all the leaf nodes of $T'$: $A_1', A_2', \ldots A_k'$, such that $\forall 1\le l\le k$, the symmetrical set difference between $A_1, A_1'$ is contained in the union of $d$ ``slabs'' in the set $\mathscr{B}$. Formally, 
$$A_l\btr A_l' \subset \bigcup_{j=1}^d B_j,\;\; \forall 1\le l\le k$$
where $B_j\in\mathscr{B}$. Since each $B_j\in\mathscr{B}$ is included in the union of two elements of $\mathscr{A}$, we have on event $\mathcal{E}$,

\begin{equation}\label{Pn A}
\mathbb{P}_n(A_l\btr A_l') \le 2d \sup_{A\in\mathscr{A}}\mathbb{P}_n(A) \le 4d\epsilon_0,\;\;\forall 1\le l\le k.
\end{equation}
Combining equation (\ref{Pn A}) with Technical Lemma \ref{nodes 2 impu}, we have on event $\mathcal{E}$,
$$|I(T,\mathbb{P}_n)-I(T',\mathbb{P}_n)| < 20 k d\epsilon_0.$$
Since $\epsilon_0$ is arbitrary, let $\epsilon_0 = \epsilon/(20 kd)$, we have with probability greater than $1-\sigma$
$$\sup_{\substack{T,T'\in\mathscr{F}_k \\ D(s(T),s(T'))<\delta}}  
|I(T, \mathbb{P}_n)-I(T',\mathbb{P}_n)|<\epsilon.$$

\paragraph{Step 2}
By step 1, for all $\epsilon>0,\sigma>0, k\in\mathbb{N}$, let $\mathscr{F}_k$ be the set of all trees obtained by $k-1$ splits, there exists $N_1>0, \delta>0$, $\forall n>N_1$, we have
\begin{equation}\label{theo contin abb}
\sup_{\substack{T,T'\in\mathscr{F}_k \\ D(s(T),s(T'))<\delta}}  
|I(T, \mathbb{P})-I(T',\mathbb{P})|<\epsilon/3
\end{equation}
and
\begin{equation}\label{empi contin abb}
\mathbb{P}\left(\sup_{\substack{T,T'\in\mathscr{F}_k \\ D(s(T),s(T'))<\delta}}  
|I(T, \mathbb{P}_n)-I(T',\mathbb{P}_n)|\ge\epsilon/3\right)<\sigma/2
\end{equation}
Let $\mathscr{G}_k\subset\mathscr{F}_k$ be an $\epsilon/3$-cover on $\mathscr{F}$ with respect to the generalized metric $D(s(\cdot), s(\cdot))$. Because the covering number is no greater than $(d\epsilon/3)^k$, without loss of generality, we let the $\mathrm{card}(\mathscr{G}) =(d\epsilon/3)^k$. Since each element of $\mathscr{G}$ is a fixed tree with at most $k$ leaf nodes and the cardinality of $\mathscr{G}$ is also finite, there exists $N_2\in\mathbb{N}$, $\forall n>N_2$, 
\begin{equation}\label{empi theo close}
\mathbb{P}\left(\sup_{T\in\mathscr{G}_k} |I(T,\mathbb{P})-I(T,\mathbb{P}_n)|>\epsilon/3\right)<\sigma/2
\end{equation}
Combining equations (\ref{theo contin abb}), (\ref{empi contin abb}) and (\ref{empi theo close}) and applying triangle inequality, we have $\forall n > \max(N_1, N_2)$
$$\mathbb{P}\left(\sup_{T\in\mathscr{F}_k} |I(T,\mathbb{P}_n)-I(T,\mathbb{P})|\ge\epsilon\right)<\sigma. $$

\paragraph{Step 3}
We now finish the proof of Technical Lemma 4 by induction on $k$. For $k=1$, there is only one tree: the root node. So the technical lemma holds naturally. Now supposing the technical lemma holds for $1, 2, \ldots, k-1$, 
We then consider it for $k$. 
By step 1, $\forall \epsilon/3$, $\forall \sigma>0$, for $k\in\mathbb{N}$, there exists $\delta_0\in(0,\delta)$, 
\begin{equation}\label{theo k conti}
\sup_{\substack{T,T'\in\mathscr{F}_k \\ D(s(T),s(T'))<\delta_0}}  
|I(T, \mathbb{P})-I(T',\mathbb{P})|<\epsilon/3.
\end{equation}
with probability $1-\sigma$. Since the technical lemma holds for $k-1$, for all $\sigma>0$, for the prementioned $\delta_0$, $\exists N_1>0$, $\forall n> N_1$, we have
\begin{equation}\label{induction assumption}
\mathbb{P}\left( D(s(\hat{T}_{n,k-1}), s(T^*_{k-1})) >\delta_0 \right)<\sigma.
\end{equation}
Denote $\hat{s}_{n,k-1}, s^*_{k-1}$ as the $(k-1)$th split of $\hat{T}_{n,k}, T^*_k$, respectively. Define two auxiliary trees $T_{k-1}^*(\hat{s}_{n,k-1})$, $\hat{T}_{n,k-1}(s^*_{k-1})$. $T_{k-1}^*(\hat{s}_{n,k-1})$ is obtained by applying split $\hat{s}_{n,k-1}$ on theoretical tree $T^*_{k-1}$, while $\hat{T}_{n,k-1}(s^*_{k-1})$ is obtained by applying split $s^*_{k-1}$ on SVR-Tree $\hat{T}_{n,k-1}$. We reassign class labels on all leaf nodes of both $T_{k-1}^*(\hat{s}_{n,k-1})$ and $\hat{T}_{n,k-1}(s^*_{k-1})$. For any leaf node $A$ of $T_{k-1}^*(\hat{s}_{n,k-1})$ or $\hat{T}_{n,k-1}(s^*_{k-1})$, its class label is $1$ if
$\mathbb{P}(Y|X\in A)\ge 1/(1+)$ and $0$ otherwise. By doing this, both $T_{k-1}^*(\hat{s}_{n,k-1})$ and $\hat{T}_{n,k-1}(s^*_{k-1})$ estimates the class label with the true probability distribution, while the split structure is prefixed.
Combining equation (\ref{theo k conti}) and (\ref{induction assumption}), we have with probability greater than $1-\sigma$,
\begin{equation}\label{cotinuity of I at k 1}
|I(\hat{T}_{n,k},\mathbb{P}) - I(\hat{T}^*_{k-1}(\hat{s}_{n,k-1}),\mathbb{P})| < \epsilon/3
\end{equation}
\begin{equation}\label{cotinuity of I at k 2}
|I(\hat{T}_{n,k-1}(s_{k-1}^*),\mathbb{P}) - I(T^*_k,\mathbb{P})| < \epsilon/3
\end{equation}
By the optimality of $\hat{T}_{n,k}$, we have
$$I(\hat{T}_{n,k},\mathbb{P}_n) + \lambda_n r(\hat{T}_{n,k}) \le I(\hat{T}_{n,k-1}(s^*_{k-1}),\mathbb{P}_n) + \lambda_n r(\hat{T}_{n,k-1}(s^*_{k-1})).$$
By Technical Lemma \ref{theoretical consistent}, the sequence of  theoretical trees is consistent. Therefore, there exists $V_0>0, k_0\in\mathbb{N}$, such that $V(T_k) \ge V_0, \forall k>k_0$. Without loss of generality, we assume here $k-1>k_0$. Thus with probability greater than $1-\sigma$, the volume of the theoretical tree satisfies
$$V(T^*_{k-1}) \ge V_0.$$
Therefore, the surface-to-volume ratio of $\hat{T}_{n,k-1}(s^*_{k-1})$ is bounded by:
$$r\left(\hat{T}_{n,k-1}(s^*_{k-1})\right)\le \frac{2dk}{V_0-(k-1)\delta_0}.$$
Therefore, we have
\begin{equation}\label{optimality Tn}
I(\hat{T}_{n,k},\mathbb{P}_n) \le I(\hat{T}_{n,k-1}(s^*_{k-1}),\mathbb{P}_n) + \lambda_n \frac{2dk}{V_0-(k-1)\delta_0}.
\end{equation}
Combining equation (\ref{optimality Tn}) and step 2, there exists $N_2\in\mathbb{N}$, $\forall n>N_2$, such that with probability greater than $1-\sigma$,
\begin{equation}\label{optimality Tn P}
I(\hat{T}_{n,k},\mathbb{P}) \le I(\hat{T}_{n,k-1}(s^*_{k-1}),\mathbb{P}) + \lambda_n \frac{2dk}{V_0-(k-1)\delta_0} + \frac{2}{3}\epsilon.
\end{equation}
Combining equations (\ref{cotinuity of I at k 1}), (\ref{cotinuity of I at k 2}) and (\ref{optimality Tn P}), $\forall n>\max\{N_1, N_2\}$, we have with probability greater than $1-2\sigma$,  
\begin{equation}\label{close to Tk*}
I(T^*_{k-1}(\hat{s}_{n,k}),\mathbb{P}) \le I(T^*_k,\mathbb{P}) + \lambda_n \frac{2dk}{V_0-(k-1)\delta_0} + \frac{4}{3}\epsilon.
\end{equation}
Denote a collection of trees as
$$\mathscr{H}_\varepsilon = \{T\in\mathscr{F}_k: s(T) = (s^*_1, s^*_2, \ldots, s^*_{k-2}, s'), \;s'\;\text{arbitrary},\;I(T, \mathbb{P})-I(T_k^*,\mathbb{P})<\epsilon\}.$$
Therefore, $\mathscr{H}_\varepsilon$ is the collection of trees whose first $k-2$ splits are the same as $T^*_k$ and the tree impurity is within $\varepsilon$ distance of $T^*_k$. Letting $\epsilon_n = \lambda_n \frac{2dk}{V_0-(k-1)\delta_0} + 4\epsilon/3$, equation (\ref{close to Tk*}) implies
\begin{equation}\label{close to Tk* 2}
D(s(T^*_{k-1}(\hat{s}_{n,k})), s(T^*_k)) \le \sup_{T\in\mathscr{H}_{\epsilon_n}} D(s(T),s(T_k^*)).
\end{equation}
Because  $T_k^*$ is the unique theoretical tree with $k$ leaf nodes and $\cap_{\epsilon>0}\mathscr{H}_\epsilon$ is the collection of theoretical trees with $k$ leaf nodes,\footnote{As we have mentioned before, if the theoretical tree is not unique, we can replace $D(s(T),s(T_k^*))$ with the infimum distance to all theoretical trees. Since the space of splits of all theoretical trees forms a closed subspace of $\mathscr{H}_1$, the same argument still holds.} we have
$$
\lim_{\varepsilon\to 0}\sup_{T\in\mathscr{H}_\varepsilon} D(s(T),s(T_k^*)) = 0.$$
This implies for $\delta>0$, there exists $\epsilon_0>0$, such that $\forall \varepsilon<\epsilon_0$, we have $\sup_{T\in\mathscr{H}_\varepsilon} D(s(T),s(T_k^*))<\delta$. Recall equation (\ref{close to Tk* 2}), noting $\lim_{n\to\infty}\lambda_n = 0$, there exists $N_3>0$, such that $\forall n>N_3$, 
$\lambda_n \frac{2dk}{V_0-(k-1)\delta_0} < \epsilon_0/2$. Since $\epsilon$ is arbitrary, we can let $4\epsilon/3<\epsilon_0/2$. Therefore, we have $\forall n>\max\{N_1, N_2, N_3\}$, with probability greater than $1-2\sigma$,
\begin{equation}\label{close-s}
D(s(T^*_{k-1}(\hat{s}_{n,k})), s(T^*_k)) < \delta.
\end{equation}
Noting on the same event where equation (\ref{close-s}) holds, the induction hypothesis (\ref{induction assumption}) also holds. We finally have $\forall n > \max\{N_1,N_2, N_3\}$
$$\mathbb{P}\left(D(s(\hat{T}_{n,k}), s(T^*_k))\ge\delta \right)<2\sigma.$$
Since $\sigma$ is arbitrary, we finish the proof.

\paragraph{If Optional Steps Are Enabled}
Now suppose optional steps in Algorithm 1 are enabled. We will reject a new feature if the impurity decrease at a new feature is no greater than $c_0\lambda_n$ plus the maximal impurity decrease at features that are already split. Because $\lambda_n$ goes to zero as $n$ goes to infinity, for any fixed $k\in\mathbb{N}$, the probability for a split to be rejected by optional steps\footnote{This means, a split is rejected by these optional step but will be accepted as risk minimizer if these steps are not enabled.} in the first $k$ splits goes to zero. Therefore for all $k\in \mathbb{N}$, as $n$ goes to zero, the probability for $\hat{T}_{n,k}$ to remain invariant with respect to optional steps goes to one. This finishes the proof.
\end{proof}

%We are now prepared to prove Lemma 4.
\begin{proof}[Proof of Lemma 4]
For all $\epsilon>0$, by Technical Lemma \ref{theoretical consistent}, there exists $K\in\mathbb{N}$, such that $\forall k > K$, the theoretical tree $T_k^*$ satisfies
\begin{equation}\label{eq theoretical 2 oracle}
|\tilde{I}(T_k^*,\mathbb{P}) - I^*|<\epsilon.
\end{equation}
Because the probability measure $\mathbb{P}$ is absolutely continuous with respect to Lebesgue measure $\mu$, for $\epsilon/5 k>0$, there exists $\delta>0$, such that $\forall A\subset \Omega$, $\mu(A)<\delta$ implies $\mathbb{P}(A)<\epsilon/(5 k)$.
Fix $k> K$. Since $\bar{a}_n\to\infty$, as $n$ goes to infinity, $\hat{T}_n$ has no less than $k$ splits, and $\hat{T}_{n,k}$ is well-defined. By Technical Lemma \ref{hat Tn 2 theoretical}, for $\delta/d>0$, $\forall \sigma>0$, $\exists N\in\mathbb{N}$, $\forall n>N$, we have
$$\mathbb{P}\left(D(s(\hat{T}_{n,k}), s(T_k^*))>\delta/d\right) <\sigma.$$
Thus by Technical Lemma \ref{nodes 2 impu},  we have
$$|I(\hat{T}_{n,k},\mathbb{P})- I(T_k^*,\mathbb{P})|\le 5 k \cdot \epsilon/(5 k) = \epsilon.$$
Let $\tilde{T}_n$ be the tree that has the same splits as $\hat{T}_{n,k}$. On each leaf node $A$ of $\tilde{T}_n$, let the class label be $1$ if and only if $\mathbb{E}(Y|X\in A)\ge /(1+)$. Since $\hat{T}_{n,k}$ is formed by the first $k$ splits of $\hat{T}_n$, $\tilde{T}\in\mathscr{T}_n$. By definition of $\tilde{T}_n$ and $T^*_k$, we have $\tilde{I}(\tilde{T}_n,\mathbb{P}) =I(\tilde{T}_n,\mathbb{P})$ and $\tilde{I}(T^*_k,\mathbb{P}) =I(T^*_k,\mathbb{P})$. Therefore we have
\begin{equation}\label{eq theoretical Tn}
|\tilde{I}(\hat{T}_{n,k},\mathbb{P})- \tilde{I}(T_k^*,\mathbb{P})|\le \epsilon.
\end{equation}
By Technical Lemma \ref{theoretical consistent}, there exists $k_0\in\mathbb{N}$, such that for all $k>k_0$, $V_k^*\ge V_0$. Without loss of generality, we assume $k>k_0$. Combining this with Technical Lemma \ref{hat Tn 2 theoretical}, the surface to volume ratio of $\tilde{T}_n$ can be bounded by
\begin{equation}\label{eq SVR Tn}
r(\tilde{T}_n) = r(\hat{T}_{n,k}) \le \frac{2dk}{V_0 - k\delta} = \frac{2dk}{V_0 - \epsilon/(5d)}.
\end{equation}
Combining equation (\ref{eq theoretical 2 oracle}), (\ref{eq theoretical Tn}) and (\ref{eq SVR Tn}), we have with probability greater than $1-\sigma$,
$$|\tilde{I}(\tilde{T}_n,\mathbb{P}) + \lambda_n r(\tilde{T}_n) - I^*| \le 2\epsilon + \lambda_n \frac{2dk}{V_0 - \epsilon/(5d)}.$$
Since $\epsilon, \sigma$ are arbitrary and $\lim_{n\to\infty}\lambda_n=0$, we finish the proof.
\end{proof}

\section{Proof of Lemma 5-8 and Proposition 1-2}
\label{sec:lemma5-8}

\subsection{Proof of Lemma 5}
\begin{proof}\hfill
\paragraph{Step 1} 
We first prove the Lemma when $\Omega_1(f)$ consists of finite unions of hyperrectangles. That is, $f$ is a tree classifier. Let $a$ be a positive constant no greater than $2/(3\gamma)$ and consider a hypercube $A$ with side length $a$. Define a series of hypercubes $A(t)$ such that $A(t)$ is a hypercube centered at the center of $A$ with side length $2t$. $A(0)$ is a single point at the center of $A$, and $A(a/2)$ is just $A$. We will show if $\partial \Omega_1(f)$ intersects $A$,  $S(\partial\Omega_1(f)\cap A(3a/2)$ has a lower bound.  The idea is to examine the definition of $\mathscr{F}_\gamma$ for all hypercubes $A(t)$. Specifically, for hypercube $A(t)$, we define $\zeta_0(t) = S(\partial A(t) \cap \Omega_0(f))$, $\zeta_1(t) = S(\partial A(t) \cap \Omega_1(f))$, and the border ``length'' between $\partial A(t) \cap \Omega_0$ and $\partial A(t) \cap \Omega_1$ as $l(t) = \mu^{(d-2)}(\partial A(t) \cap \partial \Omega_1(f))$.
Let $\Delta \zeta(t) = |\zeta_0(t) -\lim_{u\to t-} \zeta_0(t)| + |\zeta_0(t) - \lim_{u\to t+}\zeta_0(t)|$. That is, $\Delta\zeta(t)$ is the sum of differences between $\zeta_0(t)$ and its left and right limits. Because $\Omega_1(f)$ is a finite union of hyperrectangles, there are only finitely many values of $t$ such that $\Delta \zeta(t) \ne 0$. Now for the hyperrectangle $A(t)$, the volumes of $\Omega_0(f)$ and $\Omega_1(f)$ intersecting $A(t)$ are 
$$V_0(t) = \int_0^t \zeta_0(u)du, \quad \text{and} \quad V_1(t) = \int_0^t \zeta_1(u)du,$$ 
respectively. The intersecting surface area between $\Omega_0(f)$ and $\Omega_1(f)$ at the interior and the border of $A(t)$ is 
$$S_{\text{int}}(t) = \int_0^t l(u)du + \sum_{u\le t} \Delta \zeta(u).$$
Consider the alternative of letting the whole set $A(t)$ belong to $\Omega_0(t)$. Then the change of surface is $\zeta_1(t) - S_{\text{int}}(t)$ and the change of volume is $V_1(t)$. By definition of $\mathscr{F}_\gamma$, we have $\frac{\zeta_1(t)-S_{\text{int}}(t)}{V_1(t)}\ge -\gamma$, which implies that
\begin{equation}\label{svr cons 1}
\int_0^t l(u)du + \sum_{u\le t} \Delta \zeta(u) \le \zeta_1(t) + \gamma \int_0^t \zeta_1(u)du.
\end{equation}
Similarly, consider the alternative of letting the whole set $A(t)$ belong to $\Omega_1(t)$. The definition of $\mathscr{F}_\gamma$ implies that
\begin{equation}\label{svr cons 2}
\int_0^t l(u)du + \sum_{u\le t} \Delta \zeta(u) \le \zeta_0(t) + \gamma \int_0^t \zeta_0(u)du.
\end{equation}

We show the following claim: If $\zeta_0(t) = \min\{\zeta_0(t), \zeta_1(t)\}$, then $\gamma\int_0^t \zeta_0(u)du - \sum_{u\le t}\Delta \zeta(u) \le t\gamma\zeta_0(t)$. Because $\Omega_1(f)$ is a finite union of hyperrectangles, there exist a finite number of values where $\Delta \zeta(u)$ is nonzero. Denote these values as $0<t_1<t_2< \cdots <t_k<t$. For convenience, we set $t_{k+1}=t$ and $t_0 =0$. Since $\gamma t\le \gamma \cdot \frac{3}{2}a \le 1$, we have
\begin{equation}\label{dist tk}
\gamma\int_0^t \zeta_0(u)du - \sum_{u\le t}\Delta \zeta(u)
\le \gamma\int_0^t \Big[\zeta_0(u)-\sum_{s \ge t}\Delta \zeta(s) \Big] du.
\end{equation}
For all $u\in(0,t)$, there exists $0\le j\le k$, such that $t_j\le u< t_{j+1}$. $\Delta\zeta(u)$ is zero in $(t_j, t_{j+1})$, meaning there is no borders of some hyperrectangles of $\Omega_1(f)$ that overlap with borders $A(t), \;\forall t\in (t_j,t_{j+1})$. Therefore we have $\zeta_0(u)\le \lim_{s\to t_{j+1}-}\zeta_0(s)$. By definition of $\Delta\zeta(t)$, we have 
$$\lim_{s\to t_{j+1}-}\zeta_0(s)\le \min \left\{\lim_{s\to t_{j+1}+}\zeta_0(s), \zeta_0(t_{j+1}) \right\} + \Delta\zeta(t_{j+1}).$$ 
Thus we have 
$$\zeta_0(u)-\Delta\zeta(t_{j+1})\le \min \left\{\lim_{s\to t_{j+1}+}\zeta_0(s),\; \zeta_0(t_{j+1}) \right\}.$$
Repeating the same derivation for the right limit of $t_{j+1}$, we have
$$\zeta_0(u)-\Delta\zeta(t_{j+1})-\Delta\zeta(t_{j+2})\le \min \left\{\lim_{s\to t_{j+2}+}\zeta_0(s), \;\zeta_0(t_{j+2}) \right\}.$$
Since there are a finite number of nonzero points of $\Delta\zeta(t)$,  repeat the above process and we will eventually obtain that
\begin{equation}\label{dist tk 2}
\zeta_0(u)-\sum_{s \ge t}\Delta\zeta(s) \le \zeta_0(t).
\end{equation}
We combine equations (\ref{dist tk}) and (\ref{dist tk 2}), and the claim is proved. With this claim, equation (\ref{svr cons 1}), (\ref{svr cons 2}) and the fact that $t\gamma \le \gamma\cdot \frac{3}{2}a \le 1$, we have
\begin{equation}\label{svr cons 3}
\int_0^t l(u)du \le 2\min\{\zeta_0(t),\zeta_1(t)\}.
\end{equation}

We now need a bridge between $l(t)$ and $\zeta_0(t), \zeta_1(t)$, which is known as the isoperimetric inequality as stated in the following technical lemma.
\begin{tech lemma}\label{isoperimetric}
There exists a constant $c_5$ only dependent on $d$, such that for all $t>0$, 
$$l(t) \ge c_5 [\min\{\zeta_0(t), \zeta_1(t)\}]^{\frac{d-2}{d-1}}.$$
\end{tech lemma}
\begin{proof}[Proof of Technical Lemma \ref{isoperimetric}]
The exact value of $c_5$ is still unknown and the following proof gives a loose lower bound. Consider a $(d-1)$ dimensional sphere $S$ with radius $t$ whose center is the same as the center of the hypercube $A(t)$. We would like to clarify the metric and measure on $S$ and $A(t)$. The space $\partial A(t)$ is equipped with $(d-1)$ dimensional Lebesgue measure and the distance between two points on $\partial A(t)$ is defined as the smallest Euclidean distance of all paths between these two points. The space $S$ is equipped with $(d-1)$ dimensional spherical measure inherited from $d$ dimensional Lebesgue measure of $\mathbb{R}^d$ and the distance between two points on $\partial A(t)$ is also defined as the smallest Euclidean distance of all paths between these two points. In short, measures are the Lebesgue measure and metrics are the Euclidean distance. 
For both $A(t)$ and $S$, the volume of a set $B$ is defined as the measure of $B$ and the surface of a set $B$ is defined in the Minkowski content:
$$S(B) = \lim_{\epsilon\to 0} \frac{V(B_\epsilon) - V(B)}{\epsilon},$$
where $V(B_\epsilon) = \{x: d(B,x)<\epsilon\}$, $V$ and $d$ is the volume and metric of the space $B$ belongs to. Here and in the following proofs, $V(\cdot)$ and $S(\cdot)$ always refer to the surface and volume in the corresponding space. We keep the same notation for both spaces, since they are both derived from Euclidean metric and Lebesgue measure. The definition of surface and volume is consistent with the definition of surface and volume at Section 2.2 of the main paper when the space is $\partial A(t)$ and $B$ is a union of finite hyperrectangles, which are the objects of interest in the proof of this technical lemma.

We then define a bijective mapping $\varphi$ between $\partial A(t)$ and $S$. Denote the center of $S$ as $O$. Then for all $x\in \partial A(t)$, let $\varphi(x)$ be the unique intersection between the line segment $Ox$ and $S$. By definition, $\varphi(x)$ is bijective and continuous, but not differentiable at the ``edges'' of $\partial A(t)$ since $\partial A(t)$ is not a smooth manifold.

For all $\tilde B\subset S$ with $V(\tilde B) \le V(S)/2$, the isoperimetric inequality on spheres \citep{schmidt1943beweis,bogelein2017quantitative} can be written as:
\begin{equation}\label{isoperimetric sphere}
S(\tilde B) \ge c_S V(\tilde B)^{\frac{d-2}{d-1}},    
\end{equation}
where $c_S$ is a constant independent of the radius of sphere $S$. Our objective is to bridge a set on $\partial A(t)$ with a set on $S$ using the bijective mapping $\varphi$ and derive an isoperimetric inequality on $\partial A(t)$ with that on $S$. We first show for all $B\subset \partial A(t)$ consisting of unions of finite hyperrectangles, 
\begin{equation}\label{V ratio}
V(B)/V(\varphi(B))\in [1, (d-1)^{(d-1)/2}].
\end{equation}
Consider a point $x\in \partial A(t)$ and a small hypercube $dx$ satisfying $x\in dx$ and $dx\subset \partial A(t)$. By the definition of $\varphi$, we already have $V(dx) \ge V(\varphi(dx))$. To prove the other direction, denote the normal outward vector of $dx$ and $\varphi(dx)$ as $\bm{n}_1$, $\bm{n}_2$, respectively. As the side $dx$ goes to zero, we have 
$$\lim_{dx\to 0} \frac{\frac{V(\varphi(dx))}{t^{d-1}}}{\frac{V(dx)}{d(O,x)^{d-1}} \cos(\bm{n}_1, \bm{n}_2) } = 1.$$
Noting that $d(O,x) \le \sqrt{d-1}\cdot t$ and $\cos(\bm{n}_1, \bm{n}_2) \ge 1/\sqrt{d-1}$, we have $\lim_{dx\to 0} V(dx)/ V(\varphi(dx)) \le (d-1)^{(d-1)/2}$. Since $B$ consists of unions of finite hyperrectangles, we can divide $B$ into unions of arbitrarily small hyppercubes. Thus we have $V(B)/V(\varphi(B))\in [1, (d-1)^{(d-1)/2}]$. 

We then derive the relation between $S(B)$ and $S(\varphi(B))$. By definition,
\begin{align*}
S(B) = & \lim_{\epsilon\to0} \frac{V(B_\epsilon)-V(B)}{\epsilon} \\
    = & \lim_{\epsilon\to0} \frac{V(B_\epsilon \backslash B)}{\epsilon} \\
    \ge & \lim_{\epsilon\to0} \frac{V(\varphi(B_\epsilon \backslash B))}{\epsilon} 
\end{align*}
By definition of the metric on $\partial A(t)$ and $S$, for all $x_1, x_2 \in \partial A(t)$, $d(x_1, x_2) \ge d(\varphi(x_1), \varphi(x_2))$, we have $\varphi(B)_{\epsilon} \backslash \varphi(B) \subseteq \varphi(B_\epsilon \backslash B) $. Thus we have
\begin{align}\label{S ratio}
S(B) \ge & \lim_{\epsilon\to0} \frac{V(\varphi(B)_{\epsilon} \backslash \varphi(B))}{\epsilon} \nonumber \\
    =& \lim_{\epsilon\to0} \frac{V(\varphi(B)_{\epsilon}) - V(\varphi(B))}{\epsilon} \nonumber \\
    =& S(\varphi(B)).
\end{align}
If $[\zeta_0(t)- \zeta_1(t)][\varphi(\zeta_0(t))-\varphi(\zeta_1(t))]\ge 0$, that is, the smaller volume set in $\partial A(3a/2)$ has a smaller volume when mapped by $\varphi$, we let $B$ be that smaller volume set. In this case, $V(B) = \min\{\zeta_0(t), \zeta_1(t)\}$ and $S(B) = l(t)$. Combining equation (\ref{isoperimetric sphere}), (\ref{V ratio}) and (\ref{S ratio}), we have
$$l(t) \ge \frac{c_S}{(d-1)^{\frac{d-1}{2}}} [\min\{\zeta_0(t), \zeta_1(t)\}]^{\frac{d-2}{d-1}}.$$
Otherwise, if the smaller volume set in $\partial A(3a/2)$ has a greater volume when mapped by $\varphi$, then equation (\ref{V ratio}) implies
$$\max\left\{\frac{\zeta_0(t)}{\zeta_1(t)}, \frac{\zeta_1(t)}{\zeta_0(t)}\right\} \le (d-1)^{d-1}.$$
Therefore let $B$ be the higher volume set of $\partial A(3a/2)$, that is $V(B) = \max\{\zeta_0(t), \zeta_1(t)\}$ and $S(B)=l(t)$, we have 
$$l(t) \ge \frac{c_S}{(d-1)^{\frac{d-1}{2}}} [\max\{\zeta_0(t), \zeta_1(t)\}]^{\frac{d-2}{d-1}} \ge \frac{c_S}{(d-1)^{\frac{d-1}{2}}} [\min\{\zeta_0(t), \zeta_1(t)\}]^{\frac{d-2}{d-1}}.$$
In both cases, the technical lemma is proved with $c_5 = \frac{c_S}{(d-1)^{\frac{d-1}{2}}}$.
\end{proof}

We continue the proof of Lemma 5.
Combining equation \eqref{svr cons 3} and Technical Lemma \ref{isoperimetric}, if $d\ge 3$, we have that
$$
\int_0^t l(u)du \le  2\min \{\zeta_1(t), \zeta_0(t)\} 
        \le  2\left( \frac{l(t)}{c_5} \right)^{\frac{d-1}{d-2}}.$$
To simplify notations, denote $\tilde{S}_{\text{int}}(t) = \int_0^t l(u)du$. Then the inequality above implies that
$$
[\tilde{S}_{\text{int}}(t)]^{\frac{d-2}{d-1}} \le \frac{2^{\frac{d-2}{d-1}}}{c_5} \frac{d}{dt}\tilde{S}_{\text{int}}(t).
$$
Because $\partial \Omega_1(f)$ intersects with the central hypercube $A(a/2)$ of side length $a$, we have $\tilde{S}_{\text{int}}(a/2)>0$. Thus for all $t\ge a/2$, we have
$$
[\tilde{S}_{\text{int}}(t)]^{-1+\frac{1}{d-1}} \frac{d}{dt} \tilde{S}_{\text{int}}(t) \ge \frac{c_5}{2^{\frac{d-2}{d-1}}},
$$
or equivalently,
%Compute the integration, we have
$$\frac{d}{dt} \left([\tilde{S}_{\text{int}}(t)]^{\frac{1}{d-1}}\right) \ge \frac{c_5}{2^{\frac{d-2}{d-1}}(d-1)} .$$
Integrating from $t=a/2$ to $t=3a/2$ and using $\tilde{S}(a/2)>0$, we have in the cases of $d\ge 3$,
$$
S(\partial \Omega_1(f)\cap A(3a/2)) \ge\tilde{S}_{\text{int}}(3a/2) \ge \left[\frac{c_5 a}{2^{\frac{d-2}{d-1}}(d-1)}\right]^{d-1}.
$$
If $d=2$, then by Technical Lemma \ref{isoperimetric}, we have $l(t)\ge c_5$, $\forall t$. Thus
$$
S(\partial \Omega_1(f)\cap A(3a/2)) \ge\tilde{S}_{\text{int}}(3a/2) = \int_{a/2}^{3a/2} a du = c_5 a =\left[\frac{c_5 a}{2^{\frac{d-2}{d-1}}(d-1)}\right]^{d-1}.
$$
So in all cases, we have 
\begin{equation}\label{hypercube S}
S(\partial \Omega_1(f)\cap A(3a/2)) \ge \left[\frac{c_5 a}{2^{\frac{d-2}{d-1}}(d-1)}\right]^{d-1}.
\end{equation}

Let $a = 1/M$ and divide $\Omega=[0,1]^d$ into $M^d$ hypercubes as in $\mathscr{H}$. For each small hypercube, the surface of intersection between $\partial\Omega_1(f)$ and $3^d$ hypercubes centering at the said hypercube satisfies equation (\ref{hypercube S}). In case that the hypercubes lie on the border of $\Omega$, there are fewer hypercubes centered at the said hypercube, but the direction of inequality below does not change. As long as $1/M < 2/(3\gamma)$, we obtain from \eqref{hypercube S} that
$$
\big|\{A\in\mathscr{H}:A\cap\partial\Omega_1(f)\ne\emptyset\}\big| \left[\frac{c_5 a}{2^{\frac{d-2}{d-1}}(d-1)}\right]^{d-1} \le 3^d S(\partial \Omega_1(f)) \le 3^d \gamma .$$
Thus we have
$$|\{A\in\mathscr{H}:A\cap\partial\Omega_1(f)\ne\emptyset\}|\le \left[\frac{2^{\frac{d-2}{d-1}}(d-1)}{c_5}\right]^{d-1} 3^d \gamma M^{-(d-1)}.$$
Let $c' = \left[\frac{2^{\frac{d-2}{d-1}}(d-1)}{c_5}\right]^{d-1} 3^d$. We have proved Lemma 5 when $\Omega_1(f)$ is a finite union of hyperrectangles. 

\paragraph{Step 2} 
If $\Omega_1(f)$ itself is not a finite union of hyperrectangles, then by definition of $\overline{\mathscr{F}_\gamma}$, there exists a sequence of classifiers $\{f_i: i\ge 1\}$, such that $\lim_{i\to\infty}d_H(\Omega_1(f_i),\Omega_1(f))=0$. Thus for all $\varepsilon>0$, there exists $i_0\in \mathbb{N}$, such that $f_{i_0}\in\mathscr{F}_\gamma$ and  $d_H(\Omega_1(f_{i_0}),\Omega_1(f))\le \varepsilon$. Still let $a = 1/M$ and divide $\Omega$ into $M^d$ hypercubes as in $\mathscr{H}$. The key idea is to prove as long as $\partial \Omega_1(f)$ intersects with a small hypercube $A\in\mathscr{H}$, the surface of the intersection between $\partial \Omega_1(f_{i_0})$ and a larger hypercube centering at the center of $A$ is bounded from below. Since the total surface of $\partial \Omega_1(f_{i_0})$ is bounded by $\gamma$, the number of hypercubes of $\mathscr{H}$ intersecting $\partial \Omega_1(f)$ is bounded below. The idea is the same as Step 1 where we only need to bridge $\partial \Omega_1(f)$ and $\partial \Omega_1(f_{i_0})$, which follows directly from the definition of Hausdorff distance.

For an arbitrary hypercube $A\in\mathscr{H}$, we still define $A(t)$ as the hypercube centering at the center of $A$ with side length $2t$. If $\partial\Omega_1(f)\cap A\ne\emptyset$, then there exists some $t'\in[a/2-\varepsilon, a/2+\varepsilon]$, such that $S(\partial\Omega_1(f_{i_0})\cap A(t'))>0$. Define $l(t),\zeta_0(t), \zeta_1(t), \Delta\zeta(t), S_{\text{int}}(t)$ and $\tilde{S}_{\text{int}}(t)$ similarly as in Step 1, but with all $\Omega_1(f)$ replaced with $\Omega_1(f_{i_0})$. Then by the same procedure as in Step 1, we have
$\frac{d}{dt}([\tilde{S}_{\text{int}}(t)]^{\frac{1}{d-1}}) \ge \frac{c_5}{2^{d-2}{d-1}(d-1)}$ for all $t\ge t'$. Therefore $\partial\Omega_1(f)\cap A\ne \emptyset$ implies
$$S(\partial \Omega_1(f_{i_0})\cap A(3a/2)) \ge \tilde S_{\text{int}}(3a/2) \ge \left[ \frac{c_5 (3a/2-t')}{2^{\frac{d-2}{d-1}} (d-1)}\right]^{d-1} \ge \left[ \frac{c_5 (a-\varepsilon)}{2^{\frac{d-2}{d-1}} (d-1)}\right]^{d-1}.$$
Thus as long as $1/M < 2/(3\gamma)$, we have
$$|\{A\in\mathscr{H}:A\cap\partial\Omega_1(f)\ne\emptyset\}|\le \left[\frac{2^{\frac{d-2}{d-1}}(d-1)}{c_5}\right]^{d-1} 3^d \gamma (1/M-\varepsilon)^{d-1}$$
for all $\varepsilon>0$. Since $\varepsilon$ is arbitrary, we conclude that
$$|\{A\in\mathscr{H}:A\cap\partial\Omega_1(f)\ne\emptyset\}|\le \left[\frac{2^{\frac{d-2}{d-1}}(d-1)}{c_5}\right]^{d-1} 3^d \gamma M^{-(d-1)}.$$
We still let $c' = \left[\frac{2^{\frac{d-2}{d-1}}(d-1)}{c_5}\right]^{d-1} 3^d$ and the proof is complete.
\end{proof}

\subsection{Proof of Lemma 6}
\begin{proof}
For all $f\in\mathscr{N}_\epsilon$, $f$ is uniquely determined by choosing its border hypercubes. Since there are at most $c'\gamma M^{d-1}$ number of border hypercubes, we have $|\mathscr{N}_\epsilon|\le \sum_{k=1}^{c'\gamma M^{d-1}}{M^d\choose k}$. Noting for $k\le c'\gamma M^{d-1}$,
$$\frac{ {M^d \choose k-1}}{ {M^d \choose k}} = \frac{k}{M^d-k}\le \frac{c'\gamma M^{d-1}}{M^d-c'\gamma M^{d-1}} = \frac{\epsilon}{1-\epsilon}.$$
Therefore
\begin{equation}\label{comb}
\sum_{k=1}^{c'\gamma M^{d-1}}{M^d\choose k}
\le  \frac{1}{1-\epsilon/(1-\epsilon)} {M^d\choose c'\gamma M^{d-1}} 
\le \frac{3}{2} \frac{(M^d)^{c'\gamma M^{d-1}}}{(c'\gamma M^{d-1})!}, 
\end{equation}
where the last equality uses $\epsilon<1/4$. By the lower bound of Stirling's formula in \cite{robbins1955remark}, we have
\begin{equation}\label{stirling}
\log [(c'\gamma M^{d-1})!] > \log(\sqrt{2\pi}) + (c'\gamma M^{d-1}+1/2)\log(c'\gamma M^{d-1}) - c'\gamma M^{d-1} + \frac{1}{12c'\gamma M^{d-1}+1}.
\end{equation}
Taking logarithm on both sides of the inequality (\ref{comb}) and applying the lower bound in (\ref{stirling}) for the factorial, we have
\begin{align*}
& \log \left[\sum_{k=1}^{c'\gamma M^{d-1}}{M^d\choose k}  \right] \\
% & \le \log(3/2) + dc'\gamma M^{d-1} \log M - [\log(\sqrt{2\pi}) + (c'\gamma M^{d-1}+1/2)\log(c'\gamma M^{d-1}) - c'\gamma M^{d-1} + \frac{1}{12c'\gamma M^{d-1}+1}] \\
\le & c'\gamma M^{d-1}[\log M + 1 - \log(c'\gamma)] - \frac{d-1}{2}\log M - \frac{1}{12c'\gamma M^{d-1}+1} - \log(\sqrt{2\pi}) + \log(3/2) - \frac{1}{2}\log(c'\gamma) \\
\le & c'\gamma M^{d-1}[\log M + 1 - \log(c'\gamma)] \\
\le & c'\gamma \left(\frac{c'\gamma}{\epsilon}\right)^{d-1} \left[\log\left(\frac{c'\gamma}{\epsilon}\right) + 1 - \log(c'\gamma)\right].
\end{align*}
\end{proof}

\subsection{Proof of Lemma 7}
\begin{proof}
Denote $\mathscr{N}$ as
$$\mathscr{N} = \{ f: \; \Omega_1(f) \text{ is a finite union of elements of }\mathscr{H}\}.$$
For all $f\in\mathscr{F}_\gamma$, define $f_0^*$ and $B(f)$ as
$$f_0^* = \argmin_{f'\in\mathscr{N}, \Omega_1(f)\subset\Omega_1(f')} V(\Omega_1(f')),$$
$$B(f) = \bigcup_{A\in\mathscr{H},A\cap\partial \Omega_1(f)\ne \emptyset} A.$$
That is, $f_0^*$ is the smallest volume classifiers of $\mathscr{N}$ to contain the decision set of $f$ and $B(f)$ is the union of hypercubes of $\mathscr{H}$ intersecting the borders of $\Omega_1(f)$. 

We claim $\Omega_1(f_0^*)\btr \Omega_1(f)\subset B(f)$. We prove the claim by contradiction. 
By the definition of $\mathscr{H}$, $\Omega\subset\cup_{A\in\mathscr{H}}A$. Suppose the claim does not hold, then there exists $x\in A$ for some $A\in\mathscr{H}$ with $A\cap \partial \Omega_1(f)= \emptyset$. Since $\Omega_1(f)\subset\Omega_1(f_0^*)$ by definition, we have $A\subset (\Omega_1(f))^c$. Then define $f'_0$ such that $\Omega_1(f'_0) = \Omega_1(f_0^*)\backslash A$. We have $\Omega_1(f)\subset \Omega_1(f'_0)$ and $V(\Omega_1(f'_0))<V(\Omega_1(f_0^*))$, which contradicts the definition of $f_0^*$. Thus the claim holds.

We also claim $f_0^*\in\mathscr{N}_\epsilon$. By Lemma 5, it suffices to prove for all  border hypercube $A$ of $f_0^*$, $A\cap \partial \Omega_1(f)\ne \emptyset$. We again prove by contradiction. If there exists a border hypercube $A$ satisfying $A\cap \partial \Omega_1(f)=\emptyset$, noting $\Omega_1(f)\subset\Omega_1(f_0^*)$, we have $A\subset \Omega_1(f_0^*)\btr \Omega_1(f) \subset B(f)$, which means $A\cap\partial \Omega_1(f)\ne \emptyset$. Thus $f_0^*\in\mathscr{N}_\epsilon$.

By the definition of $\epsilon$ and Lemma 5, we have for all $f\in\mathscr{F}_\gamma$, 
\begin{align*}
&\inf_{f_0}\mathbb{P}(\Omega_1(f)\btr\Omega_1(f_0)) \le \mathbb{P}(\Omega_1(f)\btr\Omega_1(f_0^*)) \le \mathbb{P}(B(f)) \\
&\quad \le\frac{1}{M^d} |\{A\in\mathscr{H}: A\cap \partial \Omega_1(f)\ne \emptyset\}| \le \frac{1}{M^d} c'\gamma M^{d-1} = \epsilon.
\end{align*}
This proves the first inequality of Lemma 7. 
For the second inequality, for the same $f_0^*$ defined above, we note that
$$\mathbb{P}_n(\Omega_1(f)\btr\Omega_1(f_0^*)) \le \mathbb{P}_n(B(f)) \le \frac{1}{n}  \sum_{i=1}^n \mathbbm{1}_{\{X_i\in B(f)\}}$$
is the average of $n$ i.i.d. random variables, each following the Bernoulli distribution with mean $\mathbb{P}(B(f))\le \epsilon$. Thus we apply Chernoff bound on i.i.d. Bernoulli random variables to obtain that for all $f\in\mathscr{F}_\gamma$,
$$
% \mathbb{P}\left(\inf_{f_0\in\mathscr{N}_\epsilon} \mathbb{P}_n(\Omega_1(f)\Delta\Omega_1(f_0))\ge 2\epsilon \right) 
\mathbb{P}\left(\mathbb{P}_n(\Omega_1(f)\btr\Omega_1(f_0^*))\ge 2\epsilon \right) 
\le \mathbb{P}\left(\mathbb{P}_n(B(f) \ge 2\epsilon \right)
    \le \exp(-n \mr{KL}(2\epsilon \| \epsilon) \le \exp\{-2(\log 2) n \epsilon\},
$$
where $\mr{KL}(2\epsilon \| \epsilon)$ represents the KL divergence between the Bernoulli distribution with mean $2\epsilon$ and the Bernoulli distribution with mean $\epsilon$. Noting we already proved $\mathbb{P}(\Omega_1(f)\Delta\Omega_1(f_0^*))\le \epsilon$, we have
$$\mathbb{P}\left(\inf_{f_0\in\mathscr{N}_\epsilon} \max\{2\mathbb{P}(\Omega_1(f)\Delta\Omega_1(f_0)), \mathbb{P}_n(\Omega_1(f)\Delta\Omega_1(f_0))\}\ge 2\epsilon \right) \le \exp(-n \mr{KL}(2\epsilon \| \epsilon)) \le \exp\{-2(\log 2) n \epsilon\}.$$

To derive bounds for the supreme of all $f\in\mathscr{F}_\gamma$, by Lemma 5, $B(f)$ consists of at most $c'\gamma M^{d-1}$ small hypercubes of $\mathscr{H}$. Despite there are uncountably many elements in $\mathscr{F}_\gamma$, there are at most $\sum_{k=1}^{c'\gamma M^{d-1}}{M^d\choose k}$ different choices for $B(f)$. Thus we apply a union bound to obtain that
$$\mathbb{P}\left(\sup_{f\in\mathscr{F}_\gamma}\inf_{f_0\in\mathscr{N}_\epsilon} \max\{2\mathbb{P}(\Omega_1(f)\Delta\Omega_1(f_0)), \mathbb{P}_n(\Omega_1(f)\Delta\Omega_1(f_0))\}\ge 2\epsilon \right) 
\le \sum_{k=1}^{c'\gamma M^{d-1}}{M^d\choose k} \;\exp\{-2(\log 2) n \epsilon\}.$$

\end{proof}

\subsection{Proof of Lemma 8}
\begin{proof}
We first write $R_n(T)-R_n(f_1)$ as an average of i.i.d. random variables:
\begin{align*}
R_n(T)-R_n(f_1) & =  n^{-1} \sum_{i=1}^n \Big[  \mathbbm{1}_{\{X_i\in \Omega_1(T)\backslash\Omega_1(f_1), Y_i=0\}} -   \mathbbm{1}_{\{X_i\in \Omega_1(T)\backslash\Omega_1(f_1), Y_i=1\}} \\
    & \quad - \mathbbm{1}_{\{X_i\in \Omega_1(f_1)\backslash\Omega_1(T), Y_i=0\}} +   \mathbbm{1}_{\{X_i\in \Omega_1(f_1)\backslash\Omega_1(T), Y_i=1\}} \Big] \\
    & \triangleq   n^{-1} \sum_{i=1}^n \xi_i,
\end{align*}
% where $\alpha\ge 1$ is the weight assigned to minority class samples and $1/\alpha\le w_0\le 1$ is the normalizing constant, as defined in Section 2.1 of the main paper.
By definition, clearly $|\xi_i|\leq 1$. The expectation and variance of $\xi_i$ can be computed as:
% \begin{equation}
% \end{equation}
\begin{align}
\mathbb{E}\xi_i & = \mathbb{E}[R_n(T) - R_n(f_1)], \label{E xi} \\
\var(\xi_i) & =  \mathbb{E}\xi_i^2 - (\mathbb{E}\xi_i)^2 \nonumber\\
 & \le \mathbb{P}(\Omega_1(T)\btr\Omega_1(f_1))  \nonumber \\
&\qquad + [1-\mathbb{P}(\Omega_1(T)\btr\Omega_1(f_1))] (\mathbb{E}[R_n(T) - R_n(f_1)])^2 - (\mathbb{E}\xi_i)^2\nonumber\\
 & \le\mathbb{P}(\Omega_1(T)\btr\Omega_1(f_1)). \label{var xi}
\end{align} 

By definition, for $R(T)-R^*>7\epsilon$, $R(T)-R(f_1)\ge R(T)-R^*-|R(f^*)-R(f_1)|>6\epsilon.$ Thus we have $\frac{2}{3}[R(T)-R(f_1)]>4\epsilon$. We combine this with equation (\ref{E xi}), (\ref{var xi}) and apply Bernstein inequality to obtain that
\begin{align}\label{Rn bernstein}
& \mathbb{P}\left( \sup_{T\in |\mathscr{N}_\epsilon|, R(T)-R(f^*)\ge 7\epsilon} R_n(T) - R_n(f_1) > 4\epsilon\right) \nonumber\\
<  & \mathbb{P}\left( \sup_{T\in |\mathscr{N}_\epsilon|, R(T)-R(f^*)\ge 7\epsilon} R_n(T) - R_n(f_1)-[R(T)-R(f_1)] > \frac{1}{3}[R(T)-R(f_1)]\right) \nonumber\\
\le & |\mathscr{N}_\epsilon| \exp\left( - \frac{\Big(\frac{n}{3}[R(T)-R(f_1)]\big)^2}{2n\var(\xi_i) + \frac{2}{3}n\sup|\xi_i|[R(T)-R(f_1)]} \right) \nonumber\\
\le & |\mathscr{N}_\epsilon| \exp\left( - \frac{\frac{n}{9}[R(T)-R(f_1)]^2}{2\Big[\mathbb{P}(\Omega_1(T)\btr\Omega_1(f_1))+\frac{1}{3}[R(T)-R(f_1)]\Big]} \right) \nonumber\\
\le & |\mathscr{N}_\epsilon| \exp\left( - \frac{\frac{n}{9}[R(T)-R(f_1)]^2}{\frac{8}{3} \mathbb{P}(\Omega_1(T)\btr\Omega_1(f_1))} \right)
\end{align}
Noting that $\mathbb{P}(\Omega_1(f_1)\btr\Omega_1(f^*))\le \epsilon$ and $\mathbb{P}(\Omega_1(T)\btr\Omega_1(f^*))\ge R(T)-R^*\ge 7\epsilon$ while utilizing Condition 1, we have
\begin{equation}\label{p rela}
\mathbb{P}(\Omega_1(T)\btr\Omega_1(f_1))\le \frac{8}{7}\mathbb{P}(\Omega_1(T)\btr\Omega_1(f^*))\le \frac{8}{7c_1^{1/\kappa}}[R(T)-R^*)]^{1/\kappa}.
\end{equation}
Noting that $R(T)-R^*\ge 7\epsilon$ and $R(f_1)-R^*\le \mathbb{P}(\Omega_1(f_1)\btr\Omega_1(f^*))\le \epsilon$, we have
\begin{equation}\label{R rela}
R(T)-R(f_1)\ge \frac{6}{7} [R(T)-R^*].
\end{equation}
Combining equation (\ref{Rn bernstein}), (\ref{p rela}) and (\ref{R rela}), we have
\begin{align*}
 & \mathbb{P}\left( \sup_{T\in \mathscr{F}_{\gamma,n}, R(T)-R(f^*)\ge 7\epsilon} R_n(T) - R_n(f_1) > 4\epsilon\right) \\
\le & |\mathscr{N}_\epsilon|\exp\left(-\frac{3}{112}nc_1^{1/\kappa} [R(T)-R^*]^{2-1/\kappa}\right) \\
\le & |\mathscr{N}_\epsilon|\exp\left(-\frac{21}{16} nc_1^{1/\kappa} \epsilon^{2-1/\kappa}\right) .
\end{align*}
\end{proof}

\subsection{Proof of Proposition 1}
\begin{proof}   %\hfill
 {\bf Step 1}\quad We first prove for all $f\in\mathscr{F}_0\backslash\mathscr{F}_\gamma$, there exists $f'\in\mathscr{F}_\gamma$, such that
\begin{equation}\label{SVR compare}
S(f')-S(f) < -\gamma V(\Omega_1(f')\btr\Omega_1(f)).
\end{equation}
By definition of $\mathscr{F}_\gamma$,  there exists a $f'\in\mathscr{F}_0$ that satisfies equation (\ref{SVR compare}). The key is to prove such $f'$ lies in $\mathscr{F}_\gamma$. Since both $\Omega_1(f')$ and $\Omega_1(f)$ consist of finite unions of hyperrectangles, we assume both $f$ and $f'$ are trees. Denote $\mathscr{S}_j,\;1\le j\le d$ as the collection of split locations of tree $f$ at feature $j$. Define $a$ as the minimal Euclidean distance between any two splits on the same feature of $f$
$$a = \min_{1\le j\le d} \min_{x,y\in\mathscr{S}_j} |x-y|.$$
We claim there exists $f'\in\mathscr{F}_0$ satisfying equation (\ref{SVR compare}) and $S(f') < S(f) - \gamma a^d$. 
% Equation (\ref{SVR compare}) follows directly from definition of $\mathscr{F}_\gamma$. 
For a general tree $T$ formed by $k$ splits, similar as the proof of Technical Lemma \ref{nodes 2 impu}, denote the split locations of $f'$ as $s(T) = ((x_1,j_1), (x_2,j_2),\cdots, (x_k,j_k))$. Define the ratio between surface decrease and volume change as
$$\Delta r(s(T)) \triangleq \frac{S(f)-S(T)}{V(\Omega_1(T)\btr\Omega_1(f))}.$$
% Obviously, the ratio $\delta r(\cdot)$ is a function of splits locations. 
Noting both $S(f)-S(T)$ and $V(\Omega_1(T)\btr\Omega_1(f))$ are linear functions of the splits values $x_1, x_2, \cdots x_k$ when both the splits features $j_1, j_2, \cdots j_k$ and the ordering of all splits locations of $T$ and $f$ are unchanged (intuitively, the ``shape'' of trees are unchanged), the extreme value of $\Delta r(s(T))$ can only be achieved when the split locations of $T$ overlap with split locations of $f$. Thus when $\Delta r(\cdot)$ reaches a (local) minimum, either $V(\Omega_1(T)\btr\Omega_1(f))\ge a^d$ or $V(\Omega_1(T)\btr\Omega_1(f))=0$. By definition of $\mathscr{F}_\gamma$, there exists $f'_0$, such that $S(f'_0)-S(f) < -\gamma V(\Omega_1(f'_0)\btr\Omega_1(f))$. Noting $f'_0$ is a tree classifier, this implies $V(\Omega_1(f'_0)\btr\Omega_1(f))<0$. Let $f'$ be the local minimum of $\Delta r(\cdot)$ that is closest to $f'_0$ (in terms of distance between splits defined in equation (\ref{splits distance})), then $f'$ satisfies equation (\ref{SVR compare}) and $V(\Omega_1(f')\btr\Omega_1(f))\ge a^d$. This proves the claim.

If $f'\in\mathscr{F}_\gamma$, then we have achieved the objective of this step. Otherwise, since the split locations of $f'$ overlap with the split locations of $f$, we can repeat the above arguments and obtain some tree classifier $f''$ satisfying equation (\ref{SVR compare}) and $S(f'')<S(f) - 2\gamma a^d$. Since $S(f)$ is finite, the above process can only be repeated a finite number of times. Eventually we will obtain some classifier belonging to $\mathscr{F}_\gamma$ while satisfying equation (\ref{SVR compare}). \\

{\bf Step 2}\quad Let $f$ be an arbitrary classifier in $\mathscr{F}_0\backslash\mathscr{F}_\gamma$.
Let $f_{all}$ be a classifier satisfying $\Omega_1(f_{all})=\Omega$. Clearly $f_{all}\in\mathscr{F}_\gamma$ since $\gamma \ge 2d$. Then we have
$$\tilde{I}(f_{all},\mathbb{P})+\lambda r(f_{all})\le 1+2d\lambda, \;\;  R(f_{all})+\lambda r(f_{all})\le 1+2d\lambda.$$
Therefore if $r(f)>1+2d\lambda$, we set $f'=f_{all}$ and the proof is finished. In the rest of the proof, we only consider the situation when $r(f)\le 1+2d\lambda$. By step 1, there exists $f'\in\mathscr{F}_\gamma$ satisfying equation (\ref{SVR compare}). To limit notation, denote $V(\Omega_1(f)\btr\Omega_1(f'))=\delta$. We have
$$0<S(f') < S(f) - \gamma \delta,$$
$$V(f') \ge V(f) - \delta \ge \frac{S(f)}{1+2d\lambda} - \delta 
> \frac{S(f)-\gamma\delta}{1+2d\lambda} > 0.$$
Thus we have $\frac{S(f')}{V(f')} < \frac{S(f) -\gamma \delta}{V(f) - \delta}$. Therefore
\begin{align*}
\frac{S(f')}{V(f')} - \frac{S(f)}{V(f)}
< & \frac{S(f)-\gamma \delta}{V(f)-\delta} - \frac{S(f)}{V(f)} \\
= & \frac{\delta(-\gamma V(f) + S(f))}{(V(f)-\delta)V(f)} \\
= & \frac{(-\gamma + S/V) \delta}{V-\delta} \\
\le & \frac{(-\gamma + 2d\lambda + 1) \delta}{V-\delta} \\
\le & -\frac{\delta \rho_{\max}/\lambda}{V-\delta} \\
\le &  -\frac{\delta \rho_{\max}}{\lambda},
\end{align*}
which implies $\lambda\frac{S(f')}{V(f')} - \lambda\frac{S(f)}{V(f)}< -\delta \rho_{\max}$. By definition, we have
$$|\tilde{I}(f',\mathbb{P}) - \tilde{I}(f,\mathbb{P})|\le \mathbb{P}(\Omega_1(f)\btr\Omega_1(f')) \le \delta \rho_{\max},$$
$$|R(f') - R(f)|\le \mathbb{P}(\Omega_1(f)\btr\Omega_1(f')) \le \delta \rho_{\max}.$$
Therefore,
$$\tilde{I}(f'\mathbb{P}) + \lambda r(f')< \tilde{I}(f\mathbb{P}) + \lambda r(f),\;\;
R(f') + \lambda r(f')< R(f) + \lambda r(f).$$
Since the above proof holds for arbitrary $f\in\mathscr{F}_0\backslash\mathscr{F}_\gamma$, we finish the proof.

\end{proof}

\subsection{Proof of Proposition 2}
\begin{proof} We claim that convergence in Hausdorff distance of decision sets implies convergence in symmetric set difference of decision sets for $f_i\in\mathscr{F}_\gamma$. That is, there exists a constant $C$ independent of $\epsilon$, such that for all $\epsilon<1/2$, $\forall f_1, f_2 \in\mathscr{F}_\gamma$, $d_H(f_1, f_2)<\epsilon$ implies $V(\Omega_1(f_1)\btr\Omega_1(f_2))<C \epsilon$. 

To prove the claim, we divide $\Omega=[0,1]^d$ into $M^d$ hypercubes as in Section 3.2 of the main paper, with $M = \lfloor 1/\epsilon \rfloor$. Denote the collection of these $M^d$ hypercubes as $\mathscr{H}$. By Lemma 5, there are at most $c'\gamma M^{d-1}$ hypercubes that intersect with $\partial \Omega_1(f_1)$. Define the set $B_\epsilon$ as
$$B_\epsilon \triangleq \{x\in\Omega: d(x,\partial\Omega_1(f_1))\le \epsilon\},$$
where $d$ represents the Euclidean distance between a point and a set in $\mathbb{R}^d$. Since $\epsilon<1/M$, for any $x\in B_\epsilon$, the hypercube in $\mathscr{H}$ where $x$ lies in must either intersect with $\partial \Omega_1(f_1)$, or be adjacent to a hypercube that intersects with $\partial \Omega_1(f_1)$. Therefore we have
$$V(B_\epsilon) \le 3^d c'\gamma M^{d-1} \cdot M^{-d} \le 3^d c'\gamma M^{-1} \le 2\cdot 3^d c'\gamma \epsilon,$$
where the last inequality uses $\lfloor 1/\epsilon \rfloor < 2/\epsilon$ for $\epsilon < 1/2$. Since $d_H(\Omega_1(f_1), \Omega_1(f_2))<\epsilon$, we have $\Omega_1(f_1)\btr \Omega_1(f_2) \subset B_\epsilon$. Thus 
$$V(\Omega_1(f_1)\btr\Omega_1(f_2)) \le V(B_\epsilon) \le 2\cdot 3^d c'\gamma \epsilon.$$
Letting $C = 2\cdot 3^d c'\gamma$, the claim is proved. 

% Claim 2: For all $f_1, f_2\in\mathscr{F}_\gamma$, $|S(f_1)-S(f_2)|\le \gamma V(\Omega_1(f_1)\Delta\Omega_1(f_2))$.
% To prove claim 2, consider three sets: $\Omega_1(f_1)$, $\Omega_1(f_1)\cup\Omega_1(f_2)$ and $\Omega_1(f_2)$. Define $f_3:\Omega\to[0,1]$ as $\mathbbm{1}_{\{\Omega_1(f_1)\cup\Omega_1(f_2)\}}(X)$. Since $\Omega_1(f_1)\subset \Omega_1(f_1)\cup\Omega_1(f_2)$ and both sets consists of finite number of hyperrectangles, there exists a finite number of hyperrectangles $A_1, A_2, \cdots, A_k$, such that $(\cup_{j=1}^k A_k)\cup \Omega_1(f_1) = \Omega_1(f_1)\cup\Omega_1(f_2)$. Thus applying definition of 

We then prove the proposition. Assume there is a sequence of $f_i\in\mathscr{F}_\gamma$ whose decision sets $\Omega_1(f_i)$ converge in Hausdorff distance. Then the claim shows that $\{\Omega_1(f_i):i\geq 1\}$ also converges in symmetric set difference. Thus there exists $\Omega^*\subset\Omega$, $\forall \epsilon>0$, $\exists i_0\in\mathbb{N}$, such that $\forall i>i_0$, $V(\Omega_1(f_i)\btr \Omega^*)<\epsilon$. Thus for all $i,j>i_0$, $V(\Omega_1(f_i)\btr\Omega_1(f_j))\le 2\epsilon$. By definition of $\mathscr{F}_\gamma$, we have 
$$S(f_i)-S(f_j)\ge -\gamma V(\Omega_1(f_i)\btr\Omega_1(f_j)), \;\; S(f_j)-S(f_i)\ge -\gamma V(\Omega_1(f_i)\btr\Omega_1(f_j)).$$
Combining these two equations, we have $|S(f_i)-S(f_j)|\le |\gamma V(\Omega_1(f_i)\btr\Omega_1(f_j))| \le \gamma \epsilon$. Thus $\{S(f_i):i\geq 1\}$ is a Cauchy sequence in $\mathbb{R}$ and is therefore a convergent sequence in $\mathbb{R}$.
\end{proof}

% \subsection{Proof of Lemma 9}

% \subsection{Proof of Lemma 10}

\section{Feature Selection Consistency}
\label{sec:selection}
\subsection{Theorem Statement}
Let $X[-j]$ denote the set of all features except $X[j]$. We say $X[j]$ is redundant if conditionally on $X[-j]$, $Y$ is independent of $X[j]$.
%$$(Y | X[-j]) \perp X[j].$$
We denote $X'$ as the collection of all non-redundant features, and $X''$ as the collection of all redundant features. Under two conditions on the distribution of $Y, X'$ and $X''$, we can show if $\lambda_n$ goes to zero slower than $n^{-1/2}$, the probability of $\hat{T}_n$ 
excluding all redundant features goes to one. We assume there are $q \;(1\le q < d)$ non-redundant features denoted as $X[j_1], X[j_2], \cdots X[j_q]$. The redundant features are denoted as $X[j_{q+1]}, X[j_{q+2}], \cdots X[j_d]$.

Before stating these two conditions, we need to discuss how a split can decrease the tree impurity.  Suppose we are splitting on node $A$ at feature $X[j]$, resulting in two new leaf nodes: $A_1 = \{X\in A: X[j] \le x_j\}$ and $A_2 = \{X\in A: X[j] > x_j\}$. 
% Then the weighted conditional expectation of $Y$ in node $A_1, A_2$ is
% \begin{equation}\label{p1alpha fs}
% p_{h, \alpha} = \frac{\alpha \mathbb{E}(Y|X\in A_h)}{1-\mathbb{E}(Y|X\in A_h) + \alpha\mathbb{E}(Y|X\in A_h)},\quad h=1,2. 
% \end{equation}
Recall for arbitrary set $B\subset\Omega$, we denote $\eta(B)$ as $\eta(B)=\mathbb{E}(Y|X\in B)$. The impurity of $A_1, A_2$ is
$I(A_1, \mathbb{P}) = 1 - \eta(A_1)^2 - (1-\eta(A_1)^2 = 2\eta(A_1)(1-\eta(A_1)),$ and
$I(A_2, \mathbb{P}) = 2\eta(A_2)(1-\eta(A_2)).$
Similarly, the impurity of $A$ is $I(A,\mathbb{P}) = 2\eta(A)(1-\eta(A))$. The impurity decrease on node $A$ is
\begin{equation}\label{delta Ialpha}
\Delta I(A, \mathbb{P}) = 2\eta(A)(1-\eta(A)) - 2\eta(A_1)(1-\eta(A_1))\mathbb{P}(A_1|A) - 2\eta(A_2)(1-\eta(A_2))\mathbb{P}(A_2|A),.
\end{equation}
Noting $A_1, A_2$ is determined by $x_j$, $\Delta I(A, \mathbb{P})$ can be considered as a function of $x_j$. Denoting $\Delta I(A,\mathbb{P}) = h_{A,j}(x_j)$, then the maximal impurity decrease at feature $X[j]$ is
\begin{equation}\label{Mj}
M_{A,j} = \max_{x_j\in(x_{j,1},x_{j,2})} h_{A,j}(x_j).
\end{equation}
The quantity of $M_{A,j}$ in reducing the impurity of node $A$ is measured relative to the oracle impurity decrease at node $A$. Suppose we split node $A$ into  measurable sets $A_1$ and $A_2$ satisfying
$\mathbb{P}(A_1| A) = V_1,$ $A_1\cup A_2=\Omega,$ $A_1\cap A_2 = \emptyset,$
for all $X_1\in A_1,$ $X_2\in A_2,$ and $\eta(X_1)\le\eta(X_2),$
where $V_1\in[0,1]$ and $A_1, A_2$ are not required to be hyper rectangles. By definition, the set $A_1$ corresponds to the $V_1$ proportion of $A$  having the smallest $\eta(X)$ values, while the set $A_2$ corresponds to the $1-V_1$ proportion of $A$ having the largest $\eta(X)$ values. Similar to equation (\ref{delta Ialpha}), we can define $\eta(A_1), \eta(A_2)$ and the impurity decrease $\Delta I(A,\mathbb{P})$. Given $V_1\in [0,1]$, the impurity decrease $\Delta I(A,\mathbb{P})$ is unique, which is also the maximal impurity decrease for all measurable $A_1, A_2$ satisfying $\mathbb{P}(A_1| A) = V_1,\; \mathbb{P}(A_2|A) = 1-V_1$. Thus we can denote the impurity decrease as $\Delta I(A,\mathbb{P}) = h^*_A(V_1)$. The oracle impurity decrease is the maximal value of $h^*_A(V_1)$:
$M_{A}^* = \max_{V_1\in[0,1]} h^*_A(V_1).$
In general, the larger impurity decrease will tend to correspond to non-redundant features; this will particularly be the case under Conditions 2-3.  

\begin{condition}[]\label{condition 1}
There exists $c_3\in [0,1]$, such that for all $A\subset\Omega$, 
$\sup_{1\le j\le q} M_{A,j} \ge c_3 M_A^*.$
\end{condition}
Condition \ref{condition 1} relates the impurity decrease in non-redundant features to the oracle impurity decrease. The strength of the condition is dependent on the value of $c_3$. If $c_3=0$, the condition does not impose any restrictions; if $c_3 =1$, splits in non-redundant features can fully explain the oracle impurity decrease. 
Here we give some examples of models with different $c_3$ values.\par{}

\begin{example}[Generalized Linear Models]
Let $\Omega = [0,1]^d$ and the marginal distribution of $\mathbb{P}$ be the uniform distribution on $\Omega$. Suppose $\mathbb{E}(Y|X)= \phi(a^T X + b),$
where $\phi(\cdot)$ is a monotonically increasing function, $a\in\mathbb{R}^d$ and $b\in\mathbb{R}$. Let $A_{y1}, A_{y2}$ be measurable sets having 
$A_{y1}\cap A_{y2}=\emptyset$ and $A_{y1}\cup A_{y2} = \Omega$, which achieve the oracle impurity decrease $M_A^*$, and let $\min\{\mathbb{P}(A_{y1}), \mathbb{P}(A_{y2})\}=h$. Then the constant $c_3$ in Condition 1 satisfies
$c_3 \ge (\sqrt{2}-1)h + o(h)$ if $h \le 1/d!$ and $c_3 \ge F_{d-1}(F_d^{-1}(h)-h) - F_{d-1}(F_d^{-1}(h)-2h)$ if $h>1/d!$, 
%$$c_3 \ge \Bigg\{ \begin{aligned} & (\sqrt{2}-1)h + o(h), \;& h\le 1/d! \\
										%							&F_{d-1}(F_d^{-1}(h) -h) - F_{d-1}(F_d^{-1}(h) -2h), \;& h> 1/d! \end{aligned} $$
where  
$F_d(z) = \frac{1}{d!} \sum_{k=0}^{\lfloor z\rfloor}(-1)^k {d\choose k} (z-k)^{d-1}$ 
is the cumulative distribution function for the Irwin-Hall distribution with parameter $d$.
\end{example}
%For any fixed form of $\phi$, the upper bound $h$ can be controlled by controlling the $L^p, p\ge 1$ norm of regression coefficients $a$. Especially, when $\phi(\cdot)$ is a linear function, $h=0.5$.

\begin{example}
Let $\Omega = [0,1]^2$ and the marginal distribution of $\mathbb{P}$ be the uniform distribution on $\Omega$. Let 
$A_1 = [0,1/2]\times[0,1/2],$ $A_2 =  [1/2,1]\times[0,1/2],$
$A_3 = [0,1/2]\times[1/2,1],$ $A_4 =  [1/2,1]\times[1/2,1],$ 
$\mathbb{E}(Y|X) = 1$ if $X \in A_1\cup A_4$ and 
$0$ if $X\in A_2\cup A_3.$ Then $c_3=0$.
\end{example}

%The proofs of the bounds on $c_3$ in Examples 1-2 are in the Appendix.

\begin{condition}\label{condition 2}
The weighted probability measure $\mathbb{P}$ has density $\rho(X', X'')$ in $\Omega$. Moreover, 
$\rho(X', X'') = \rho_{1}(X') \rho_{2}(X'') + \rho_{3}(X', X''),$
where for all $(X', X'')\in \Omega$, 
$\rho_{3}(X', X'') \le c_4 \rho(X', X''),$
with $c_4\in [0,1]$ a constant.
\end{condition}

Condition \ref{condition 2} asserts that the joint density of $X', X''$ can be decomposed into an independent component plus a dependent component, where the dependent component is dominated by  the overall density multiplied by a constant $c_4$. This condition controls the dependence between $X'$ and $X''$. Since given $X'$, $Y$ is independent of $X''$, this condition will also control the dependence between $Y$ and $X''$. The strength of the condition depends on the constant $c_4$. When $c_4=1$, the condition imposes no restrictions; when $c_4=0$, the condition asserts $X'$ and $X''$ are completely independent.\par{}

Using Conditions \ref{condition 1}-\ref{condition 2}, we establish feature selection consistency in Theorem \ref{feature selection}.
\begin{theorem}[Feature selection consistency]\label{feature selection}
Assume conditions of Theorem 1 are satisfied, Condition \ref{condition 1}, \ref{condition 2} hold with $c_3>(1-c_4)/[c_4(2-c_4)]$ and $\lambda_n\ge c_5 n^{-(1/2-\beta)}$ for some constant $c_5>0$, $\beta\in (0, 1/2)$.
If the optional steps in Algorithm 1 are enabled, we have
$$\lim_{n\to\infty}\mathbb{P}(\hat{T}_n \;\text{does not have splits in}\; X[j_{q+1}], X[j_{q+2}], \cdots X[j_d]) = 1.$$
\end{theorem}
%The proof of Theorem \ref{feature selection} is in the Appendix. 
We would like to discuss two issues about the conditions of Theorem \ref{feature selection}. First, a weak condition $\lambda_n \gg n^{-1/2}$ is required by this theorem. Combined with the requirement that $\lambda_n\to 0$ as $n$ goes to infinity for Theorem 1, $\lambda_n$ can take values between $1$ and $n^{-1/2}$. In our experiments, $\lambda_n$ is taken to be $O(n^{-1/3})$ and $\beta=1/6$. Second, conditions 1 and 2 are complementary in Theorem \ref{feature selection}. If $c_4$ is smaller (i.e., condition 2 is stronger), $c_3$ can be smaller (i.e., condition 1 is weaker). The opposite also holds. The following two corollaries cover two special cases.
%in one case we assume the impurity decrease in non-redundant features equals to the oracle impurity decrease while imposes no assumption on density of features; in the other case we assume no conditions on impurity decrease at non-redundant features while $X'$ must be independent of $X''$. They are stated in the following two corollaries.

\begin{corollary}\label{feature selection monotonic}
Assume conditions of Theorem 1 are satisfied, condition \ref{condition 1} is satisfied with $c_3=1$ (that is, in each hyper-rectangle, the maximal impurity decrease at non-redundant features is equal to the oracle impurity decrease) and $\lambda_n\ge c_5 n^{-(1/2-\beta)}$ for some constant $c_5>0$, $\beta\in (0, 1/2)$.
If the optional steps in Algorithm 1 are enabled, we have
$$\lim_{n\to\infty}\mathbb{P}(\hat{T}_n \;\text{does not have splits in}\; X[j_{q+1]}, X[j_{q+2}], \cdots X[j_d]) = 1.$$
\end{corollary}
%Corollary \ref{feature selection monotonic} can be proven by slightly modifying the proof of theorem \ref{feature selection}; refer to the Appendix.

\begin{corollary}\label{feature selection independent}
Assume conditions of Theorem 1 are satisfied, $X'$ and $X''$ are independent and $\lambda_n\ge c_5 n^{-(1/2-\beta)}$ for some constant $c_5>0$, $\beta\in (0, 1/2)$.
If the optional steps in Algorithm 1 are enabled, we have
$$\lim_{n\to\infty}\mathbb{P}(\hat{T}_n \;\text{does not have splits in}\; X[j_{q+1]}, X[j_{q+2}], \cdots X[j_d]) = 1.$$
\end{corollary}
Corollary \ref{feature selection independent} is a direct result of theorem \ref{feature selection} with $c_4=0$ and an arbitrary value of $c_3$.

\subsection{Proof of Theorem \ref{feature selection}}
The proof of Theorem \ref{feature selection} mainly consists of two parts. The first part works on the true distribution $\mathbb{P}$, proving that under $\mathbb{P}$, the split with highest impurity decrease is always in non-redundant features. The second part works on the randomness brought by $\mathbb{P}_n$, proving that with high probability, the randomness of the tree impurity decrease can be controlled by the threshold $c_0\lambda_n$. Combining these two parts with Theorem 1, we can show the probability of rejecting all redundant features in the tree building procedure goes to zero. We now focus on the first part.

\begin{lemma}\label{theoretical exclude}
Under the conditions of Theorem \ref{feature selection}, for all hyperrectangles $A\subset \Omega$, we have 
    $\sup_{1\le j\le q} M_{A,j} \ge \sup_{q+1\le l\le d} M_{A,l},$
where $M_{A,j}$ is the maximal impurity decrease of feature $j$ at node $A$ defined in equation (\ref{Mj}).
\end{lemma}
\begin{proof}[Proof of Lemma \ref{theoretical exclude}]
We utilize $M_A^*$ as a bridge to compare the values between $\sup_{1\le j\le q} M_{A,j}$ and $\sup_{q+1\le l\le d} M_{A,l}$. By condition 1, we have $\sup_{1\le j\le q} M_{A,j} \ge c_3 M_A^*$. It remains to compare the values between $\sup_{q+1\le l\le d} M_{A,l}$ and $M_A^*$. 
Suppose hyperrectangle $A$ is split into two measurable (but not necessarily hyperrectangular) sets $A_1$ and $A_2$, such that $A_1\cap A_2 =\emptyset$ and $A_1\cup A_2 = A$. 
% For simplicity of notation, $\forall X\in A$, we denote $p(X) = \mathbb{E}(Y|X)$. Recall for all $A\subset \Omega$, we let $p(A)=\mathbb{E}(Y|X\in A)$. Further denote $p_{1} = \mathbb{E}(Y|X\in A_1)$, $p_{2} = \mathbb{E}(Y|X\in A_2)$, 
Denote $V_1 = \mathbb{P}(A_1|A)$, $V_2 = \mathbb{P}(A_2|A)$.
The impurity decrease on node $A$ can be computed as
\begin{equation}\label{p1 p2 1}
\Delta I(A,\mathbb{P}) = 2\eta(A)(1-\eta(A)) - 2V_1 \eta(A_1)(1-\eta(A_1)) - 2V_2 \eta(A_1)(1-\eta(A_2)).
\end{equation}
By definition of conditional expectation, we also have
\begin{equation}\label{p1 p2 2}
V_1 \eta(A_1) + V_2\eta(A_2)  = p(A).
\end{equation}
Combining equation (\ref{p1 p2 1}), (\ref{p1 p2 2}) with the fact that $V_1 + V_2 = 1$, we have
\begin{equation}\label{p1 p2 3}
\Delta I(A,\mathbb{P}) = V_1 V_2 (\eta(A_1) - \eta(A_2))^2.
\end{equation}
By equation (\ref{p1 p2 3}), when $V_1, V_2$ are fixed, the impurity decrease on $A$ is proportional to the squared difference between $\eta(A_1)$ and $\eta(A_2)$. Let $A_{y1}, A_{y2}$ be a pair of split sets that achieves impurity decrease $h_A^*(V_1)$; i.e., 
$$\mathbb{P}(A_{y_1}) = V_1\mathbb{P}(A),\; A_{y_1}\cap A_{y_2} = \emptyset, \;A_{y_1}\cup A_{y_2} =A,$$
$$\forall X_1\in A_{y_1}, X_2\in A_{y_2}, \; \eta(X_1) \le \eta(X_2).$$
Let $y_{1}, y_{2}$ be the conditional expectation of $Y$ at $A_{y_1}, A_{y_2}$ under $\mathbb{P}$:
$$y_{1} = \mathbb{E}(Y|X\in A_{y1}), \quad y_{2} = \mathbb{E}(Y|X\in A_{y2}).$$
The critical value of conditional expectation (under $\mathbb{P}$) between $A_{y_1}, A_{y_2}$ is denoted as $y_t$:
$$y_t\in[0,1]: \;\; \sup_{X\in A_{y_1}} \eta(X) \le y_t \le \inf_{X\in A_{y_2}} \eta(X).$$
In the rest of the proof, supppose we split at a redundant feature to obtain $A_1$, $A_2$. For all fixed $V_1\in[0,1]$, we compare $|\eta(A_1)-\eta(A_2)|$ and $|y_{1}-y_{2}|$, thus comparing the values between $M_{A,l}$ and $M_A^*$, $\forall l\ge q+1$. For hyperrectangle $A$, we use $\mr{Proj}(A,X')$, $\mr{Proj}(A,X'')$ to denote the projection of $A$ on $X', X''$, respectively. By definition, $A = \mr{Proj}(A,X')\times \mr{Proj}(A,X'')$. We first consider the properties of $\eta(A_1)$ and $\eta(A_2)$.
\begin{align*}
\eta(A_1)V_1\mathbb{P}(A) = & 
					\int_{A_1}\rho_{1}(X')\rho_{2}(X'')\eta(X)dX'dX'' + \int_{A_1}\rho_{3}(X',X'')\eta(X)dX'dX'' \\
			= & \frac{\int_{\mr{Proj}(A_1,X')}\eta(X)\rho_{1}(X')dX'}{\int_{\mr{Proj}(A_1,X')} \rho_{1}(X')dX'} 
					\int_{\mr{Proj}(A_1,X')} \rho_{1}(X')dX' \int_{\mr{Proj}(A_1,X'')} \rho_{2}(X'')dX'' + \\ 
				& \frac{\int_{A_1} \rho_{3}(X',X'')\eta(X)dX'dX''}{\int_{A_1} \rho_{3}(X',X'')dX'dX''} \int_{A_1} \rho_{3}(X',X'')dX'dX'' 
\end{align*}
Denote
$$y_A \triangleq \frac{\int_{\mr{Proj}(A,X')}\eta(X)\rho_{1}(X')dX'}{\int_{\mr{Proj}(A,X')} \rho_{1}(X')dX'},$$
$$\gamma_1 = \int_{A_1} \rho_{3}(X',X'')dX'dX'' / \mathbb{P}(A_1).$$
Noting $\mr{Proj}(A_1,X') = \mr{Proj}(A,X')$, $\eta(X)$ is independent of $X''$ and $y_A$ being well-defined, we have
\begin{align*}
\eta(A_1)V_1\mathbb{P}(A) 
			= & \left[y_A (1-\gamma_1) + \frac{\int_{A_1} \rho_{3}(X',X'')\eta(X)dX'dX''}{\int_{A_1} \rho_{3}(X',X'')dX'dX''} \gamma_1\right] V_1\mathbb{P}(A).
\end{align*}
Therefore,
\begin{align*}
[\eta(A_1) - (1-\gamma_1)y_A - \gamma_1 y_t] V_1\mathbb{P}(A) 
		= & \int_{A_1} \rho_{3}(X',X'')[\eta(X) - y_t]dX' dX''.
\end{align*}
Similarly, define
$$\gamma_2 = \int_{A_2} \rho_{3}(X',X'')dX'dX'' / \mathbb{P}(A_2),$$
we have
$$[\eta(A_2) - (1-\gamma_2)y_A - \gamma_2 y_t] V_1\mathbb{P}(A) 
		= \int_{A_2} \rho_{3}(X',X'')[\eta(X) - y_t]dX' dX''.$$
Without loss of generality, we assume $\eta(A_1) - (1-\gamma_1)y_A - \gamma_1 y_t \le \eta(A_2) - (1-\gamma_2)y_A - \gamma_2 y_t$. Then
\begin{align*}
[\eta(A_1) - (1-\gamma_1)y_A - \gamma_1 y_t]V_1\mathbb{P}(A)
		= & \int_{A_1} \rho_{3}(X',X'')[\eta(X) - y_t]dX' dX'' \\
		\ge & \int_{A_1} \rho_{3}(X',X'')\min\{\eta(X) - y_t, 0\}dX' dX'' \\
		\ge & \int_A \rho_{3}(X',X'')\min\{\eta(X) - y_t, 0\}dX' dX'' \\
		= & \int_{\{X\in A: \eta(X)< y_t\}} \rho_{3}(X',X'')[\eta(X) - y_t]dX' dX'' \\
		\ge & c_4\int_{\{X\in A: \eta(X)< y_t\}} \rho_(X',X'')[\eta(X) - y_t]dX' dX''
\end{align*}
By the definition of $y_t$, 
$$\{X\in A:\eta(X)< y_t\} \subset A_{y_1} \subset \{X\in A: \eta(X)\le y_t\}.$$
Therefore 
\begin{align*}
[\eta_{1} - (1-\gamma_1)y_A - \gamma_1 y_t]V_1\mathbb{P}(A)
		\ge & c_4\int_{A_{y_1}} \rho_{3}(X',X'')[\eta(X) - y_t]dX' dX'' \\
		= & c_4\mathbb{P}(A) V_1 (y_{1} - y_t),
\end{align*}
Thus
\begin{equation}\label{step 3 1}
\eta(A_1) - (1-\gamma_1)y_A - \gamma_1 y_t \ge c_4(y_{1} - y_t).
\end{equation}
Similarly,  we have
\begin{equation}\label{step 3 2}
\eta(A_2) - (1-\gamma_2)y_A - \gamma_2 y_t \le c_4(y_{2} - y_t).
\end{equation}
We then consider the difference between $y_A$ and $y_t$. 
\begin{align*}
 & |y_A-y_t|\int_A \rho_{1}(X')\rho_{2}(X'') dX' dX'' \\
		= &\left| \int_A [\eta(X)-y_t]\rho_{1}(X')\rho_{2}(X'')  dX' dX''\right| \\
		\le & \int_A |\eta(X)-y_t| \rho_{1}(X')\rho_{2}(X'')  dX' dX'' \\
		\le & \int_A |\eta(X)-y_t| \rho(X',X'')  dX' dX'' \\
		= & \int_{A_{y_1}} |\eta(X)-y_t| \rho(X',X'')  dX' dX'' +  \int_{A_{y_2}} |\eta(X)-y_t| \rho(X',X'')  dX' dX''\\ 
		= & \left|\int_{A_{y_1}} [\eta(X)-y_t] \rho(X',X'')  dX' dX''\right| +  \left|\int_{A_{y_2}} [\eta(X)-y_t] \rho(X',X'')  dX' dX''\right| \\ 
		= & \left(|y_{1}-y_t|V_1 + |y_{2}-y_t| V_2\right)\mathbb{P}(A).
\end{align*}
Because 
$$\int_A \rho_{1}(X')\rho_{2}(X'') dX' dX'' \ge \int_A (1-c_4) dX' dX'' \ge (1-c_4)\mathbb{P}(A),$$
$$V_1<1, \;\; V_2<1,$$
we have
\begin{equation}\label{step 3 3}
|y_A-y_t|\le \left(|y_{1}-y_t| + |y_{2}-y_t| \right)\frac{1}{1-c_4} = \frac{|y_{1}-y_{2}|}{1-c_4}.
\end{equation}
Combining equations (\ref{step 3 1}), (\ref{step 3 2}) and (\ref{step 3 3}), we have
\begin{align}\label{step 3 4}
|\eta(A_2) - \eta(A_1)| 
	= & \left|[\eta(A_2) - (1-\gamma_2)y_A - \gamma_2 y_t] - [\eta(A_1) - (1-\gamma_1)y_A - \gamma_1 y_t] + (\gamma_1 - \gamma_2)(y_A - y_t)\right|\nonumber \\
	\le  & |[\eta(A_2) - (1-\gamma_2)y_A - \gamma_2 y_t] - [\eta(A_1) - (1-\gamma_1)y_A - \gamma_1 y_t]| + |(\gamma_1 - \gamma_2)(y_A - y_t)| \nonumber \\
	\le & c_4(y_{2}-y_{1}) + c_4 \frac{|y_{1}-y_{2}|}{1-c_4} \nonumber \\
	\le & \frac{c_4(2-c_4)}{1-c_4}|y_{2} - y_{1}|.
\end{align}
Because equation (\ref{step 3 4}) holds for all $V\in [0,1]$, recalling equation (\ref{p1 p2 3}), we have 
$$M_{A,l} \le \frac{c_4(2-c_4)}{1-c_4} M_A^*, \;\; \forall l\ge q+1.$$
Recalling that $c_3 > \frac{c_4(2-c_4)}{1-c_4}$, we finish the proof.
\end{proof}
We now work on the second part, proving the randomness in tree impurity decrease can be bounded by $c_0\lambda_n$ with high probability. Using the same notations as in Technical Lemma \ref{theoretical consistent}, we denote $\Delta I(A, (z,j),\mathbb{P})$ as the impurity decrease after splitting node $A$ at $X[j]=z$.

\begin{lemma}\label{noise thre}
Suppose $\lambda_n \ge n^{-1/2+\beta}$ and the density of $X$ exists, Then for all for all hyperrectangles $A\subset \Omega$, we have the following inequality regarding impurity decrease under true and empirical measure:
% Suppose $\lambda_n \ge n^{-1/2+\beta}$ and the density of $X$ exists, then for all hyperrectangles $A\subset \Omega$, let $A_1, A_2$ be the two child nodes of $A$ obtained by an arbitrary split, and let $n'$ be the number of samples in $A$. To highlight the dependence of impurity decrease and the child node $A_1, A_2$, denote the impurity decrease on node $A$ after spliting $A$ into $A_1, A_2$, under measure $\mathbb{P}$ and $\mathbb{P}_n$, as $\Delta I(A,A_1,\mathbb{P})$ and $\Delta I(A,A_1,\mathbb{P}_n)$, respectively. Then we have
$$\lim_{n\to\infty}\mathbb{P}\left(\sup_{(z,j)} |\Delta I(A,(z,j),\mathbb{P}) - \Delta I(A, (z,j),\mathbb{P}_n)|>\frac{c_0\lambda_n n}{2n'}\right) \le \left[\frac{8}{c} n^{1/2-\beta} + 2\right] \exp(-c^2n^\beta),$$
where the $c=\frac{c_0}{44}$.
\end{lemma}
\begin{proof}[Proof of Lemma \ref{noise thre}]
Since $0\le \Delta I(A, (z,j),\mathbb{P}) , \Delta I(A, (z,j),\mathbb{P}_n) \le 1$, the lemma automatically holds when $c_0\lambda_n n/n'>1$. We only need to consider the case when $c_0\lambda_n n/n'\le1$. The proof can be divided into three steps. In the first step, we define an $\epsilon$-net on the space of all possible splits $(z,j)$, with the distance being the symmetric set difference. We prove this $\epsilon$-net also corresponds to a $5\epsilon$-net of $\Delta I(A,(z_j),\mathbb{P})$, while with high probability corresponds to a $10\epsilon$-net of $\Delta I(A,(z,j),\mathbb{P}_n)$. In the second step, we show with high probability, the difference between $\Delta I(A,(z,j),\mathbb{P})$ and $\Delta I(A,(z,j),\mathbb{P}_n)$ can be bounded uniformly for all $(z,j)$ in the $\epsilon$-net. The last step combines the results of the previous two steps, calculating tree impurity decrease from impurity decrease at node $A$, giving a uniform bound of $ |\Delta I(T,\mathbb{P}) - \Delta I(T,\mathbb{P}_n)|$. 

\paragraph{Step 1}
Denote the hyperrectangle $A$ as $A = [x_{1,l}, x_{1,r}]\times[x_{2,l}, x_{2,r}]\times \cdots \times[x_{d,l},x_{d,r}]$. 
% Without loss of generality, we assume $(x_{1,l}, x_{2,l},\ldots, x_{d,l})\in A_1$, i.e., after a split at $A$, $A_1$ always contains the ``left corner'' of $A$. Let $n'$ be the number of training samples inside $A$. Let $\epsilon = c\lambda_n n/n'$, and $m = \lceil 1/\epsilon \rceil$. 
For each feature $X[j]$, define a series of points $x_{j,0}< x_{j,1}< \ldots < x_{j,m}$ such that $x_{j,0} = x_{j,l}$, $x_{j,m} = x_{j,r}$,
$$\mathbb{P}(X[j]\in [x_{j,k},x_{j,k+1}] |X\in A) = 1/m, \;\;\forall 0\le k \le m-1.$$
That is, $x_{j_1}, x_{j_2}, \ldots , x_{j_{m-1}}$ are all the $m$th quantiles of the marginal probability measure of $X[j]$. 
Define a collection of splits on $A$ as
$$\mathscr{S} = \big\{ (z,j): 1\le j\le d, z \in\{x_{j_1}, x_{j_2}, \ldots , x_{j_{m-1}}\} \big\}.$$
% Define a collection of subsets of $A$ as:
% $$\mathscr{A} = \big\{A = [x_{1,l}, x_{1,r}]\times[x_{2,l}, x_{2,r}]\times \cdots \times[x_{j-1,l}, x_{j-1,r}]
% \times[x_{j,0},x_{j,k}]\times[x_{j+1,l}, x_{j+1,r}]\times \cdots \times[x_{d,0},x_{d,1}]: \; 1\le j\le d, 1\le k \le m-1\big\}.$$
% That is, $\mathscr{A}$ contains all the possible $A_1$ that are obtained by spliting at $X[j] = x_{j,k}, \; 1\le j\le d, 1\le k \le m-1$. By the definition of $x_{j,k}$, for all $A_1$ that are obtained by a single split at $A$, there exists $A_1'\in\mathscr{A}$, such that $\mathbb{P}(A_1\Delta A_1'|A) \le 1/m \le \epsilon$. 
Suppose $A_1$ is the left leaf node after splitting $A$ at $X[j]=z$ with $(z,j)\in\mathscr{S}$, then by definition of $\mathscr{S}$, there exists $(z',j)\in\mathscr{S}$ splitting $A$ with left leaf node $A'_1$, such that $\mathbb{P}(A'\Delta A'_1)\le 1/m\le \epsilon$. Thus by technical lemma \ref{nodes 2 impu}, the set 
$$\{\Delta I(A,(z,j),\mathbb{P}): (z,j)\in\mathscr{S}\}$$
forms a $5 \epsilon$-net in the space of $\Delta I(A,(z,j),\mathbb{P})$.
Define another collection of subsets of $A$ as:
\begin{align*}
\mathscr{B} = \big\{A = & [x_{1,l}, x_{1,r}]\times[x_{2,l}, x_{2,r}]\times \cdots   \times[x_{j-1,l}, x_{j-1,r}]
\times[x_{j,k},x_{j,k+1}]\times[x_{j+1,l}, x_{j+1,r}]\times \cdots \times[x_{d,0},x_{d,1}]: \\
 & \; 1\le j\le d, 0\le k \le m-1\big\}.    
\end{align*}
We have for all $B\in\mathscr{B}$, $\mathbb{P}(B|A) = 1/m<\epsilon$. By Hoeffding's inequality, $\forall B\in \mathscr{B}$,
$$\mathbb{P}\left(\Big|\frac{n\mathbb{P}_n(B)}{n'}-\frac{1}{m}\Big| \ge \epsilon\right) \le 2\exp(-2n'\epsilon^2) = 2\exp(-2c^2\lambda_n^2 n^2/n') \le 2\exp(-2c^2\lambda_n^2 n)\le 2\exp(-2c^2n^{2\beta}).$$
Because $1/\epsilon = (c\lambda_n n/n' )^{-1}\le n^{1/2-\beta}/c$, applying a union bound for all the $B\in\mathscr{B}$, we have with probability greater than $1-2n^{1/2-\beta} \exp(-2c^2n^{2\beta})/c$ that
$$\sup_{B\in\mathscr{B}}\mathbb{P}_n(B|A) \le 2\epsilon.$$
Therefore, by Technical Lemma \ref{nodes 2 impu}, the set 
$$\{\Delta I(A,(z,j),\mathbb{P}_n): (z,j)\in\mathscr{S}\}$$
forms a $10\epsilon$-net in the space of $\Delta I(A,(z,j),\mathbb{P}_n)$.

\paragraph{Step 2}
% Denote $p_{1} = \mathbb{E}(Y|X\in A_1)$, $\hat p_{1}=\mathbb{E}_{n}(Y|X\in A_1)$, $p_{2} = \mathbb{E}(Y|X\in A_2)$, $\hat p_{2}=\mathbb{E}_{n}(Y|X\in A_2)$, $p = \mathbb{E}(Y|X\in A)$, $\hat p=\mathbb{E}_{n}(Y|X\in A)$, 
Recall for set $A\subset\Omega$, we denote $\eta(A) = \mathbb{E}(Y|A)$.
Denote $A_1=\{X\in A: X[j]<z\}$ and $A_2=\{X\in A: X[j]\ge z\}$,
% For split $(z,j)$, for short of notation, denote $\eta_1=\mathbb{E}(Y|X\in A, X[j]<z)$ and $\eta_2=\mathbb{E}(Y|X\in A, X[j]\ge z)$. Denote their empirical counterpart as $\eta_{n,1}$ and $\eta_{n,2}$.
then we have
\begin{align}\label{Delta I decompose}
 & |\Delta I(A, (z,j),\mathbb{P}) - \Delta I(A, (z,j),\mathbb{P}_n)| \nonumber\\
	= & \big|2\eta(A)(1-\eta(A)) - 2\eta(A_1)(1-\eta(A_1)) \mathbb{P}(A_1|A) - 2\eta(A_2)(1-\eta(A_2))\mathbb{P}(A_2|A) \nonumber  \\
		& - 2\eta_n(A)(1-\eta_n(A)) + 2\eta_n(A_1)(1-\eta_n(A_1)) \mathbb{P}_{n}(A_1|A) + 2\eta_n(A_2)(1-\eta_n(A_2))\mathbb{P}_{n}(A_2|A) \big|\nonumber \\
	\le & 2|(\eta(A)-\eta_n(A))(1-\eta(A)-\eta_n(A))| + \nonumber\\
		& 2|\eta(A_1)(1-\eta(A_1)) \mathbb{P}(A_1|A) - \eta_n(A_1)(1-\eta_n(A_1)) \mathbb{P}_{n}(A_1|A)| + \nonumber\\
		& 2|\eta(A_2)(1-\eta(A_2)) \mathbb{P}(A_2|A) - \eta_n(A_2)(1-\eta_n(A_2)) \mathbb{P}_{n}(A_2|A)| \nonumber\\
	\le & 2|\eta(A) - \eta_n(A)| + \frac{1}{2}|\mathbb{P}(A_1|A) - \mathbb{P}_{n}(A_1|A)| + 
				2\mathbb{P}(A_1|A)|\eta(A_1) - \eta_n(A_1)| + \nonumber\\
		& \frac{1}{2}|\mathbb{P}(A_2|A) - \mathbb{P}_{n}(A_2|A)| + 
				2\mathbb{P}(A_2|A)|\eta(A_2) - \eta_n(A_2))|
\end{align}
For all $A_1\in\mathscr{A}$, by Hoeffding's inequality, 
\begin{equation}\label{PnA}
\mathbb{P}\left(\Big|\frac{n\mathbb{P}(A_1)}{n'}-\frac{\mathbb{P}(A_1)}{\mathbb{P}(A)}\Big| \ge \epsilon\right) \le 2\exp(-2n'\epsilon^2) \le 2\exp(-2c^2n^{2\beta}).
\end{equation}
Noting $\mathbb{P}(A_1|A) - \mathbb{P}_{n}(A_1|A) + \mathbb{P}(A_2|A) - \mathbb{P}_{n}(A_2|A) = 1-1=0$, we have $|\mathbb{P}(A_2|A) - \mathbb{P}_{n}(A_2|A)|=|\mathbb{P}(A_1|A) - \mathbb{P}_{n}(A_1|A)|$. Thus the same bound as equation (\ref{PnA}) applies to $|\mathbb{P}(A_2|A) - \mathbb{P}_{n}(A_2|A)|$. 
It remains to bound the term $\mathbb{P}(A_1|A)|\eta(A_1) - \eta_n(A_1)|$ and $\mathbb{P}(A_2|A)|\eta(A_2) - \eta_n(A_2))|$. 
For all $A_1\in\mathscr{A}$, let $\epsilon_{A_1} = \epsilon/\sqrt{\mathbb{P}(A_1|A)}$, then by Hoeffding's inequality,
$$
\mathbb{P}\left(\big|\mathbb{E}_n(Y|A_1) -\mathbb{E}(Y|A_1)\big| \ge \epsilon_{A_1}\right) \le 2\exp(-2n'\mathbb{P}(A_1|A)\epsilon^2_{A_1}) = 2\exp(-2n'\epsilon^2) \le 2\exp(-2c^2n^{2\beta}).$$
Therefore 
\begin{align}\label{p1alpha}
\mathbb{P}\left(\mathbb{P}(A_1|A)|\eta(A_1) - \eta_n(A_1)| \ge \epsilon\right)   
  %  & \le \mathbb{P}\left(\mathbb{P}(A_1|A)|\eta(A_1) - \eta_n(A_1)| \ge \epsilon\right) \nonumber \\
    & \le \mathbb{P}\left(\sqrt{\mathbb{P}(A_1|A)}|\eta(A_1) - \eta_n(A_1)| \ge \epsilon\right) \nonumber \\
    & = \mathbb{P}\left(|\eta(A_1) - \eta_n(A_1)| \ge \epsilon_{A_1}\right) \nonumber \\
    & \le 2\exp(-2c^2n^{2\beta}).
\end{align}
Similarly, we have
\begin{equation}\label{palpha}
\mathbb{P}\left(\mathbb{P}(A_2|A)|\eta(A_2) - \eta_n(A_2))| \ge \epsilon\right) \le 2\exp(-2c^2n^{2\beta}),\;\;\;
\mathbb{P}\left(\big|\eta_n(A) - \eta(A)\big| \ge \epsilon\right) \le 2\exp(-2c^2n^{2\beta}).
\end{equation}
Combining equation (\ref{Delta I decompose}), (\ref{PnA}), (\ref{p1alpha}) and (\ref{palpha}), and applying a union bound, we have with probability greater than 
$$1-\Big[\frac{6}{c}n^{1/2-\beta}+2\Big] \exp(-2c^2n^{2\beta}),$$
for all $(z,j)\in\mathscr{S}$, the difference between $\Delta I(A,(z,j),\mathbb{P})$ and $\Delta I(A,(z,j),\mathbb{P}_n)$ is bounded by
$$\sup_{(z,j)\in\mathscr{S}}|\Delta I(A,(z,j),\mathbb{P})-\Delta I(A,(z,j),\mathbb{P}_n)|\le  (2+1/2+2+1/2+2)\epsilon = 7\epsilon.$$

\paragraph{Step 3}
Let $\sigma(n)$ be
\begin{equation}\label{sigma n}
\sigma(n) = \Big[\frac{8}{c}n^{1/2-\beta}+2\Big] \exp(-c^2n^\beta).
\end{equation}
Combining the results of step 1 and step 2, with probability greater than $1-\sigma(n)$, $\forall (z,j)$, there exists $(z',j)\in\mathscr{S}$, such that the following events hold simultaneously
$$|\Delta I(A,(z,j),\mathbb{P}) - \Delta I(A,(z',j),\mathbb{P})| < 5\epsilon, $$
$$|\Delta I(A,(z,j),\mathbb{P}_n) - \Delta I(A,(z',j),\mathbb{P}_n)| < 10\epsilon, $$
$$|\Delta I(A,(z',j),\mathbb{P}) - \Delta I(A,(z',j),\mathbb{P}_n)| < 7\epsilon, $$
Therefore
$$
|\Delta I(A,A_1,\mathbb{P}) - \Delta I(A,A_1,\mathbb{P}_n)| < 22\epsilon.$$
Letting $c=\frac{c_0}{44}$ and noticing $\lim_{n\to\infty}\sigma(n) = 0$, we finish the proof.
\end{proof}

We are now prepared to prove Theorem \ref{feature selection}.
\begin{proof}[Proof of Theorem \ref{feature selection}]
Let $T_k^*$ be the theoretical tree with $k$ leaf nodes defined in Algorithm 2 of the supplementary material, and let $\hat{T}_{n,k}$ be the tree consisting of the first $k$ splits of tree $\hat{T}_n$. The proof of theorem \ref{feature selection} consists of two steps. In the first step, we show there exists $k_0\in\mathbb{N}$, such that as $n$ goes to infinity, the probability of $\hat{T}_{n,k_0}$ including all non-redundant features and excluding all redundant features goes to one. In the second step, we show if all the non-redundant features are already included in the tree, the probability of including any new redundant features in all the following split procedures goes to zero. 

We begin with the first step. By Lemma \ref{theoretical exclude}, the theoretical tree $T_k^*$ will not include any redundant features. We now show all non-redundant features are included in $T_k^*$ provided $k$ is large enough. Let $f_k^*$ be the classifier associated with tree $T_k^*$. Let $\psi_k(x)$ be the leaf node of $T_k^*$ that contains $x$, by Technical Lemma \ref{theoretical consistent}, we have $\lim_{k\to\infty}\eta(\psi_k(x)) =\eta(x), \forall x\in \Omega$. Suppose there exists a non-redundant feature $X[j]$, such that $X[j]$ is excluded in splits of $T_k^*,\; \forall k\in\mathbb{N}$. Then let $\psi_\infty(x) = \lim_{k\to\infty}\psi_k(x)$, write $\psi_\infty(x)=[a_{\infty,1}(x),b_{\infty,1}(x)]\times[a_{\infty,2}(x),b_{\infty,2}(x)]\times\cdots\times[a_{\infty,d}(x),b_{\infty,d}(x)]$ and denote the projection of $\psi_\infty(x)$ into $X[-j]$ as $\psi_\infty^{(-j)}(x)$, we have $\psi_\infty(x)=\psi_\infty^{(-j)}(x)\times[0,1]$.
By Technical Lemma \ref{theoretical consistent}, $\forall x\in\Omega$, $\eta(x)$ is constant in $\psi_\infty(x)$. Therefore when $X[-j]\in\psi_\infty^{(-j)}(x)$, $\mathbb{E}(Y|X[-j]) = \mathbb{E}(Y|X)$. 
Noting $\cup_{x\in\Omega}\psi_\infty(x)=\Omega$, we have $\cup_{x\in\Omega}\psi_\infty^{(-j)}(x)=[0,1]^{d-1}$. 
% Noting the above argument hold for all $x\in\Omega$ 
Thus for all $X[-j]\in[0,1]^{d-1}$, $\mathbb{E}(Y|X[-j]) = \mathbb{E}(Y|X)$. This contradicts with the assumption that $j$th feature is a non-redundant feature.
% Then by definition of theoretical tree, given $X[-j]$, $e_k(X)$ is constant with respect to $X[j]$. Thus we have $\mathbb{E}(e_k(X)|X[-j]) = e_k(X),\;\forall k\in \mathbb{N}$. Because $|e_k(X)|$ is upbounded by $1$, by dominated convergence theorem, we have $$\mathbb{E}(Y|X) = \lim_{k\to\infty}e_k(X) = \lim_{k\to\infty}\mathbb{E}(e_k(X)|X[-j]) = \mathbb{E}(\mathbb{E}(Y|X)|X[-j]) = \mathbb{E}(Y|X[-j]),\; a.s.$$
% where the last inequality uses the fact that the sigma algebra of $X[-j]$ is contained in the sigma algebra of $X$. Hence we have $\mathbb{E}(Y|X[-j]) = \mathbb{E}(Y|X)$. This contradicts the fact that $X[j]$ is a non-redundant feature. 
Therefore, for all non-redundant features $X[j]$, there exists some $k_j\in\mathbb{N}$, such that feature $X[j]$ is included in splits of $T_k^*$ for $k\ge k_j$.  Since the number of features is finite, there exists $k_0\in\mathbb{N}$, such that $T_{k_0}^*$ includes all non-redundant features. 
By Technical Lemma 4 in the supplementary material, as $n$ goes to infinity, with probability tending to one, $\hat T_{n,k_0}$ includes all non-redundant features while excluding all redundant features.

\par{}
It remains to show no redundant features will be included after the $k_0$th split. For the $k'$th ($k_0<k'\le \bar{a}_n$) split, which splits node $A$ into $A_1, A_2$, denote $\Delta I_{n}(T,\mathbb{P}_n)$ and $\Delta I_{r}(T,\mathbb{P}_n)$ as the maximal tree impurity decrease measured by $\mathbb{P}_n$ when the split is in a non-redundant feature and redundant feature, respectively. Then by Algorithm 1, the split at a redundant feature will be rejected if
$\Delta I_{r}(T,\mathbb{P}_n) \le \Delta I_{n}(T,\mathbb{P}_n) + c_0\lambda_n.$
Similarly, denote $\Delta I_{n}(A,\mathbb{P}_n)$ and $\Delta I_{r}(A,\mathbb{P}_n)$ as the maximal impurity decrease of node $A$ measured by $\mathbb{P}_n$ when the split is in a non-redundant feature and redundant feature, respectively. Further denote $\Delta I_{n}(A,\mathbb{P})$ and $\Delta I_{r}(A,\mathbb{P})$ as the maximal impurity decrease of node $A$ when the probability measure is the true probability $\mathbb{P}$.
By lemma \ref{theoretical exclude}, for the impurity decrease on node $A$ measured under true probability $\mathbb{P}$, we have
$\Delta I_{r}(A,\mathbb{P}) \le \Delta I_{n}(A,\mathbb{P}).$
By lemma \ref{noise thre}, we have with probability greater than $1-\sigma(n)$,
$$|\Delta I_{n}(A,\mathbb{P}_n) - \Delta I_{n}(A,\mathbb{P})|\le \frac{c_0\lambda_n n}{2n'},$$
$$|\Delta I_{r}(A,\mathbb{P}_n) - \Delta I_{r}(A,\mathbb{P})|\le \frac{c_0\lambda_n n}{2n'}.$$
Therefore, we have with probability greater than $1-\sigma(n)$,
\begin{align*}
 & \Delta I_{r}(T,\mathbb{P}_n)-\Delta I_{n}(T,\mathbb{P}_n) \\
	 = & \mathbb{P}_{n}(A) [\Delta I_{r}(A,\mathbb{P}_n) -\Delta I_{n}(A,\mathbb{P}_n)] \\
	 \le & \mathbb{P}_{n}(A) [|\Delta I_{r}(A,\mathbb{P}_n) -\Delta I_{r}(A,\mathbb{P})| + 
	 \Delta I_{r}(A,\mathbb{P}) -\Delta I_{n}(A,\mathbb{P}) +
	 |\Delta I_{n}(A,\mathbb{P}) -\Delta I_{n}(A,\mathbb{P}_n)|] \\
		\le &  \frac{n'}{n} \Big(\frac{c_0\lambda_n n}{2n'} + 0 + \frac{c_0\lambda_n n}{2n'}\Big) \\
		\le & c_0\lambda_n.
\end{align*}
Therefore, for all $k'$ ($k_0<k'\le \bar{a}_n$), the probability of splitting at a redundant feature is smaller than $\sigma(n)$. From equation (\ref{sigma n}), we have $\lim_{n\to \infty} n\sigma(n) = 0$.
Since $\bar{a}_n\le n$, with probability tending to one, $\hat{T}_n$ will not split at redundant features.
\end{proof}

\subsection{Proof of Corollary 1}
We first prove the results of Lemma 4 still hold under the condition of corollary 1.
\begin{tech lemma}\label{theoretical exclude 2}
If Condition 1 holds with $c_3=0$, for all hyperrectangles $A\subset \Omega$, we have
    $$\sup_{1\le j\le q} M_{A,j} \ge \sup_{q+1\le l\le d} M_{A,l},$$
where $M_{A,j}$ is the maximal impurity decrease of feature $j$ at node $A$.
\end{tech lemma}
\begin{proof}[Proof of Technical Lemma \ref{theoretical exclude 2}]
By definition of $M_A^*$, we have $M_A^* \ge M_{A,l}$, $\forall q+1 \le l \le d$. Since $c_3 = 1$, we have $\sup_{1\le j\le q} M_{A,j} = M_A^*$, therefore $\sup_{1\le j\le q} M_{A,j} \ge M_{A,l},\; \forall q+1\le l \le d$.
\end{proof}
The proof of corollary 1 using Technical Lemma \ref{theoretical exclude 2} and Lemma 6 is the same as the proof of Theorem 2 using Lemma 5 and Lemma 6.

\subsection{Proof of $c_3$ Values in Examples 1-2}
\subsubsection{Example 1}
\begin{proof}

Using similar notation as in the proof of Theorem 2, let $\mathbb{P}(A_{y1})=V_1$, $\mathbb{P}(A_{y2})=V_2$, and $p = \mathbb{E}(Y|X\in A)$, $y_{1} = \mathbb{E}(Y|X\in A_{y1})$, $y_{2}=\mathbb{E}(Y|X\in A_{y2})$.
Without loss of generality, we assume $\mathbb{P}(A_{y2}) = h \le \mathbb{P}(A_{y1})$, and $y_{1} \le y_{2}$. Since $\mathbb{E}(Y|X)$ is a monotonically increasing function of $a^T+b$, there exists $t\in\mathbb{R}$, such that $\forall X_1\in A_{y1}, X_2\in A_{y2}$,  we have $a^T X_1 + b \le t \le a^T X_2 + b$.

Let $A_1, A_2$ be two hyperrectangles obtained by splitting at one of the features, which satisfies $\mathbb{P}(A_1) = V_1$ and $\mathbb{P}(A_2) = V_2$. The feature to split is selected by minimizing $\mathbb{P}(A_1\backslash A_{y1})$. Denote $\mathbb{P}(A_1\backslash A_{y1})$ as $V'$. 
% The expectation of $Y$ (under $\mathbb{P}$) in $A_1, A_2$ are denoted as $p_{1} = \mathbb{E}(Y|A_1)$, $p_{2} =  \mathbb{E}(Y|A_2)$. 
Denote the set difference as
$A_{y1}\backslash A_1 \triangleq B,\;\; A_1\backslash A_{y1} \triangleq C.$ 
% Further denote the expectation of $Y$ (under $\mathbb{P}$) in $B, C$ as
% $$\eta(B) = \mathbb{E}(Y|X\in B), \quad 
% \eta(C) = \mathbb{E}(Y|X\in C).$$
Because $\mathbb{P}(A_1) =V_1= \mathbb{P}(A_{y1})$ and $\mathbb{P}(A_2)=V_2=\mathbb{P}(A_{y2})$, we have
$$\eta(A_1) = y_{1} + (\eta(C)-\eta(B))\frac{V'}{V_1},\;\; \eta(A_2) = y_{2} + (\eta(B)-\eta(C))\frac{V'}{V_2}.$$
Therefore
\begin{equation}\label{p1-p2 1}
|\eta(A_2) - \eta(A_1)| \ge |y_{2}-y_{1}| - |\eta(B)-\eta(C)|\Big(\frac{V'}{V_1}+\frac{V'}{V_2}\Big).
\end{equation}
We first show $|\eta(B)-\eta(C)|\le |y_{2}-y_{1}|$. Because $\mathbb{E}(Y|X)$ is a monotonically increasing function of $a^T X + b$, we have $\eta(B)\le y_{2}$. We then compare $\eta(B)$ and $y_{1}$. Let $t_B = \inf_{X\in B}(a^T X + b)$ and $A_{y1}'$ be
$$A_{y1}' = \left\{ X\in A_{y1}: a^T X + b \ge t_B\right\}.$$
That is, $A_{y1}'$ is the subset of $A_{y1}$ whose $a^T X + b$ values are no smaller than the infimum of $a^T X+b$ values in $B$. Let $y_{1}' = \mathbb{E}(Y|X\in A_{y1}')$,
Since $\mathbb{E}(Y|X)$ is a monotonically increasing function of $a^T X + b$, we have $y_{1} \le y_{1}'$. Because the marginal of $\mathbb{P}$ is uniform on $\Omega$, we have
$$\eta(B) = \frac{1}{\mathbb{P}(B)}\int_{t_B}^t \phi(z) \rho_B(z) dz,$$
where
$$\rho_B(z) = \lim_{\Delta z\to 0}\frac{\mathbb{P}(X\in B: a^T X+b\in (z-\Delta z, z))}{\Delta z}.$$
Similarly, we have
$$y_{1}' = \frac{1}{\mathbb{P}(A_{y1}')}\int_{t_B}^t \phi(z) \rho_{1}'(z) dz,$$
where
$$\rho_{1}'(z) = \lim_{\Delta z\to 0}\frac{\mathbb{P}(X\in A_{y1}': a^T X+b\in (z-\Delta z, z))}{\Delta z}.$$
Both $\eta(B)$ and $y_{1}'$ are averages of $\phi(z)$ from $t_B$ to $t$, 
with the weights being $\rho_B(z)/\mathbb{P}(B)$ and $\rho_1'(z)/\mathbb{P}(A_{y1}')$, respectively. 
By the construction of sets $B$ and $A_{y1}'$,  we can see $[\rho_1'(z)/\mathbb{P}(A_{y1}')] / [\rho_B(z)/\mathbb{P}(B)]$ is decreasing at $(t_B, t)$. Thus we have $y_{1}'\le \eta(B)$. Recalling that $y_{1} \le y_{1}'$ and $\eta(B)\le y_{2}$, we have
$$y_{1}\le \eta(B) \le y_{2}.$$
Similarly,  we have $y_{1}\le \eta(C) \le y_{2}$. Thus we have proved $|\eta(B)-\eta(C)|\le |y_{2}-y_{1}|$. Combining this with equation (\ref{p1 p2 3}) and (\ref{p1-p2 1}), $c_3$ can be lower bounded by
\begin{equation}\label{c1 bound 1}
c_3 \ge 1 - \frac{|\eta(B)-\eta(C)|(V'/V_1 + V'/V_2)}{|y_{2} - y_{1}|} \ge 1-\frac{V'}{V_1} - \frac{V'}{V_2}.
\end{equation}
We then consider $V'/V_1+V'/V_2$. For simplicity of notation, we can assume $A$ is the hypercube $[0,1]^d$. If $A$ is not this hypercube, we can make a linear transformation on $X$ so that $A$ becomes hypercube $[0,1]^d$ under the transformed $X$. All the conditions in this example are invariant to linear transformation. It suffices to consider the worse case scenarios when $V'/V_1+V'/V_2$ are largest, which is when $a$ is parallel to $(1, 1, \ldots 1)^T$. In that case, the angle between $a$ and each feature is equally large and the least angle between $a$ and all features reaches the maximum. We discuss $V'/V_1+ V'/V_2$ of the worst case scenario in two situations: $h\le 1/d!$ and $h>1/d!$.

\begin{figure}[h]
     \centering
     \begin{subfigure}[b]{0.48\textwidth}
         \centering
         \includegraphics[width=\textwidth]{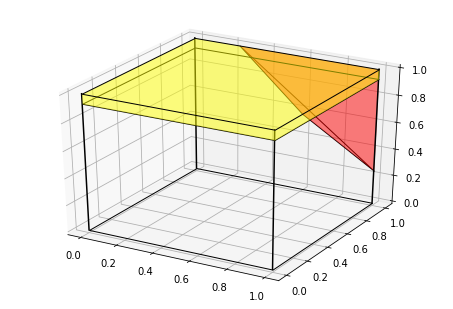}
         \caption{$h\le 1/d!$}
     \end{subfigure}
     \begin{subfigure}[b]{0.48\textwidth}
         \centering
         \includegraphics[width=\textwidth]{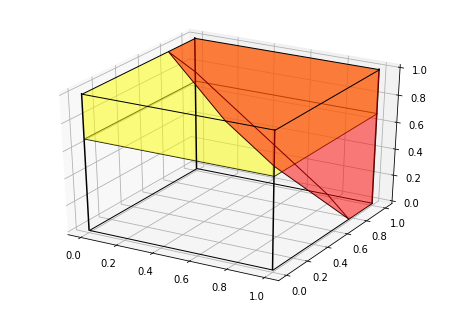}
         \caption{$h> 1/d!$}
	\end{subfigure}
    \caption{Diagrams for two situations when $d=3$. The set $A_{y2}$ is colored in red while the set $A_2$ is colored in yellow. The set $A_{y2}\cap A_2$ is colored in orange.}\label{example fig}
\end{figure}

\paragraph{Situation One: $\bm{h\le1/d!}$} In the worst case when $a$ is parallel to $(1, 1, \ldots 1)^T$, the set $A_{y2}$ is a hyperpyramid as shown in Figure \ref{example fig}(a). Denote the length of edge of $A_{y2}$ which is on the edge of hypercube $[0,1]^d$ as $e$. Because $\mathbb{P}(A_2) = \mathbb{P}(A_{y2})$, we have
$$h = \frac{e^d}{d!}.$$
Therefore
$$
\frac{V'}{V_1} + \frac{V'}{V_2} = \frac{V'}{(1-h)h} = \frac{\big[1-h^{\frac{d-1}{d}}(d!)^{-\frac{1}{d}}\big]^d}{1-h} = 1 - d(d!)^{-\frac{1}{d}} h^{\frac{d-1}{d}} + h + o(h).$$
Because $h\in(0,1/2]$, we have $h^{d-1}{d}\ge g$. Let $g(d) = d(d!)^{-\frac{1}{d}}$. Then 
$$\log g(d) = \log d - \frac{1}{d}\sum_{i=1}^d \log i = \frac{1}{d} \sum_{i=1}^d (\log d - \log i).$$
Therefore $g(d)$ is an increasing function. The minimal number of dimensions to consider feature selection is $d=2$, so we have $g(d)\ge g(2) = \sqrt{2}$. Therefore we have
$$\frac{V'}{V_1} + \frac{V'}{V_2} \le 1 - (\sqrt{2}-1)h + o(h).$$
Thus when $h\le1/d!$ we have
$$c_3 \ge (\sqrt{2}-1)h + o(h).$$

\paragraph{Situation Two: $\bm{h>1/d!}$}
Because $h\in (0,1/2]$, situation two can only happen if $d\ge 3$. We still consider the worst case scenario when $a$ is parallel to $(1,1,\ldots 1)^T$ as shown in Figure \ref{example fig}(b). Suppose $A$ is split at feature $j$ to obtain $A_1, A_2$, then we say the feature $j$ is the direction of \textbf{height}, while all the other features form the space of \textbf{base}. Denote the volume of cross section between hyperplane $X[j]=t$ and $A_{y2}$ as $S(t)$. Then we have
$$h = \int_0^1 S(t) dt, \;\; V' = \int_0^{1-h} S(t)dt.$$
Denote $S_{up}, S_{low}$ as
$$S_{up} = \frac{1}{h}\int_{1-h}^1 S(t) dt, \;\; S_{low} = \frac{1}{1-h}\int_0^{1-h} S(t) dt.$$
That is, $S_{up}$ is the average base area of the set $A_{y2}\cap A_2$, and $S_{low}$ is the average base area of the set $C$. We have
$$c_3 \ge 1 - \frac{V'}{V_1} - \frac{V'}{V_2} = 1 - \frac{S_{low}(1-h)}{h(1-h)} = 1 - \frac{S_{low}}{h}.$$
Noting $S_{low}(1-h) + S_{up}h = h$, we have
\begin{equation}\label{c1 Sup-Slow}
c_3 \ge S_{up} - S_{low}.
\end{equation}
Let $F_d(x)$ be the cumulative distribution function of the Irwin-Hall distribution with parameter $d$. By definition $F_d(x)$ is the probability that the sum of $d$ independent uniform random variables is no greater than $x$. Let $z = F_d^{-1}(h)$, then we have $S(1) = F_{d-1}(z)$. Thus $S(t) = F_{d-1}(z+t-1), \;\forall t\in[0,1]$. Because $S(t)$ is increasing with $t$, we have $S(1) = F_{d-1}(z) \ge h \ge S(0) = F_{d-1}(z-1)$. Thus $z\in(F^{-1}_{d-1}(h)-1, F^{-1}_{d-1}(h)+1)$. To compare $S_{up}$ and $S_{low}$, define a new quantity as 
$$S_m \triangleq \frac{1}{h} \int_{1-2h}^{1-h} F_{d-1}(z-1+t) dt.$$
Because $F_{d-1}(z)$ is increasing with $z$, we have $S_m \ge S_{low}$. The difference of $S_{up}$ and $S_m$ can be written as
\begin{align*}
S_{up}-S_{m} = & \frac{1}{h} \int_0^h [F_{d-1}(z-h+t) - F_{d-1}(z-2h+t)]dt \\
							\ge & \inf_{t\in[0,h]} [F_{d-1}(z-h+t) - F_{d-1}(z-2h+t)]  \\
							= & \inf_{t\in[0,h]} \int_0^h F'_{d-1}(z-2h+s+t)ds,
\end{align*}
where $F'_{d-1}(z)$ is the density function of the Irwin-Hall distribution with parameter $d-1$. Because $F'_{d-1}(z)$ increases at $[0, (d-1)/2]$ and decreases at $[(d-1)/2, d-1]$, the minimum of $\int_0^h F'_{d-1}(z-2h+s+t)ds$ is achieved at either $t=0$ or $t=h$. For $t=0$, 
$$\int_0^h F'_{d-1}(z-2h+s+t)ds =\int_{z-2h}^{z-h} F'_{d-1}(z)ds.$$
For $t=h$, 
$$\int_0^h F'_{d-1}(z-2h+s+t)ds =\int_{z-h}^{z} F'_{d-1}(z)ds.$$
Because $F'_{d-1}(z)$ is symmetric with respect to $z=(d-1)/2$, $\int_{z-2h}^{z-h} F'_{d-1}(z)ds$ will be no greater than $\int_{z-h}^{z} F'_{d-1}(z)ds$ if $z-h \le (d-1)/2$. Noting
$z - h = z - F_d(z)$, 
$\frac{\partial}{\partial z}(z-h) = 1 - F'_d(z) \ge 0$
and $[z - F_d(z)]|_{z=d/2} = (d-1)/2$, we have proved 
$$S_{up}-S_{m} \ge \int_{z-2h}^{z-h} F'_{d-1}(z)ds = F_{d-1}(F_d^{-1}(h) -h) - F_{d-1}(F_d^{-1}(h) -2h).$$
Thus 
\begin{equation}\label{Sup-Slow}
S_{up}-S_{low} \ge F_{d-1}(F_d^{-1}(h) -h) - F_{d-1}(F_d^{-1}(h) -2h).
\end{equation}
Combining equation (\ref{c1 Sup-Slow}) and (\ref{Sup-Slow}), we have
$$c_3 \ge F_{d-1}(F_d^{-1}(h) -h) - F_{d-1}(F_d^{-1}(h) -2h).$$
We finish the proof. 
\end{proof}

\subsubsection{Example 2}
\begin{proof}
Let $A = \Omega$, and rectangles $A_1, A_2$ be obtained by an arbitrary split at feature $X[1]$ or $X[2]$. Then it's easy to see $\mathbb{P}(Y|A_1) = \mathbb{P}(Y|A_2) = 0.5.$. Noting $\mathbb{P}(Y|A) = \mathbb{P}(Y) = 0.5$, the impurity does not change after splitting, thus $c_3 = 0$.
\end{proof}

\section{Supplementary Materials for Algorithm 1}
\label{sec:algorithm1}
\subsection{Details of Algorithm 1}
Algorithm 1 in the main paper is represented in details here in Algorithm \ref{algo:3}.

\begin{algorithm}[h!] 
\SetAlgoLined
\KwResult{Output the fitted tree}

Input training data $\{(X_i, Y_i)\}_{i=1}^n$, impurity function $f(\cdot)$,  weight for minority class $\alpha$, SVR penalty parameter $\lambda_n$, and maximal number of leaf nodes $\bar{a}_n\in\mathbb{N}$. Let the root node be $\Omega$, and $\mr{node.X} = \{X_i\}_{i=1}^n$, $\mr{node.Y} = \{Y_i\}_{i=1}^n$. Let $\mr{node\_queue} = [\mr{root}]$\;

Denote the $j$th coordinate of the $i$th sample as $X[j]_i$. For $j=1,\ldots,d$, sort $\{X_i\},1\le i\le n$ by their $j$th feature $\{X[j]_i\}, 1\le i\le n$. Denote the sorted increasing subscripts of $\{X_i\}_{i=1}^n$ as $(j_1, j_2, \cdots j_n)$, i.e., $X[j]_{j_1} \le X[j]_{j_2} \cdots \le X[j]_{j_n}$\;

Set $R'=+\infty$. For each split, we compare the risk of the split tree to $R'$, accepting the new split if the risk after the new split is smaller than $R'$\;

\While{$\mr{node\_queue}$ is not empty and number of leaf nodes $\le \bar{a}_n$}{
		Dequeue the first entity in $\mr{node\_queue}$, denoting it as $\mr{node}$. Denote the sample size in $\mr{node}$ as $n'$\;
        Denote the number of features that have already been split as $d'$. Rearrange $1, 2, \cdots d$ into a list $J_0$, such that the first $d'$ elements in $J_0$ corresponds to indices of features that have already been split. Let $\Delta I_0 = 0$\;
		\For{$j$ in $J_0$}{
				For $j$th feature, denote the pre-sorted subscripts of $\mr{node.X}$ as $(j'_1, j'_2, \cdots j'_{n'})$\; 
				
				\For{$i$ in $1:n'-1$}{ 
						split the current tree at $X_j=(X[j]_{j'_i}+X[j]_{j'_{i+1}})/2$\; 
						(Optional) compute the (unsigned) tree impurity decrease $\Delta I(T, \mathbb{P}_n)$ for the current split. If $j$th feature has already been split, let $\Delta I_0 = \min\{\Delta I_0, \Delta I(T, \mathbb{P}_n)\}$; Otherwise, reject the current split if $\Delta I(T,\mathbb{P}_n) < \Delta I_0 + \lambda_n$\;
						Compute the risk $\tilde{I}(T,\mathbb{P}_n) + \lambda_n r(T)$ for all $4$ ways of class label assignment after this split. Denote the smallest risk of the four trees as $R_{i,j}$ and the corresponding class labels for left and right child as $\mr{lab}^l_{i,j},\; \mr{lab}^r_{i,j}$\;
						}

				Let $i_0 =\arg\min_{i} R_{i,j}$, $\hat{x}_j = (X[j]_{j'_{i_0}}+X[j]_{j'_{i_0+1}})/2$, and $R_j = R_{i_0,j}$, $\mr{lab}^l_j=\mr{lab}^l_{i_0,j}$, $\mr{lab}^r_j=\mr{lab}^r_{i_0,j}$\;  
				}
			
 		Let $j_0 =\arg\min_{j} R_{j}$, $\hat{x} = \hat{x}_{j_0}$, and $R = R_{j_0}$, $\mr{lab}^l=\mr{lab}^l_{j_0}$, $\mr{lab}^r=\mr{lab}^r_{j_0}$ \;  
		\eIf{$R<R'$}{
				Accept the split. Let $\mr{node.left.X} =\{X\in \mr{node}: X[j]\le \hat{x}\}$, and $\mr{node.right.X} = \{X\in \mr{node}: X[j] > \hat{x}\}$. Assign $\mr{node.Y}$ to $\mr{node.left.Y}$ and $\mr{node.right.Y}$ according to the assignment of $\mr{node.X}$. Assign class labels to two child nodes as: $\mr{node.left.lab} = \mr{lab}^l, \; \mr{node.right.lab}=\mr{lab}^r$\;
				Update $R'$: $R':=R$\;
				Enqueue $\mr{node.left}$, $\mr{node.right}$ to the end of $\mr{node\_queue}$\;
				}{
				Reject the split\;
				}
}

\caption{Detailed Steps of SVR-Tree}
\label{algo:3}
\end{algorithm}

\subsection{Computational Complexity}
Recall the number of training samples is $n$ and the number of features is $d$. Denote the depth of the estimated tree as $h$ and let the maximal number of leaf nodes be $\bar{a}_n = O(\sqrt{n})$.\footnote{In our experiments, we set $\bar{a}_n = 2\sqrt{n}$.} The storage complexity is trivial and is the same as usual decision trees, i.e., $O(dn)$. We analyze the time complexity of Algorithm 1 in this section. We first introduce the approach to compute surface-to-volume ratio, then discuss the time complexity of building the tree.

\paragraph{Computing SVR}
We first consider computation of surface and volume for a single hyperrectangle in $\mathbb{R}^d$. The computation of volume involves multiplication of $d$ side lengths and takes $O(d)$ time. The computation of surface  consists of two steps. The first step is the computation of $d-1$ dimensional volumes for all $d$ unique ``faces'', each can be done by dividing the volume by a side length. Computing for all ``faces'' takes $O(d)$ time. The second step is adding up $d-1$ dimensional volumes of all $d$ unique ``faces'' and multiplying by $2$, which takes $O(d)$ time. Thus the computation of surface also takes $O(d)$ time in total. 

Now consider computing SVR for a classification tree. Suppose in some intermediate state of building the SVR tree, the current tree has $m$ leaf nodes $R_1, R_2, \cdots R_m$. We already know the surface and volume of the current tree. Now we need to perform a split at node $R_1$, which has $n'$ samples. Suppose we split at $X[j] = x_j$ and obtain two child nodes. The volume of the tree after this split can be computed by adding the volume of the minority class child node(s), which takes $O(d)$ time. The major concern lies in the computation of the surface area. If both child nodes are in the majority class or the minority class, the surface either does not change or changes by the overlapping surface between $R_1$ and some $R_j$s. It takes $O(d)$ time to compute the surface of $R_1$, and $O(md)$ time to compute all the overlapping surfaces between $R_1$ and $R_k, 2\le k\le m$. Therefore, the time complexity to compute the surface of the tree after splitting is $O(md)$.

If one child node belongs to the minority class and the other belongs to the majority class, the problem becomes more complicated. Let $S_{01,j}(x_j)$ be the surface of the split tree if $R_1$ is split at $X[j]=x_j$ with the left child labeled as $0$ and the right child labeled as $1$. $S_{10,j}(x_j)$ is similarly defined. Both $S_{01,j}(x_j)$ and $S_{10,j}(x_j)$ are piecewise linear functions whose change points can only exist at borders of $R_k, 2\le k\le m$. Therefore, we first compute the analytical forms of $S_{01,j}(x_j)$ and $S_{10,j}(x_j)$ for all $1\le j\le d$. This requires us to find all the borders of $R_k, 2\le k\le m$, to compute all the overlapping surface between $R_1$ and $R_k, 2\le k\le m$ and to compute the surface of $R_1$ itself. Thus the process of computing the analytical forms of $S_{01}(x_j)$ and $S_{10}(x_j)$ for all $1\le j\le m$ takes $O(md)$ time. The cost of evaluating $S_{01,j}(x_j)$ and $S_{10,j}(x_j)$ at a specific value $x_j$ can take as much as $O(m)$ time; but if  the samples are pre-sorted for all features, it only takes $O(dm+dn')$ time to evaluate $S_{01}(x_j)$ and $S_{10}(x_j)$, $\forall 1\le j\le d$, at all the possible split locations of $R_1$. Therefore, it takes $O(md+n'd)$ time to compute surface area for all the possible split locations and class label assignments. Similarly, the costs of computing volume for all possible split locations is $O(n'd)$. Thus for all the possible split locations and class label assignments at $R_1$, the total cost of computing SVR is $O(md+n'd)$.

\paragraph{Time Complexity of Algorithm 1}
Before working on any splits, we first need to sort the whole dataset for all features,  taking $O(dn\log n)$ time. For a node with $n'$ samples and $m$ leaf nodes in the current tree, there are $dn'$ possible split locations. It takes $O(dn')$ time to compute signed impurity for all possible class label assignments and all split locations; $O(md+n'd)$ time to compute surface-to-volume ratio for all possible class label assignments and all split locations; $O(dn')$ time to find the best split when all the signed impurities and SVR are already computed. The overall time complexity of finding the best split at this node is $O(md+n'd)$. Let $a_n$ be the number of leaf nodes of the tree output by Algorithm 1, with $a_n\le \bar{a}_n = O(\sqrt{n})$. The time complexity is 
$$O(dn\log n) + \sum_{m=1}^{ a_n} O(md+n'd) \le O(dn\log n + nhd),$$
where we use the fact that $\sum_{m=1}^{ a_n} n' \le nh$. Since the algorithm generates a binary tree, the average depth of $h$ is $O(\log n)$ and the worst depth of $h$ is $O(n)$. The overall average time complexity is $O(dn\log n)$, which is the typical average time complexity for many classification tree algorithms including CART and C4.5 \citep{sani2018computational}.

\section{Supplementary Material for Numerical Studies} \label{sec:datasets}

\subsection{Details of Datasets}

We provide the details of 12 datasets used in the numerical studies in Section 4 of the main text. 

\paragraph{Pima dataset}
\url{https://sci2s.ugr.es/keel/dataset.php?cod=21} consists of 768 samples and
8 features. The objective is to determine whether each person has diabetes. The positive class has 268 samples and is selected as the minority class.

\paragraph{Titanic dataset}
\url{https://sci2s.ugr.es/keel/dataset.php?cod=189} consists of 2201 samples and 3 features. The objective is to classify whether a person on board the Titanic can survive. The features are social class (first class, second class, third class, crewmember), age (adult or child) and sex. There are often multiple samples with the same features. The survive class has 711 samples and is chosen as the minority class.

\paragraph{Phoneme dataset}
\url{https://sci2s.ugr.es/keel/dataset.php?cod=105} consists of 5404 samples and 5 features. The objective is to classify to distinguish between nasal (class 0) and oral sounds (class 1). Class 1 has 1586 samples and is considered the minority class.

\paragraph{Vehicle dataset}\citep{siebert1987vehicle} 
\url{https://archive.ics.uci.edu/ml/datasets/Statlog+(Vehicle+Silhouettes)}
consists of 846 samples and 18 features. The aim is to classify a silhouette into one of four types of vehicles. As in \cite{he2008adasyn}, we choose class ``Van'' as the minority class and the other three types of vehicles as the majority class, resulting in 199 minority class samples. 

\paragraph{Ecoli}
\url{https://sci2s.ugr.es/keel/dataset.php?cod=61} consists of 336 samples and 7 features. The third feature is removed since all but one samples are of the same value, providing little help in cross-validation prediction, resulting in 6 features. The objective of the dataset is to predict the location site of the protein. We let ``pp'' be the minority class and combine all other classes to form majority class.

\paragraph{Segment dataset}
\url{https://sci2s.ugr.es/keel/dataset.php?cod=148} consists of 2308 samples and 18 features. It is an imbalanced version of the original segment dataset (\url{http://archive.ics.uci.edu/ml/datasets/Image+Segmentation}) preprocessed by KEEL repository. There are only two classes and the positive class, having 329 samples, is selected as minority class.

\paragraph{Wine dataset}\citep{cortez2009modeling} 
\url{https://archive.ics.uci.edu/ml/datasets/Wine+Quality}
collects information on wine quality. We focus on the red wine subset, which has 1599 samples and 11 features. We let the minority class be samples 
having quality $\ge 7$, while the majority class has 
quality $\le 6$.  This generates 217 minority class samples. 

\paragraph{Page dataset}
\url{https://sci2s.ugr.es/keel/dataset.php?cod=147} consists of 5472 samples and 10 features. The objective is to classify the tpye of a page block. The class ``text'' with 559 samples is selected as the minority class and all other classes are combined to form the majority class.

\paragraph{Satimage dataset}
\url{https://archive.ics.uci.edu/ml/datasets/Statlog+(Landsat+Satellite)}
consists of a training set and a test set. We have 6435 samples and 36 features. As in \cite{chawla2002smote}, we choose class `4'' as the minority class and collapsed all other classes into a single majority class, resulting in 626 minority class samples. 

\paragraph{Glass dataset} \url{https://sci2s.ugr.es/keel/dataset.php?cod=121} consists of 214 samples and 9 features. The objective is to classify samples into one of seven types of glass. We choose class ``2" as the minority class and all other classes as the majority class, yielding 17 minority class samples.   

\paragraph{Abalone dataset}
\url{https://archive.ics.uci.edu/ml/datasets/Abalone}
aims to predict the age of abalone by physical measurements. As in \cite{he2008adasyn}, we let class ``18'' be the minority class and class ``9'' be the majority class. This yields 731 samples in total, among which 42 samples belong to the minority class. We also remove the discrete feature ``sex'', which gives us 9 features. 

\paragraph{Yeast dataset}
\url{https://sci2s.ugr.es/keel/dataset.php?cod=133} consists of 1484 samples and 8 features. The objective is to determine the location site of each cell. The class ``ME2'' is selected as minority class and all other classes are combined to form the majority class, resulting 51 minority class samples.

\subsection{Details of Experimental Results}
The experimental results on all 12 datasets are displayed below in Table \ref{results}.

{
\centering
\scriptsize
\renewcommand{\arraystretch}{1.5}
% \begin{table}
\begin{longtable}[t]{c|c|ccccc}
\caption{Results of numerical studies on real world data sets.\label{results}} \\
\hline
\endfirsthead
\hline
\endhead
\hline
\endfoot
\hline
\endlastfoot
Data set & Method & Accuracy & Precision & TPR & F-measure & G-mean  \\ \hline
\multirow{7}{*}{Pima} 
				& SVR & 0.6922(0.0211) & 0.5451(0.0247) & 0.7332(0.0223) & 0.6247(0.0153) & 0.7004(0.0163) \\*
				& SVR-Select & 0.6809(0.0158) & 0.5317(0.0176) & \textbf{0.7349}(0.0233) & 0.6166(0.0120) & 0.6918(0.0124) \\*
				& Duplicate & 0.7348(0.0168) & 0.6257(0.0318) & 0.6108(0.0615) & 0.6152(0.0299) & 0.6977(0.0260) \\* 
				& SMOTE & 0.7335(0.0177) & 0.6348(0.0376) & 0.5724(0.0672) & 0.5981(0.0342) & 0.6826(0.0295) \\*
				& BSMOTE & 0.7328(0.0177) & 0.6266(0.0346) & 0.5948(0.0619) & 0.6072(0.0289) & 0.6907(0.0250) \\*
				& ADASYN & \textbf{0.7352}(0.0170) & 0.6283(0.0276) & 0.5972(0.0687) & 0.6096(0.0387) & 0.6931(0.0326)  \\* 
                & HDDT & 0.7534(0.0049) & \textbf{0.6631}(0.0088) & 0.5963(0.0073) & \textbf{0.6297}(0.0069) & \textbf{0.7067}(0.0053) \\ \hline
\multirow{7}{*}{Titanic} 
				& SVR & 0.7671(0.0047) & 0.6787(0.0172) & 0.5321(0.0095) & \textbf{0.5963}(0.0027) & 0.6840(0.0026) \\*
				& SVR-Select & 0.7660(0.0060) & 0.6753(0.0206) & \textbf{0.5340}(0.0126) & 0.5960(0.0032) & \textbf{0.6841}(0.0033)  \\*
				& Duplicate & \textbf{0.7855}(0.0036) & \textbf{0.8794}(0.0485) & 0.3930(0.0205) & 0.5419(0.0102) & 0.6181(0.0116)  \\* 
				& SMOTE  & \textbf{0.7855}(0.0036) & \textbf{0.8794}(0.0485) & 0.3930(0.0205) & 0.5419(0.0102) & 0.6181(0.0116) \footnote{The fact that SMOTE yields exactly the same result as duplicated sampling is not a coincidence, since there are lots of samples with similar or same feature values and the collection of synthetic samples of SMOTE is the same as that of duplicated sampling. Please refer to the data description in the supplementary material for details.} \\* 
				& BSMOTE & 0.7791(0.0047) & 0.8192(0.0425) & 0.4095(0.0239) & 0.5447(0.0152) & 0.6251(0.0148)  \\*
				& ADASYN & 0.7755(0.0023) &  0.7657(0.0335) & 0.4444(0.0373) & 0.5603(0.0221) & 0.6433(0.0217) \\* 
                & HDDT & 0.6768(0.0000) & 0.0000(0.0000) & 0.000(0.0000) & 0.0000(0.0000) & 0.0000(0.0000)   \\ \hline
\multirow{7}{*}{Phoneme} 
				& SVR & 0.8350(0.0060) & 0.6749(0.0113) & \textbf{0.8458}(0.0154) & 0.7506(0.0076) & 0.8380(0.0062) \\*
				& SVR-Select & 0.8377(0.0040) & 0.6807(0.0064) & 0.8420(0.0098) & 0.7528(0.0060) & \textbf{0.8389}(0.0049)  \\*
				& Duplicate & 0.8554(0.0038) & 0.7567(0.0110) & 0.7482(0.0154) & 0.7522(0.0068) & 0.8205(0.0065) \\* 
				& SMOTE & 0.8598(0.0054) & \textbf{0.7637}(0.0119) & 0.7567(0.0121) & 0.7601(0.0089) & 0.8264(0.0068) \\*
				& BSMOTE & 0.8568(0.0060) & 0.7580(0.0143) & 0.7527(0.0127) & 0.7552(0.0093) & 0.8230(0.0070) \\*
				& ADASYN & \textbf{0.8607}(0.0050) & 0.7661(0.0116) & 0.7565(0.0115) & \textbf{0.7612}(0.0082) & 0.8269(0.0063)  \\* 
                & HDDT & 0.7065(0.0000) & 0.0000(0.0000) & 0.000(0.0000) & 0.0000(0.0000) & 0.0000(0.0000) \\ \hline               
\multirow{7}{*}{Vehicle} & SVR & 0.9305(0.0068) & 0.8459(0.0167) & 0.8621(0.0297) & 0.8535(0.0155) & \textbf{0.9055}(0.0145) \\*
				& SVR-Select & 0.9248(0.0081) &  0.8253(0.0275) & 0.8648(0.0199) & 0.8442(0.0150) & 0.9031(0.0099) \\*
				& Duplicate & 0.9287(0.0096) & 0.8465(0.0262) & 0.8520(0.0231) & 0.8490(0.0196) & 0.9006(0.0136) \\*
				& SMOTE & \textbf{0.9314}(0.0082) & 0.8543(0.0223) & 0.8548(0.0289) & \textbf{0.8541}(0.0181) & 0.9033(0.0150) \\*
				& BSMOTE & 0.9285(0.0085) & 0.8471(0.0243) & 0.8505(0.0241) & 0.8485(0.0179) & 0.9000(0.0132)  \\*
				& ADASYN & 0.9310(0.0086) & \textbf{0.8580}(0.0222) & 0.8477(0.0292) & 0.8524(0.0190) & 0.9004(0.0156) \\*
				& HDDT & 0.6583(0.0030) & 0.3968(0.0029) & \textbf{0.8698}(0.0091) & 0.5450(0.0043) & 0.7184(0.0040) \\ \hline
\multirow{7}{*}{Ecoli} 
				& SVR & 0.9193(0.0123) & 0.7274(0.0530) & 0.7779(0.0436) & \textbf{0.7499}(0.0319) & 0.8570(0.0223) \\*
				& SVR-Select & 0.9193(0.0120) & 0.7297(0.0468) & 0.7683(0.0481) & 0.7471(0.0359) & 0.8525(0.0268) \\*
				& Duplicate & 0.9224(0.0111) & 0.7857(0.0555) & 0.6952(0.0612) & 0.7349(0.0389) & 0.8178(0.0343) \\*
				& SMOTE & \textbf{0.9233}(0.0122) & \textbf{0.7974}(0.0624) & 0.6856(0.0568) & 0.7348(0.04161) & 0.8134(0.0325)  \\*
				& BSMOTE & 0.9207(0.0137) & 0.7916(0.0731) & 0.6769(0.0507) & 0.7266(0.0393) & 0.8078(0.0286) \\*
				& ADASYN & 0.9128(0.0180) & 0.7281(0.0739) & 0.7183(0.0499) & 0.7202(0.0448) & 0.8248(0.0285)  \\* 
                & HDDT & 0.8907(0.0079) & 0.5962(0.0198) & \textbf{0.9231}(0.0086) & 0.7243(0.0151) & \textbf{0.9037}(0.0063) \\ \hline
\multirow{7}{*}{Segment} & SVR & 0.9922(0.0017) & 0.9696(0.0089) & 0.9763(0.0061) & 0.9729(0.0060) & 0.9855(0.0033) \\*
				& SVR-Select & \textbf{0.9932}(0.0013) & \textbf{0.9753}(0.0080) & 0.9769(0.0053) & \textbf{0.9761}(0.0047) & \textbf{0.9863}(0.0027)  \\*
				& Duplicate & 0.9913(0.0012) & 0.9707(0.0069) & 0.9685(0.0073) & 0.9696(0.0044) & 0.9817(0.0036) \\* 
				& SMOTE & 0.9916(0.0010) & 0.9722(0.0061) & 0.9688(0.0081) & 0.9705(0.0035) & 0.9820(0.0038) \\*
				& BSMOTE & 0.9915(0.0011) & 0.9726(0.0082) & 0.9676(0.0059) & 0.9701(0.0037) & 0.9814(0.0027)  \\*
				& ADASYN & 0.9919(0.0011) & 0.9713(0.0063) & 0.9716(0.0064) & 0.9714(0.0040) & 0.9833(0.0032) \\*
				& HDDT & 0.8330(0.0021) & 0.4599(0.0031) & \textbf{0.9825}(0.0023) & 0.6266(0.0031) & 0.8910(0.0017) \\ \hline
\multirow{7}{*}{Wine} & SVR & 0.8391(0.0143) & 0.4398(0.0306) & 0.6574(0.0340) & 0.5263(0.0278) & 0.7549(0.0205) \\*
				& SVR-Select & 0.8395(0.0098) & 0.4405(0.0207) & \textbf{0.6641}(0.0337) & \textbf{0.5290}(0.0184) & \textbf{0.7584}(0.0173) \\*
				& Duplicate & \textbf{0.8741}(0.0097) & \textbf{0.5469}(0.0429) & 0.4703(0.0508) & 0.5027(0.0311) & 0.6628(0.0335)\\*
				& SMOTE & 0.8719(0.0086) & 0.5365(0.0425) & 0.4650(0.0843) & 0.4923(0.0527) & 0.6561(0.0588) \\*
				& BSMOTE & 0.8708(0.0075) & 0.5317(0.0333) & 0.4482(0.0623) & 0.4829(0.0383) & 0.6461(0.0430) \\
				& ADASYN & 0.8717(0.0060) & 0.5324(0.0263) & 0.4836(0.0727) & 0.5030(0.0411) & 0.6692(0.0482) \\*
				& HDDT & 0.8378(0.0038) & 0.4349(0.0087) & 0.6516(0.0119) & 0.5216(0.0092) & 0.7516(0.0073) \\
				\hline
\multirow{7}{*}{Page} 
				& SVR & 0.9656(0.0015) & 0.8248(0.0098) & 0.8429(0.0117) & 0.8337(0.0076) & 0.9087(0.0062) \\*
				& SVR-Select & 0.9647(0.0016) & 0.8252(0.0124) & 0.8311(0.0140) & 0.8280(0.0077) & 0.9024(0.0073)  \\*
				& Duplicate & \textbf{0.9686}(0.0013) & 0.8465(0.0081) & \textbf{0.8463}(0.0135) & \textbf{0.8463}(0.0072) & \textbf{0.9119}(0.0071) \\* 
				& SMOTE & 0.9683(0.0021) & \textbf{0.8468}(0.0133) & 0.8426(0.0144) & 0.8446(0.0101) & 0.9099(0.0078) \\*
				& BSMOTE & 0.9682(0.0015) & 0.8444(0.0103) & 0.8448(0.0127) & 0.8445(0.0077) & 0.9109 (0.0067) \\*
				& ADASYN & 0.9677(0.0021) & 0.8436(0.0131) & 0.8403(0.0148) & 0.8418(0.0103) & 0.9084(0.0080)  \\* 
                & HDDT & 0.9024(0.0030) & 0.5479(0.0338) & 0.2720(0.0260) & 0.3622(0.0215) & 0.5141(0.0233)  \\ \hline
\multirow{7}{*}{Satimage} 
				& SVR & 0.9009(0.0037) & 0.4935(0.0130) & 0.7050(0.0188) & \textbf{0.5805}(0.0138) & 0.8061(0.0110) \\*
				& SVR-Select & 0.8982(0.0032) & 0.4847(0.0111) & \textbf{0.7193}(0.0231) & 0.5789(0.0109) & \textbf{0.8123}(0.0120) \\*
				& Duplicate & 0.9191(0.0040) & 0.5982(0.0295) & 0.5236(0.0292) & 0.5572(0.0135) & 0.7093(0.0177) \\* 
				& SMOTE & 0.9204(0.0032) & 0.6127(0.0298) & 0.5057(0.0320) & 0.5525(0.0129) & 0.6982(0.0203) \\*
				& BSMOTE & \textbf{0.9218}(0.0037) & \textbf{0.6340}(0.0403) & 0.4780(0.0321) & 0.5430(0.0141) & 0.6803(0.0210) \\*
				& ADASYN & 0.9218(0.0034) & 0.6273(0.0367) & 0.4958(0.0339) & 0.5519(0.0135) & 0.6922(0.0218) \\* 	
                & HDDT & 9027(0.0000) & 0.0000(0.0000) & 0.0000(0.0000) & 0.0000(0.0000) & 0.0000(0.0000) \\ \hline
\multirow{7}{*}{Glass} & SVR & 0.8859(0.0198) & \textbf{0.2932}(0.1064) & 0.2853(0.1006) & \textbf{0.2851}(0.0965) & 0.5092(0.0926)  \\*
				& SVR-Select & 0.8592(0.0437) & 0.2232(0.0905) &  0.2559(0.0989) & 0.2271(0.0764) & 0.4716(0.0910)  \\*
				& Duplicate & 0.9052(0.0142) & 0.2068(0.1855) & 0.0853(0.0755) & 0.1118(0.0893) & 0.2329(0.1672) \\* 
				& SMOTE & 0.9049(0.0135) & 0.2157(0.1474) & 0.1000(0.0813) & 0.1290(0.0924) & 0.2676(0.1573) \\*
				& BSMOTE & 0.8934(0.0182) & 0.1762(0.1282) & 0.1176(0.0931) & 0.1341(0.0972) & 0.2804(0.1810) \\*
				& ADASYN & \textbf{0.9068}(0.0142) & 0.2126(0.1895) & 0.1059(0.0959) & 0.1348(0.1155) & 0.2521(0.1964) \\* 
				& HDDT & 0.4803(0.0128) & 0.1124(0.0123) & \textbf{0.8000}(0.0570) & 0.1971(0.0115) & \textbf{0.6011}(0.0195) \\ \hline
\multirow{7}{*}{Abalone} & SVR & 0.9062(0.0116) & 0.2770(0.0471) & 0.3798(0.0696) & \textbf{0.3168}(0.0431) & 0.5939(0.0543)  \\*
				& SVR-Select & 0.9060(0.0243) & 0.2697(0.0444) & 0.3381(0.0777) & 0.2941(0.0369) & 0.5599(0.0522) \\*
				& Duplicate & 0.9308(0.0081) & 0.3251(0.1231) & 0.1905( 0.0768) & 0.2330(0.0830) & 0.4138(0.1192) \\* 
				& SMOTE & 0.9313(0.0071) & 0.3286(0.1037) & 0.2190(0.0992) & 0.2556(0.0972) &  0.4407(0.1359) \\*
				& BSMOTE & \textbf{0.9359}(0.0073) & \textbf{0.4046}(0.1166) & 0.1964(0.0752) & 0.2546(0.0783) & 0.4297(0.0878) \\*
				& ADASYN & 0.9329(0.0085) & 0.3920( 0.1074) & 0.2369(0.0704) & 0.2841(0.0621) &  0.4742(0.0756)  \\* 
				& HDDT & 0.8413(0.0086) & 0.1795(0.0123) & \textbf{0.4917}(0.0274) & 0.2629(0.0160) & \textbf{0.6510}(0.0187) \\ \hline
\multirow{7}{*}{Yeast} 
				& SVR & 0.9368(0.0100) & 0.2641(0.0410) & 0.4510(0.0888) & \textbf{0.3290}(0.0416) & 0.6524(0.0619)  \\*
				& SVR-Select & 0.9287(0.0122) & 0.2360(0.0397) & 0.4559(0.0715) & 0.3077(0.0391) & 0.6544(0.0496) \\*
				& Duplicate & \textbf{0.9641}(0.0037) & \textbf{0.3907}(0.2084) & 0.1196(0.0784) & 0.1727(0.0981) & 0.3077(0.1544) \\* 
				& SMOTE & 0.9594(0.0045) & 0.3620(0.0859) & 0.2098(0.0830) & 0.2554(0.0752) & 0.4458(0.0887) \\*
				& BSMOTE & 0.9602(0.0043) & 0.3664(0.1012) & 0.2020(0.0884) & 0.2499(0.0937) & 0.4333(0.1065)  \\*
				& ADASYN & 0.9605(0.0045) & 0.3709(0.1246) & 0.2078(0.1099) & 0.2491(0.0986) & 0.4288(0.1436) \\* 	
                & HDDT & 0.7604(0.0200) & 0.1049(0.0101) & \textbf{0.7853}(0.0287) & 0.1849(0.0162) & \textbf{0.7722}(0.0211) \\ \hline
\multirow{7}{*}{\makecell{Average Ranking}}
				& SVR & 4.92 & 5.25 & \textbf{2.25} & \textbf{2.08} & \textbf{2.17} \\*
				& SVR-Select & 5.58 & 5.00 & \textbf{2.25} & 2.75 & 2.50 \\*
				& Duplicate & 2.79 & 3.29 & 5.13 & 4.88 & 4.96  \\* 
				& SMOTE & \textbf{2.54} & \textbf{2.29} & 5.04 & 4.21 & 4.71 \\*
				& BSMOTE & 3.33 & 2.92 & 5.58 & 5.00 & 5.42 \\*
				& ADASYN & 2.75 & 2.92 & 4.58 & 4.00 & 4.42 \\* 
				& HDDT & 6.08 & 6.33 & 3.17 & 5.08 & 3.83 \\ \hline	
\end{longtable}
% \end{table}
}

\bibliographystyle{chicago}
\bibliography{IBbib}